\documentclass[journal]{IEEEtran}
\usepackage[utf8]{inputenc}

\usepackage{amsmath,amssymb, amsthm, mathtools, braket}
\usepackage{braket}
\usepackage{graphicx}
\DeclareGraphicsExtensions{.eps,.pdf,.jpg,.png}
\usepackage{caption,subcaption} % depreciated, but still used in  IEEEtran

\usepackage{enumitem}
\usepackage{algorithm, algorithmic}
\usepackage{url,doi}
\usepackage[]{hyperref}
\hypersetup{%
linkcolor=blue,anchorcolor=blue,%
citecolor=blue,filecolor=blue,%
menucolor=blue,%
urlcolor=blue,%
colorlinks=true}

\DeclareMathOperator*{\argmin}{arg\,min}

\newcounter{MYtempeqncnt}

\usepackage{bm}

\long\def\comment#1{}

\newfont{\bbb}{msbm10 scaled 700}

\newcommand{\bv}{{\bf b}}

\newcommand{\ev}{{\bf e}}
\newcommand{\fv}{{\bf f}}

\newcommand{\hv}{{\bf h}}

\newcommand{\nv}{{\bf n}}

\newcommand{\tv}{{\bf t}}
\newcommand{\uv}{{\bf u}}

\newcommand{\vv}{{\bf v}}
\newcommand{\xv}{{\bf x}}
\newcommand{\yv}{{\bf y}}
\newcommand{\zv}{{\bf z}}
\newcommand{\zerov}{{\bf 0}}

\newcommand{\Am}{{\bf A}}
\newcommand{\Bm}{{\bf B}}

\newcommand{\Fm}{{\bf F}}
\newcommand{\Gm}{{\bf G}}
\newcommand{\Hm}{{\bf H}}

\newcommand{\Km}{{\bf K}}
\newcommand{\Lm}{{\bf L}}
\newcommand{\Mm}{{\bf M}}

\newcommand{\Pm}{{\bf P}}

\newcommand{\Rm}{{\bf R}}
\newcommand{\Sm}{{\bf S}}
\newcommand{\Tm}{{\bf T}}
\newcommand{\Um}{{\bf U}}

\newcommand{\Ec}{{\cal E}}
\newcommand{\Fc}{{\cal F}}
\newcommand{\Gc}{{\cal G}}
\newcommand{\Hc}{{\cal H}}

\newcommand{\Kc}{{\cal K}}

\newcommand{\Mc}{{\cal M}}

\newcommand{\Sc}{{\cal S}}
\newcommand{\Tc}{{\cal T}}

\newcommand{\Vc}{{\cal V}}

\newcommand{\Lcb}{{\bm {\mathcal L}}}

\newcommand{\Scc}{{\Sc^c}}

\newtheorem{theorem}{Theorem}
\newtheorem{proposition}{Proposition}

\newtheorem{definition}{Definition}

\begin{document}

\title{Guided Signal Reconstruction Theory
%Projected Least Squares Signal Reconstruction Orthogonal to Sampling Space
%A Framework for Sampling and Reconstruction of Signals in Arbitrary Hilbert Spaces
}
\author{Andrew Knyazev,~\IEEEmembership{Senior Member,~IEEE, Fellow, SIAM}, Akshay Gadde,~\IEEEmembership{Student Member,~IEEE,},\\
Hassan Mansour,~\IEEEmembership{Member,~IEEE,}, Dong Tian,~\IEEEmembership{Member,~IEEE}
\thanks{A. Knyazev, H. Mansour, and D. Tian are with the Mitsubishi Electric Research Laboratories (MERL), 201 Broadway, 8th Floor
Cambridge, MA 02139-1955, e-mail: \{\href{mailto:knyazev@merl.com}{Knyazev}, \href{mailto:mansour@merl.com}{Mansour}, \href{mailto:tian@merl.com}{Tian}\}@merl.com, WWW: http://www.merl.com/people/\{\href{http://www.merl.com/people/knyazev}{knyazev}, \href{http://www.merl.com/people/mansour}{mansour}, \href{http://www.merl.com/people/tian}{tian}\}.}
    \thanks{A. Gadde is with the University of Southern California (USC), Los Angeles, CA, and has been an intern at MERL, e-mail: \href{mailto:agadde@usc.edu}{agadde@usc.edu}.}
     \thanks{This work has been presented in parts at the 2015 IEEE Global Conference on Signal and Information Processing (GlobalSIP), Orlando, FL, 2015~\cite{7418335}.}
      \thanks{A preliminary version of the manuscript is posted at arXiv.}
    %\thanks{Manuscript compiled \today.}
    }

\markboth%{Transactions on Information Theory,~Vol.~?, No.~?, ?~2017}%
{Knyazev \MakeLowercase{\textit{et al.}}: Guided Signal Reconstruction Theory}{}    
    
\maketitle

\begin{abstract}
An axiomatic approach to signal reconstruction is formulated, involving 
a sample consistent set and a guiding set, describing desired reconstructions. 
New frame-less reconstruction methods are proposed, based on  
a novel concept of a reconstruction set, defined as a shortest pathway 
between the sample consistent set and the guiding~set. 
Existence and uniqueness of the reconstruction set are investigated in a Hilbert space, where the guiding set is a closed subspace and
the sample consistent set is a closed plane, formed by a sampling subspace.
Connections to earlier known consistent, generalized, and regularized reconstructions are clarified. New stability and reconstruction error bounds are derived, using the largest nontrivial angle between the sampling and guiding subspaces. Conjugate gradient iterative reconstruction algorithms are proposed and illustrated numerically for image magnification.  
\end{abstract}

\IEEEpeerreviewmaketitle

\emph{%Surgeon's law: 
Quotient space law:``When in doubt, cut it out!''}

\section{Introduction}

Signal reconstruction is a standard problem that arises naturally in signal processing and machine learning. A classical example is reconstruction of band-limited signals from their time-domain samples. Recently, reconstruction of signals on graphs from signal samples on a subset of nodes of the graph is gaining popularity (e.g.,\ \cite{Narang-GlobalSIP-13,WangCG15,7437485,7480396,7208894,7581102,7439829}) and finds applications in graph-based semi-supervised learning; see, e.g.,~\cite{Gadde-KDD-14}. In this context, the signals are considered to be band-limited with respect to eigenvalues of a graph Laplacian. 

A Hilbert space framework allows investigating signal reconstruction in a general and concise manner. To this end, we consider a problem of determining a reconstruction $\hat{\fv}\in\Hc$ of an unknown original signal $\fv\in\Hc$ from a measurement of $\fv$, where $\Hc$ is a Hilbert space equipped with a scalar product $\Braket{\cdot,\cdot}$ and a corresponding norm $\|\cdot\|$. The measurement of $\fv$ is defined as a result $\Sm\fv$ of an action of an orthogonal projector $\Sm$ onto a closed subspace $\Sc \subseteq \Hc$ called the \emph{sampling subspace}. 

The original signal $\fv$ is typically not known, only the sampled original signal  $\Sm\fv$ is available as an input to a reconstruction method.
Since sampling involves loss of information, we need some a priori assumptions on the original signal $\fv$ to be recovered. One such assumption may be that the signal $\fv$ belongs to a closed subspace $\Tc \subseteq  \Hc$ that can be thought of as a \emph{target reconstruction subspace}.
Alternatively, the signal $\fv$ may not lie strictly in $\Tc$, but may be well approximated by its projection on the subspace $\Tc \subseteq  \Hc$. We~prefer to call $\Tc$ a \emph{guiding reconstruction subspace}, since the reconstructed signal  $\hat{\fv}$ may not necessarily be restricted to $\Tc$. 
Another example of a prior structure is that the signal $\fv$ belongs to a compact subset of $\Hc$, determined by ``smoothness'' of $\fv$. 
In any case, the reconstruction that minimizes the reconstruction error $\|\hat{\fv} - \fv\|$ is naturally desired.

The guiding set can be determined using a model or other form of description of desirable reconstructed signal behavior, e.g., learned from training datasets. For signals with natural spectral properties, spectral transforms, e.g., Fourier, cosine, and wavelet transforms, can be used to transform signals into a spectral domain, where the guiding subspace can be chosen as corresponding to certain frequency ranges, e.g., assuming that the desired signal is band-limited. 

For signals without self-evident spectral properties, the~signals are embedded into a specially constructed structure, depending on the type of the signal, e.g., a graph, or a Riemannian manifold, wherein spectral properties are determined by an~``energy'' norm and its defining operator, e.g., graph Laplacian, or the Laplace-Beltrami operator, correspondingly~\cite{Kheradmand2014}. The energy norm can be constructed, using a given guiding signal, or from a database of a priori learned signals,   
depending on a signal similarity measure in the signal space, which can comprise, e.g., correlation, coherence, divergence, or metric, depending on the type of the signal and desired reconstruction properties. 
For example, in the graph-based signal processing, edge weights can be determined using distances between vertex-localized delta-function signals, spanning the given guiding signal. 

The~guiding subspace can then be chosen to approximate an invariant subspace of the energetic operator, corresponding to certain ranges in its spectrum, e.g., assuming that the desired signal is band-limited, having components primarily from the low part of the spectrum of the energetic operator. 

\subsection{Notation}
Let $\Sm$ and $\Tm$ be the orthoprojectors onto the closed subspaces $\Sc$ and $\Tc$, respectively. Let $\Sm^\perp = \mathbf{I} - \Sm$ and $\Tm^\perp = \mathbf{I}- \Tm$, where $\mathbf{I}$ is the identity operator, denote the orthoprojectors onto their orthogonal complements $\Sc^\perp$ and $\Tc^\perp$. Let $R(\Am)$ denote the range of operator $\Am$ and $N(\Am)$ its null space; e.g.,\ $\Sc = R(\Sm)$ and $\Sc^\perp = R(\Sm^\perp) = N(\Sm)$. 
$\Am^*$ denotes $\Am$ adjoint. 

We measure (sample) an element $\fv \in \Hc$ by its projection on $\Sc$, i.e. the observed sample is given by $\Sm\fv$,
and want to reconstruct $\fv$ from $\Sm\fv$.
The signal $\fv$ to be reconstructed can be split into two orthogonal components:
\begin{equation}
\fv = \Sm\fv + \xv, \text{ where } \Sm\fv \in \Sc \text{ and } \xv \in \Sc^\perp,
\label{eq:split}
\end{equation}
where $\Sm\fv$ is the observed sample of $\fv$ and $\xv$ contains the missing information to be reconstructed. 

\subsection{Prior work}
Two main kinds of sample consistent reconstructions are known: subspace-based constrained reconstructions using oblique projectors leading to $\hat\fv\in\Tc$, e.g.,~\cite{Unser-TSP-94, Eldar-JFA-03,Eldar-SPM-09}, and energy minimization-based reconstructions, e.g.,\ in~\cite{Eldar-SPM-09} and generalized abstract splines~\cite[Sec.~4]{Arias2014}. 
Practical reconstruction is usually performed using frames for $\Sc$ and $\Tc$. In this context, $\Sc$ is separable and comes, e.g., with an orthonormal countable frame $F$. Consequently, $\Tm F$ is also a frame for $\Tc$, having the frame operator $\Tm\Sm\Tm$ restricted to $\Tc$, assuming $\Sc^\perp\cap\Tc = \{\zerov\}$ and strict positivity of the \emph{minimal gap}~\cite[Sec.~IV-4]{kato} between $\Sc^\perp$ and $\Tc$, which makes the inverse of the frame operator bounded. A general approach we present in this paper is frame-less, dealing directly with the orthogonal projectors $\Sm$ and $\Tm$ onto the subspaces $\Sc$ and $\Tc$.

A set of all signals, having the same sample $\Sm\fv$, is a closed plane $\Sm\fv+\Sc^\perp$
that we call a ``consistent plane.''
But $\Sm\fv+\Sc^\perp$ and $\Tc$ generally do not intersect, in which case no reconstruction $\hat{\fv}$ can be constrained to 
both sets as required in~\cite{Unser-TSP-94, Eldar-JFA-03}. 
For a solution, which is in both $\Tc$ and $\Sm\fv + \Sc^\perp$, to exist for any $\fv$, we need $\Sc^\perp+\Tc = \Hc$. 
Additionally, for such a solution to be unique, we need $\Sc^\perp\cap\Tc = \{\zerov\}$. Otherwise, there can be multiple signals in $\Tc$ having the same samples. If both of these conditions are satisfied, then a unique sample consistent solution in $\Tc$ is given by $\Pm_{\Tc\perp\Sc}\fv$, where $\Pm_{\Tc\perp\Sc}$ is an oblique projector on $\Tc$ along $\Sc^\perp$. Non-uniqueness caused by $\Sc^\perp\cap\Tc \neq \{\zerov\}$ can be mathematically resolved by replacing $\Hc$ with a quotient space $\Hc/\{\Tc\cap \Sc^\perp\}$. After such a replacement, we have $\Sc^\perp\cap\Tc = \{\zerov\}$, which we assume for the rest of the section. Practically, one can choose a unique solution by imposing additional constraints~\cite{Hirabayashi-Unser-2007}. 

The assumption  $\Sc^\perp+\Tc = \Hc$ can be disadvantageous and very restrictive in applications. 
Even if $\Sm\fv+\Sc^\perp$ and $\Tc$ do intersect, finding their intersection numerically may be difficult, 
as the intersection may be very sensitive to their mutual position; see for example the ``generalized reconstruction'' scheme of~\ \cite{adcock2012generalized,ADCOCK2012357,Adcock:2013:BCR}. A cure proposed in \cite{adcock2012generalized} is oversampling, which leads to a smaller consistent plane $\Sm\fv+\Sc^\perp$ that no longer intersects with~$\Tc$ and where the \emph{strictly guided} reconstructed signal is defined as a point in $\Tc$ having the smallest distance to $\Sm\fv+\Sc^\perp$. 
The subspace $\Tc$ is treated literally as the target subspace, thus, 
enforcing the constraint $\hat{\fv}\in\Tc$ and relaxing $\hat{\fv}\in\Sm\fv+\Sc^\perp$ using least squares. 

The strictly guided generalized reconstruction methods from \cite{adcock2012generalized} and the minimax regret in~\cite{Eldar-TSP-06} may be sample inconsistent, since they place the reconstructed signal into the guiding subspace. 
In~contrast, \cite{bansal2005bandwidth} puts the reconstructed signal into the consistent plane, relaxing the property that $\hat{\fv}\in\Tc$ by minimizing instead the energy in $\Tc^\perp$. The reconstructed signal is defined as a point in the sample consistent plane $\Sm\fv+\Sc^\perp$ having the smallest distance to $\Tc.$
This approach is motivated by a realization that in practical applications, such as bandwidth expansion of narrowband audio signals, it may be difficult to explicitly find a frame of or even choose a trustworthy target reconstruction subspace~$\Tc$. Thus, the subspace $\Tc$ can be used as a guide, not as a true target, where we trust the sampling more than the guiding, as in \cite{bansal2005bandwidth}.%, the subspace $\Tc$ is obtained by learning from a database of relevant signals. 

Regularization-based methods, suggested in \cite{Narang-GlobalSIP-13,Kheradmand2014}, determine reconstruction by solving an unconstrained problem minimizing a weighted sum of a loss function and a regularization term using a regularization parameter. The regularization parameter needs to be chosen {\it a priori}---a common  difficulty of  regularization-based methods---and may greatly affect the reconstruction quality. The authors of \cite{Narang-GlobalSIP-13} assume existence and uniqueness of the intersection of the sample-consistent reconstruction plane $\Sm\fv+ \Sc^\perp$ with the  guiding reconstruction subspace $\Tc$ for any original signal $\fv$, just as in~\cite{Hirabayashi-Unser-2007,Corach-Giribet-2011}. 

\subsection{Main contributions}
Let us assume that the  guiding reconstruction subspace $\Tc$ is available in some form, e.g.,\ implicitly via an action of the corresponding (possibly approximate) orthogonal projector $\Tm$. We formulate a least squares approach that allows implicit, frame-less, and approximate 
descriptions of $\Tc^\perp$, e.g., in a form of a filter function, approximately suppressing $\Tc$ components of a signal. Additionally, the least squares approach allows and can benefit from oversampling, as in the generalized reconstruction of \cite{adcock2012generalized},
making our reconstruction algorithms more stable, compared to classical constrained reconstructions using oblique projectors in~\cite{Unser-TSP-94, Eldar-JFA-03}. 

We describe a unified view of consistent, generalized and regularization based reconstruction methods. 
A novel concept of a reconstruction set is introduced, We explain how it relates to the regularization-based methods of \cite{Narang-GlobalSIP-13,Kheradmand2014}.
Conditions of existence and uniqueness of the reconstructed signal are obtained, using~\cite{Knyazev-CMM-07} and beyond. Stability and reconstruction error bounds are derived that improve those following from the bounds in~\cite{Knyazev-CMM-07}. 
We suggest a numerically efficient iterative reconstruction algorithm, based on a conjugate gradient method, which approximates our reconstruction and only needs actions of orthoprojectors onto the subspaces $\Sc$ and~$\Tc$. %, or similar signal filters. 
We also derive convergence rate bounds of iterative algorithms and reconstruction error bounds, depending on angles between the subspaces $\Sc$ and $\Tc$.

%Last, but not least, we introduce a novel concept of a reconstruction set, in the next section, and explain how it relates to the regularization-based methods of \cite{Narang-GlobalSIP-13,Kheradmand2014}. 

\section{Reconstruction Set} 
A case, where both procedures, sampling and guiding, can be equally trusted, but the guiding set contains no sample consistent signals, reminds us of Buridan's donkey that is equally hungry and thirsty and that is placed precisely midway between a stack of hay (the guiding set) and a pail of water (the sample consistent set), so it will have to die of both hunger and thirst, since it cannot make any rational decision to choose one over the other. %Without loss of generality, one can assume the donkey being spherical. 

To save the hypothetical donkey, we define a set of reconstructions given by convex combinations of strictly guided and consistent reconstructions. As stated before, the guiding set (or subspace) may contain no sample consistent solutions. When the samples are noisy, the true signal does not lie in the sample consistent plane. The true signal may not be entirely contained in the guiding subspace either. In such a case, it is unclear, which reconstruction, consistent~\cite{bansal2005bandwidth} or strictly guided (generalized)~\cite{adcock2012generalized}, is better to choose. 

This situation is illustrated in Fig.~\ref{fig:1} by a simple motivating geometric example, where $\dim\Hc=3, \dim\Sc=2,$ and $\dim\Tc=1.$ 
Here, the set of all signals, having the same 
sample $\Sm\fv$ is evidently a line $\Sm\fv+\Sc^\perp$.
The lines $\Sm\fv+\Sc^\perp$ and $\Tc$ generally do not intersect, so no reconstruction $\hat{\fv}$ can be constrained to 
both lines as required in~\cite{Unser-TSP-94, Eldar-JFA-03}; see Fig.~\ref{fig:1}. 
        \begin{figure}
                \includegraphics[width=0.85\linewidth]{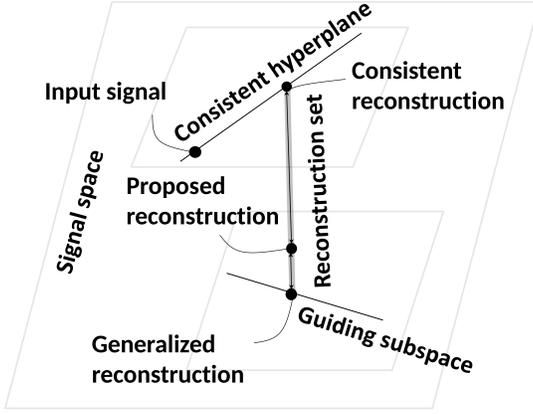}
                \caption{An example of the reconstruction set in 3D}\label{fig:1}
        \end{figure}

We observe in Fig. \ref{fig:1} that, on the one hand, the consistent reconstruction of \cite{bansal2005bandwidth} can be viewed as a minimizer of a distance \emph{from} an element of the consistent plane $\Sm\fv+\Sc^\perp$ \emph{to} the guiding subspace  $\Tc$, while, on the other hand, the generalized reconstruction of \cite{adcock2012generalized} is an element \emph{from} the guiding subspace  $\Tc$, minimizing the distance \emph{to} the consistent plane $\Sm\fv+\Sc^\perp$. Clearly, equalities hold 
\begin{align}
\inf_{\hat\fv\in\Sm\fv+\Sc^\perp}\quad\inf_{\hat{\tv}\in\Tc} \|\hat{\fv}-\hat{\tv}\|&=
\inf_{\hat{\fv}\in\Sm\fv+\Sc^\perp,\, \hat{\tv}\in\Tc} \|\hat{\fv}-\hat{\tv}\|\label{eq:recon_problem_geom01}\\
&=\inf_{\hat{\tv}\in\Tc}\quad\inf_{\hat{\fv}\in\Sm\fv+\Sc^\perp} \|\hat{\fv}-\hat{\tv}\|,
\label{eq:recon_problem_geom02}
\end{align}
where the minimizers $\hat{\fv}$ and $\hat{\tv}$ are called \emph{consistent} and \emph{generalized} reconstructions, 
respectively, 
giving us a hint to define  a reconstruction set, which is a shortest pathway set 
\emph{between} the consistent plane $\Sm\fv+\Sc^\perp$ and the guiding subspace~$\Tc$.
In Fig. \ref{fig:1}, the reconstruction set is a closed interval, with the end points being the consistent reconstruction $\hat{\fv}$ of \cite{bansal2005bandwidth}
and the generalized reconstruction $\hat{\tv}$ of \cite{Adcock:2013:BCR}. If it is unclear, which one of the procedures,  sampling or guiding, can be trusted more, any element of the reconstruction set becomes a valid candidate for reconstruction. 
%In terms of Buridan's donkey, an element of the reconstruction set represents a mix of hay and water.

Moreover, Fig. \ref{fig:1} and our discussion above suggest us to propose a general  definition of a reconstruction set as a shortest pathway 
between a given guiding set and a sample consistent set, defined as a set of signals sample-consistent with the original signal, see Fig. \ref{fig:2}.
        \begin{figure}
                \includegraphics[width=0.85\linewidth]{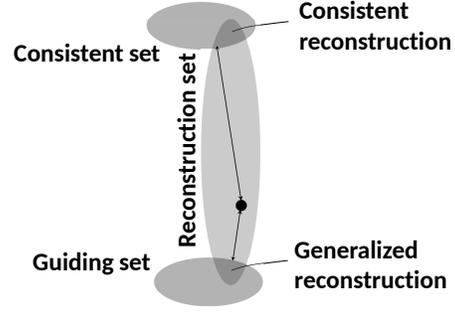}
                \caption{The reconstruction set in a metric space}\label{fig:2}
        \end{figure}
The shortest pathway between two sets can be formally defined as a convex set of elements, such that any element of the shortest pathway minimizes a sum of a distance between the element and the first set and a distance between the element and the second set. 

The consistent reconstruction is the intersection of the reconstruction set and the consistent set. The generalized (strictly guided) reconstruction is the intersection of the reconstruction set and the guiding set.
In this definition, one only needs a structure of a metric space with a distance, thus allowing nonlinear and even multi-valued sampling procedures and general guiding sets. 

For example, the interesting recent work of Adcock and Hansen~\cite{Adcock2016} combines generalized reconstruction with infinite dimensional compressed sensing in a natural framework of Banach spaces. We believe that our notion of the reconstruction set can be extended to such a framework, 
allowing one to find the reconstructed signal that is not strictly sparse, but is guided by a reconstruction subspace, identified by Adcock and Hansen's method in~\cite{Adcock2016}. 

In this paper, however, we limit ourselves to the traditional Hilbert space framework,  where the guiding set  is a closed subspace and
the sample consistent set is a closed plane. When the generalized reconstruction and the consistent reconstruction exist and are unique, the reconstruction set is simply their convex hull---a closed interval in this case, exactly as illustrated in Fig. \ref{fig:1} in the 3D space $\Hc$.  

Another possibility, not addressed here, is where our deterministic setup is augmented by a probabilistic approach, where signals are random. For example, either, or both, consistent and guiding sets may be determined using probability distributions. In this case, the reconstruction set is also determined by a probability distribution using a statistical distance between random variables or samples. 

Having to output the whole reconstruction set of multiple reconstructed signals may not be appropriate in applications, even where the parametrization using the end points of the reconstruction set is possible. To pick up a single reconstructed signal from the reconstruction set, one needs extra information, for example, a cost/quality function, e.g.,\ Buridan's donkey can choose a proper healthy mix of hay and water following a given dietary function. 
Then, one may output only the signals in the neighborhood of the reconstruction set that minimize the cost/quality function.  In Sec.~\ref{sec:regularization}, we show how to select the optimal solution, if the amount of noise is known, and relate the reconstruction set to regularization methods. 
Finally, if the cost/quality function should be trusted more than both the consistent and the guiding set, one may choose to minimize a weighted sum of the  cost/quality function and distances to the reconstruction and sampling sets. 
%Such a choice is, however, outside of the scope of the present work.

\section{Overview of reconstruction in a Hilbert space}\label{s:overview}

The intersection  $\Sc^\perp\cap\Tc$ consists of signals in the guiding reconstruction subspace $\Tc$ with zero samples, projections on~$\Sc$. Its important role in the reconstruction is stated in the following assumption.
 \begin{enumerate}[label=\bfseries (A\arabic*), topsep=0pt, itemsep=0pt]
  \setcounter{enumi}{-1}
  \item \label{AU} \emph{Reconstruction Uniqueness}: A reconstruction $\hat{\fv}$ of a given signal $\fv$ is unique if and only if $\Sc^\perp\cap\Tc = \{\zerov\}$.
Otherwise all possible reconstructions form the closed plane defined as $\hat{\fv} + \{\Sc^\perp\cap\Tc\}$.
\end{enumerate}
Possible basic assumptions on the reconstruction can be: 
\begin{enumerate}[label=\bfseries (A\arabic*), topsep=0pt, itemsep=0pt]
\item \label{ASC} \emph{Sample Consistent}: The reconstructed signal yields the same sample as the original signal, i.e. $\Sm\hat\fv = \Sm\fv,~\forall\fv$. 

\item \label{ASS} \emph{Sample Sufficient}: The reconstructed signal is fully determined, up to signals in $\Sc^\perp\cap\Tc$, by the sample of the original signal, i.e. $\hat\fv_1-\hat\fv_2\in \Sc^\perp\cap\Tc,~\forall\fv_1$ and $\fv_2$ such that $\Sm\fv_1=\Sm\fv_2$. 

\item \label{AECR} \emph{Guiding Subspace Reconstruction}: Signals in the guiding reconstruction subspace are reconstructed within the subspace, i.e.\ $\hat\fv\in\Tc,~\forall\fv\in\Tc$. 

\item \label{ARS} \emph{Reconstruction Stability}: A small change in the original signal results in a proportionally small change in the reconstructed signal, up to signals in $\Sc^\perp\cap\Tc$.
\end{enumerate}

Axioms \ref{ASC} and \ref{ASS} imply that repeated reconstruction does not change, up to signals in $\Sc^\perp\cap\Tc$, an already reconstructed signal, i.e. 
$\hat\fv_2-\fv_2\in \Sc^\perp\cap\Tc,~\forall\fv_2$ such that $\fv_2=\hat\fv_1$, for an arbitrary $\fv_1.$ Indeed, $\Sm\hat\fv_1 = \Sm\fv_1$ by \ref{ASC}, so let us denote 
 $\fv_3=\Sm\hat\fv_1 = \Sm\fv_1$. Axiom \ref{ASS} gives 
 $\hat\fv_2 - \hat\fv_3\in \Sc^\perp\cap\Tc$, using $\fv_3=\Sm\hat\fv_1=\Sm\fv_2$, and
 $\hat\fv_3 - \hat\fv_1\in \Sc^\perp\cap\Tc$, using $\fv_3 = \Sm\fv_1$, thus
   $\hat\fv_2 - \hat\fv_1\in \Sc^\perp\cap\Tc$, which proves the claim.

Axioms \ref{ASC} and  \ref{AECR}  imply full conditional reconstruction, where signals in the guiding reconstruction subspace are exactly reconstructed, up to signals in $\Sc^\perp\cap\Tc$, i.e.\ we have that $\hat\fv-\fv\in\Sc^\perp\cap\Tc,~\forall\fv\in\Tc$. Indeed, \ref{ASC} is equivalent to $\hat\fv-\fv\in\Sc^\perp,~\forall\fv\in\Hc$; at the same time, \ref{AECR} is  equivalent to $\hat\fv-\fv\in\Tc,~\forall\fv\in\Tc$. Thus,  $\hat\fv-\fv\in\Sc^\perp\cap\Tc,~\forall\fv\in\Tc.$

On the one hand, we want to define a reconstruction operator $\Rm:\Hc\to\Hc$, i.e.\ the reconstructed signal $\hat{\fv}$ of $\fv$ is given by $\hat{\fv} = \Rm\fv$, which requires uniqueness of  $\hat{\fv}$. On the other hand, the nontrivial intersection $\Sc^\perp\cap\Tc\neq\{\zerov\}$ naturally appear in some applications; see, e.g.,\   ~\cite{Hirabayashi-Unser-2007}. Not having additional information, one cannot decide if any one reconstruction from the  plane $\hat{\fv} + \{\Sc^\perp\cap\Tc\}$ is better or worse than another, according to \ref{AU}.
Mathematically, we can resolve the issue by replacing the space $\Hc$ with a quotient-space  $\Hc/\{\Sc^\perp\cap\Tc\}$, collapsing $\Sc^\perp\cap\Tc$ into zero, and consistently replacing the subspaces $\Sc$ and $\Tc$ with similar quotient-spaces. After such replacements, we have $\Sc^\perp\cap\Tc = \{\zerov\}$, which we now assume for the rest of this section, so the reconstruction operator $\Rm$ is correctly defined by  $\hat{\fv} = \Rm\fv$. 

Below we list possible requirements for the reconstruction operator $\Rm$, matching  \ref{ASC}, \ref{ASS},  \ref{AECR}, and \ref{ARS}:
\begin{enumerate}[label=\bfseries (B\arabic*), topsep=0pt, itemsep=0pt]
\item \label{BSC} \emph{Sample Consistent}: The reconstructed signal yields the same sample as the original signal, i.e. $\Sm\Rm=\Sm$. 

\item \label{BSS} \emph{Sample Sufficient}: The reconstructed signal is fully determined by the sample of the original signal, i.e. the identity
$\Rm\fv_1=\Rm\fv_2$ holds $\forall\fv_1$ and $\fv_2$ such that $\Sm\fv_1=\Sm\fv_2$. 

\item \label{BECR} \emph{Guiding Subspace Reconstruction}: The guiding reconstruction subspace $\Tc$ is $\Rm$-invariant, i.e. the inclusion
 $\Rm\fv\in\Tc$ hold $\forall\fv\in\Tc$. 

\item \label{BRS} \emph{Reconstruction Stability}: The reconstruction  operator $\Rm$ is continuous. 

\end{enumerate}

We note that  \ref{BSS} implies (and for a linear reconstruction operator  $\Rm$ is equivalent to) the identity $\Rm=\Rm\Sm$. Therefore, axioms \ref{BSC} and \ref{BSS} lead to $\Rm^2=\Rm$, i.e. that the reconstruction operator is a projector (idempotent), since 
$\left(\Rm\right)\Rm=\left(\Rm\Sm\right)\Rm=\Rm(\Sm\Rm)=\Rm(\Sm)=\Rm.$ 

Having $\Sc^\perp\cap\Tc = \{\zerov\}$ in addition to axioms \ref{BSC} and  \ref{BECR},  implies a full conditional reconstruction, where signals in the guiding reconstruction subspace are exactly reconstructed, i.e.\ $\Rm\fv = \fv,~\forall\fv\in\Tc$. Indeed, \ref{BSC} is equivalent to $\Rm\fv-\fv\in\Sc^\perp,~\forall\fv\in\Hc$, and \ref{BECR} is  equivalent to $\Rm\fv-\fv\in\Tc,~\forall\fv\in\Tc$, thus  $\Rm\fv-\fv\in\Sc^\perp\cap\Tc=\{\zerov\},~\forall\fv\in\Tc.$
 
Making requirement \ref{BECR} stricter, such that the reconstructed signal $\Rm{\fv}$ is always constrained to the guiding (in this case actually target) reconstruction subspace $\Tc$, in addition to  \ref{BSC} and \ref{BSS}, results in a single valid choice of the reconstruction operator $\Rm$, given by an oblique projector $\Pm_{\Tc\perp\Sc}$, see~\cite{Unser-TSP-94, Eldar-JFA-03}, onto the subspace $\Tc$ along the orthogonal complement  $\Sc^\perp$ to the sampling subspace $\Sc$. Defining the oblique projector requires assuming $\Sc^\perp+\Tc = \Hc$ in addition to $\Sc^\perp\cap\Tc = \{\zerov\}$, together necessary and sufficient for existence and uniqueness of the intersection of the sample-consistent reconstruction plane $\Sm\fv+ \Sc^\perp$ with the  guiding reconstruction subspace $\Tc$ for any original signal $\fv$; see~\cite{Hirabayashi-Unser-2007,Corach-Giribet-2011}. The linear operator $\Rm=\Pm_{\Tc\perp\Sc}$ satisfies \ref{BSC}, see~\cite{Unser-TSP-94, Eldar-JFA-03},
and is bounded; see \cite{Corach-Giribet-2011} and our discussion in Sec.~\ref{s:GenRec}. 

The traditional assumption  $\Sc^\perp+\Tc = \Hc$ may result in the oblique projector $\Pm_{\Tc\perp\Sc}$ with a large norm. 
To circumvent the  assumption  $\Sc^\perp+\Tc = \Hc$, authors of \cite{Adcock:2013:BCR} propose 
a more general constrained reconstruction using the oblique projector  $\Rm=\Pm_{\Tc\perp\Sm(\Tc)}$, onto the subspace $\Tc$ along the orthogonal complement to the sampling subspace $\Sm(\Tc)\subseteq\Sc$. 
This reconstruction minimizes a distance from the reconstructed signal within the guiding subspace $\Tc$ to 
 the sample-consistent reconstruction plane $\Sm\fv+ \Sc^\perp$. 
 If the distance is zero, the reconstruction is sample consistent, i.e.\ satisfies  \ref{BSC}; 
 otherwise it represents the generalized reconstruction of~\cite{adcock2012generalized}.

Assumptions \ref{BSC}, \ref{BSS}, and \ref{BECR} may be approximated, or even completely abandoned. For example, minimax regret in~\cite{Eldar-TSP-06} leads to the reconstruction $\Rm=\Tm\Sm$, where $\Tm$ is an orthogonal projector onto the  guiding reconstruction subspace $\Tc$, which easily meets requirements \ref{BSS}, a stricter version of \ref{BECR}, and \ref{BRS}, but not \ref{BSC}. 

Sometimes, no target or even guiding reconstruction subspace is available or known at all, so assumptions  \ref{BECR} or $\Rm{\fv} \in \Tc$ are inapplicable and replaced with signal energy minimization. The reconstructed signal $\hat{\fv}=\Rm{\fv}$ in \cite{Eldar-SPM-09} solves the following constrained minimization problem
\begin{equation}
\inf_{\hat{\fv}} \|\Hm\hat{\fv}\| \text{ subject to } \Sm\hat{\fv} = \Sm\fv,
\label{eq:Lrecon_problem}
\end{equation}
with a non-singular operator $\Hm$. 
Taking $\Hm=\Tm^\perp+\alpha\mathbf{I}$ with $\alpha\to0$ in \eqref{eq:Lrecon_problem} 
approximates our core minimization problem, introduced in the next section.

\section{Proposed Reconstruction Methods}\label{s:prm}
\subsection{Sample Consistent reconstruction}
We first propose a novel formulation and algorithms for the \emph{sample consistent} reconstruction, used in \cite{bansal2005bandwidth}, which relaxes the constraint that $\hat{\fv} \in \Tc$, used in~\cite{Unser-TSP-94, Eldar-JFA-03}, instead minimizing the energy in $\Tc^\perp$, consistently with the sample, as in  \ref{ASC}. We provide mathematical background, taking advantage of a theory developed in~\cite{Knyazev-CMM-07}, that is then used to address the issues of existence, uniqueness, and to prove  \ref{AECR} and \ref{ARS}, giving necessary theoretical foundation, supplementing~\cite{bansal2005bandwidth}.

Specifically, the reconstructed signal $\hat{\fv}$ is determined as a solution of the following constrained minimization problem
\begin{equation}
\inf_{\hat{\fv}} \|\hat{\fv}-\Tm\hat{\fv}\| \text{ subject to } \Sm\hat{\fv} = \Sm\fv,
\label{eq:recon_problem}
\end{equation}
which is equivalent to the problem  
\begin{equation}
\inf_{\hat\xv \in \Sc^\perp} \Braket{\left(\hat\xv+\Sm\fv\right), \Tm^\perp \left(\hat\xv+\Sm\fv\right)},
\label{eq:opt_energy}
\end{equation}
where $\hat\xv=\hat\fv-\Sm\fv$.
If the solutions $\hat{\fv}$ and $\hat{\xv}$ to problems  \eqref{eq:recon_problem} and  \eqref{eq:opt_energy}, correspondingly, are not unique, we choose solutions in the corresponding factor-spaces, e.g., the normal (i.e. with the smallest norm) solutions $\hat{\fv}_n$ and $\hat{\xv}_n$ to guarantee the uniqueness, required to define the reconstruction operator $\Rm.$
The reconstruction based on solving  \eqref{eq:recon_problem} satisfies assumptions \ref{AU}, \ref{ASC}, and \ref{ASS} by design. 

Under the assumptions  $\Sc^\perp+\Tc = \Hc$ and $\Sc\cap\Tc^\perp = \{\zerov\}$, traditional in the literature,  the solution $\hat{\fv}$ of \eqref{eq:recon_problem} is just the same as the result of the oblique projection $\Pm_{\Tc\perp\Sm}$ in~\cite{Unser-TSP-94, Eldar-JFA-03}, but our method and the resulting algorithms are different, based only on actions of orthogonal projectors $\Tm$ and $\Sm$ without necessarily using frames. Moreover, we need neither of the assumptions, which makes our method robust in applications, and allows choosing a greater variety of the subspaces, compared to conventional reconstruction.  
For example, violating the assumption $\Sc^\perp+\Tc = \Hc$ allows oversampling, e.g.,\ for handling noisy data and sensors, 
as advocated in \cite{adcock2012generalized,Adcock:2013:BCR}.

Problem \eqref{eq:opt_energy} can be equivalently written in the following operator form,  
\begin{equation}
\left(\Sm^\perp \Tm^\perp\right) \big| _{\Sc^\perp} \xv = - \Sm^\perp\Tm^\perp \Sm\fv, 
\label{eq:recon_sym_eq_bl}
\end{equation} 
where $\left(\cdot\right)| _{\Sc^\perp}$ denotes the operator restriction to its invariant subspace $\Sc^\perp$ (i.e. the domain of $\Sm^\perp \Tm^\perp $ is restricted to $\Sc^\perp$). 
If $\hat{\xv}$ is a solution to the above problem, then the reconstructed signal $\hat{\fv} = \hat{\xv} + \Sm\fv$ equivalently satisfies
\begin{equation}
\Sm^\perp\Tm^\perp\hat{\fv} = \zerov \text{ and } \Sm\hat{\fv} = \Sm\fv,
\label{eq:recon_sym_eq_con_zero}
\end{equation}
which is an operator form of our constrained minimization \eqref{eq:recon_problem}.

System of equations \eqref{eq:recon_sym_eq_con_zero} is a particular case of the following system, investigated in~\cite{Knyazev-CMM-07} (see also~\cite{Arias2014})
\begin{equation}
\Sm^\perp(\Am\hat{\fv} - \hv) = \zerov  \text{ and }
\Sm(\hat{\fv}-\fv) = \zerov,
\label{eq:recon_sym_eq_con}
\end{equation}
where $\Am$ is a bounded self-adjoint non-negative operator on $\Hc$, i.e.\ $\Am=\Am^\star\geq0$. When $\hv = \zerov$ and $\Am = \Tm^\perp$, we get system~\eqref{eq:recon_sym_eq_con_zero} and $N(\Am)=\Tc$. If we split $\hat{\fv}$ as in \eqref{eq:split} then system \eqref{eq:recon_sym_eq_con} is equivalent to
\begin{equation}
\left(\Sm^\perp \Am\right) \big| _{\Sc^\perp} \xv = \Sm^\perp\left(\hv -  \Am \Sm\fv\right).
\label{eq:recon_sym_eq}
\end{equation}
Conditions for existence and uniqueness of the solutions of equations  \eqref{eq:recon_sym_eq_con} and \eqref{eq:recon_sym_eq} derived in~\cite{Knyazev-CMM-07} are being adapted below for reconstruction problem \eqref{eq:recon_sym_eq_con_zero} in Sec.s \ref{s:uniquness} and \ref{sec:es}.

Systems \eqref{eq:recon_sym_eq_bl} and \eqref{eq:recon_sym_eq_con_zero} are advantageous for numerical solution, e.g.,\ can be solved iteratively. In Sec. \ref{sec:CG}, 
we propose a conjugate gradient iterative method for solving \eqref{eq:recon_sym_eq_bl}. The matrix of the orthoprojector 
$\Tm^\perp$ or $\Tm=\mathbf{I} -\Tm^\perp$ is not needed in an iterative solver, and can be substituted with a function defining a multiplication of the orthoprojector by a given vector. 
The~multiplication can be approximate, e.g.,\ implementing an action of a signal filter, as we describe in Sec.~\ref{sec:gs},
instead of relying on a traditional frame-based definition of the guiding subspace $\Tc$. 
Moreover, a generic filter may substitute $\Tm$ or $\Tm^\perp$ in \eqref{eq:recon_problem}, \eqref{eq:opt_energy}, and \eqref{eq:recon_sym_eq_con_zero}, 
but analyzing such a substitution is beyond the scope of the present paper.

Least squares minimization formulations \eqref{eq:recon_problem} and  \eqref{eq:opt_energy} have an equivalent elegant geometric interpretation, cf., equality  \eqref{eq:recon_problem_geom01}, e.g.,\ for lines in 3D in Fig. \ref{fig:1}, 
where the second  minimization problem in \eqref{eq:recon_problem_geom01} simply determines the shortest distance between the sample-consistent closed plane $\Sm\fv+\Sc^\perp$ and the guiding closed subspace~$\Tc$.
Indeed, in the first minimization problem in \eqref{eq:recon_problem_geom01}, the inner minimization for a fixed vector $\hat{\fv}$ is  $\inf_{\hat{\tv}\in\Tc} \|\hat{\fv}-\hat{\tv}\|$, which always has a solution $\hat{\tv}=\Tm\hat{\fv}\in\Tc$ using the orthogonality argument. The outer minimization is then 
exactly our problem \eqref{eq:recon_problem}.

\vspace{-0.3cm}
\subsection{Strictly Guided (Generalized) Reconstruction}
In~the last minimization problem, \eqref{eq:recon_problem_geom02}, we swap the order of  minimization, compared to  the first minimization problem in equality \eqref{eq:recon_problem_geom01}. 
We can call a solution $\hat{\tv}\in\Tc$ of problem  \eqref{eq:recon_problem_geom02} a \emph{strictly guided} reconstructed signal. 
In Sec. \ref{sec:cgr}, we discuss additional assumptions that turn our strictly guided reconstructed signal $\hat{\tv}\in\Tc$ into well-known generalized reconstructed signal, proposed in \cite{adcock2012generalized,Adcock:2013:BCR}.

By analogy with the operator form \eqref{eq:recon_sym_eq_con_zero} of the first minimization problem in \eqref{eq:recon_problem_geom01}, minimization problem \eqref{eq:recon_problem_geom02} is equivalent to %the following equation
\begin{equation}
\Tm\Sm\left(\hat{\tv}-\fv\right) = \zerov, \text{ where } \hat{\tv} \in\Tc.
\label{eq:sgrecon}
\end{equation}
Indeed, for a fixed vector $\tv$, the minimization problem $\inf_{\hat{\fv}\in\Sm\fv+\Sc^\perp} \|\hat{\fv}-\tv\|$
is equivalent in our Hilbert space $\Hc$ to the orthogonality condition $\hat{\fv}-\tv\perp\Sc^\perp$, i.e. $\hat{\fv}-\tv\in\Sc$, 
which is equivalent to $\hat{\fv}=\Sm\fv+\Sm^\perp\tv$ solving \eqref{eq:sgrecon}, and thus turning the ``inf'' into ``min'' in the minimization. 
Due to the linear constraint $\tv\in\Tc$ in the outer minimization in \eqref{eq:recon_problem_geom02}, its minimizer $\hat{\tv}\in\Tc$, if exists, satisfies 
the orthogonality condition 
\[\hat{\fv}-\hat{\tv}=\left(\Sm\fv+\Sm^\perp\hat{\tv}\right)-\hat{\tv}= \Sm\left(\hat{\tv}-\fv\right)\in\Tc^\perp,\]
equivalent to \eqref{eq:sgrecon}, which completes the argument.

It is interesting to compare the solution $\hat{\tv}$ of \eqref{eq:sgrecon} to the constrained frame-less reconstruction given in \cite{adcock2012generalized,Adcock:2013:BCR} by the oblique projector  $\Pm_{\Tc\perp\Sm\Tc}$ on the closed subspace $\Tc$ along the closed subspace $\left(\Sm\Tc\right)^\perp$. Motivated by Fig. \ref{fig:1},
we prove in Sec.~\ref{s:GenRec},  that  $\Pm_{\Tc\perp\Sm\Tc}\fv=\hat{\tv}$,  under an additional assumption $\Sc^\perp\cap\Tc=\{\zerov\}$ for uniqueness of 
$\hat{\tv}$, which is required to define the single-valued operator $\Pm_{\Tc\perp\Sm\Tc}$. Other applicability assumptions in \cite{adcock2012generalized,Adcock:2013:BCR} are equivalent to ours. In the case $\Sc\cap\Tc^\perp \neq \{\zerov\}$, we also discuss how 
the strictly guided reconstruction  $\hat{\fv}$ can be defined via a factor space analysis.

Equation  \eqref{eq:sgrecon} can be solved iteratively, e.g., by the conjugate gradient method, if equivalently transformed into 
$\left(\Tm\Sm\right) \big| _{\Tc}\hat{\tv}= \Tm\Sm\fv,$ or $\Tm\Sm\Tm\hat{\tv}= \Tm\Sm\fv,$
providing us with an interesting alternative to solving \eqref{eq:recon_sym_eq_bl},
cf. \cite{7480396,adcock2012generalized,Adcock:2013:BCR}. For~example, equation \eqref{eq:sgrecon} does not require knowing the sampling subspace $\Sc$ explicitly, in contrast to \eqref{eq:recon_sym_eq_bl}. Moreover, the matrix of the orthoprojector $\Sm$ is not needed in \eqref{eq:sgrecon}, being replaced with a sampling function defining a multiplication of the orthoprojector $\Sm$ by a given vector. Furthermore, the sampling function can be approximate, not necessarily having a null-space, and may even change during the course of iterations, e.g.,\ varying in time for time-series signals or depending on the current iterative reconstructed signal. 

Flexibility of approximating both the sampling and the guiding procedures, which is possible in formulation \eqref{eq:sgrecon}, appears important in practical applications, but such extensions are beyond the scope of the present paper. We only note here that the minimax regret \cite{Eldar-TSP-06}  reconstruction $\Tm\Sm\fv$ can be interpreted as a rudimentary one-step of an iterative solver with the zero initial guess for solving $\Tm\Sm\Tm\hat{\tv}= \Tm\Sm\fv.$ 

\subsection{Implications of conditions of optimality}
Turning our attention to the second  minimization problem in \eqref{eq:recon_problem_geom01}, for the shortest distance between the sample-consistent closed plane $\Sm\fv+\Sc^\perp$ and the guiding closed subspace~$\Tc$, we obtain the following first-order necessary and sufficient conditions of optimality 
\begin{equation}
\left\{
\setlength\arraycolsep{0pt}
\begin{array}{ r @{{}={}} l }
\hat{\tv} &  \Tm\hat{\fv} \\
\hat{\fv} &  \Sm\fv+\Sm^\perp\hat{\tv}\\
\end{array}
\right.,
 \label{eq:mixed}
\end{equation}
already derived just above. 

Both the sample consistent   $\hat{\fv}$  and strictly guided $\hat{\tv}$ reconstructions can in principle be computed 
together by solving the system of equations \eqref{eq:mixed} numerically. Instead of doubling the number of unknowns, one can substitute
the second equation $\hat{\fv} = \Sm\fv+\Sm^\perp\hat{\tv}$ in system \eqref{eq:mixed} into the first one, $\hat{\tv} = \Tm\hat{\fv}$, obtaining 
the equation $\left(\mathbf{I}- \Tm\Sm^\perp\right)\hat{\tv}=\Tm\Sm\fv$ for $\hat{\tv}$ only; cf. e.g.,\ \cite[Sec. IV]{7480396}. The latter equation turns into already considered above equation \eqref{eq:sgrecon}, since $\mathbf{I}- \Tm\Sm^\perp=\mathbf{I}- \Tm + \Tm\Sm$ and 
 $\hat{\tv}\in\Tc$, so that $\left(\mathbf{I}- \Tm\right)\hat{\tv}=\zerov,$
if $\Tm$ is indeed the exact orthoprojector onto $\Tc$, as we assume throughout the paper.
 
We only use \eqref{eq:mixed} here to discover a very important identity, in the next paragraph.  

Multiplying both parts of the first equation in \eqref{eq:mixed} by $\Tm$, we see that $\hat{\fv}-\hat{\tv}\in\Tc^\perp$. 
Multiplying both parts of the second equation in \eqref{eq:mixed} by $\Sm^\perp$, one confirms that $\hat{\fv}-\hat{\tv}\in\Sc$, as already used in deriving equation \eqref{eq:sgrecon}.
Thus, $\hat{\fv}-\hat{\tv}\in\Sc\cap\Tc^\perp$.
Moreover, it follows from \eqref{eq:mixed} that 
\[
 \Pm_{\Sc\cap\Tc^\perp}\left(\hat{\fv}-\hat{\tv}\right)=\Pm_{\Sc\cap\Tc^\perp}\hat{\fv}=\Pm_{\Sc\cap\Tc^\perp}\fv, 
\]
where $\Pm_{\Sc\cap\Tc^\perp}$ is an orthoprojector onto the closed subspace $\Sc\cap\Tc^\perp.$
We come to a simple orthogonal decomposition  
\begin{equation}
\hat{\fv} = \hat{\tv} + \Pm_{\Sc\cap\Tc^\perp}\fv. 
 \label{eq:tf}
\end{equation}

Knowing \eqref{eq:mixed} and \eqref{eq:tf} allows us to primarily concentrate in the rest of the paper on the consistent reconstruction $\hat{\fv}$.

\subsection{Reconstruction Set and Regularization}
Assuming that all minimization problems in \eqref{eq:recon_problem_geom01}--\eqref{eq:recon_problem_geom02} have solutions, we define a reconstruction set as a union of closed intervals with the end points $\hat{\fv}\in\Sm\fv+\Sc^\perp$ and $\hat{\tv}=\Tm\hat{\fv}\in\Tc$. If the solution is unique, 
the reconstruction set is a single interval, as illustrated in Fig.~\ref{fig:1}.

Within the reconstruction set, the sample consistent reconstructed signal  $\hat{\fv}$ is evidently expected to have the smallest 
reconstruction error $\left\| \hat{\fv} - \fv\right\|$, assuming that the sample $\Sm\fv$ is accurate. 
Identities \eqref{eq:mixed} and \eqref{eq:tf} and the Pythagorean theorem immediately imply the following theorem.
\begin{theorem}
\label{thm:rcsize}
Let $\hat{\fv}_\alpha$ be any point in the reconstruction set given by $\hat{\fv}_\alpha =\alpha\hat{\fv}+(1-\alpha)\hat{\tv}, \ 0 \leq \alpha \leq 1$. Then the reconstruction error is given by
\[
	\left\| \hat{\fv}_\alpha - \fv\right\|^2 = \left\| \hat{\fv} - \fv\right\|^2 + (1-\alpha)^2\left\|\Pm_{\Sc\cap\Tc^\perp}\fv\right\|^2,
\]
where  $\left\|\Pm_{\Sc\cap\Tc^\perp}\fv\right\|$ is the shortest distance, defined by \eqref{eq:recon_problem_geom01}, between the sample-consistent closed plane $\Sm\fv+\Sc^\perp$ and the guiding closed subspace~$\Tc$.
\end{theorem}
If we trust that the sample-consistent closed plane $\Sm\fv+\Sc^\perp$ is actually accurate, by Theorem  \ref{thm:rcsize}, the reconstruction error is indeed minimized on the sample consistent reconstructed signal,  $\hat{\fv}\in\Sm\fv+\Sc^\perp$, given by the end point $\alpha=1.$ 
If~there is noise in sample measurements, we may decide to trust the guiding closed subspace $\Tc$ more than the sample $\Sm\fv$ 
and choose as our output reconstruction a convex linear combination $\alpha\hat{\fv}+(1-\alpha)\Tm\hat{\fv}$ within the reconstruction set, 
where $0\leq\alpha<1.$ The other extreme choice $\alpha=0$ gives the strictly  guided reconstruction $\hat{\tv}=\Tm\hat{\fv}$, already discussed.

Having the complete reconstruction set determined allows selecting a single reconstructed signal in it, e.g., by minimizing some cost/quality function, where minimization is constrained to a neighborhood of the reconstruction set. A signal energy is one example of the cost/quality function, e.g., leading to minimization like in  \eqref{eq:Lrecon_problem}, but  constrained to a neighborhood of the reconstruction set. 
Such a procedure eliminates a typical difficulty of choosing a good regularization parameter in regularization-based inconsistent methods in \cite{Narang-GlobalSIP-13}.

In practical applications, it is common that the sampling procedure involves an inaccuracy in the input signal, where the inaccuracy may appear due to one or a combination of a noise, a limited accuracy of a sensor providing the sampling procedure, and a limited precision of data representing the input signal. If one can determine a level of the inaccuracy in the input signal relative to the shortest distance $\left\|\Pm_{\Sc\cap\Tc^\perp}\fv\right\|$ in \eqref{eq:recon_problem_geom01} between the sample consistent and guiding sets, the cost/quality function can be constructed that takes the level of the inaccuracy into account. 

One also can relax the reconstruction set constraint in the cost/quality function minimization, and consider alternative formulations, e.g., like in interior point methods, minimizing a weighted sum based on the cost/quality function and a distance to the reconstruction set, or based on primal-dual relaxations, but this goes beyond the scope of the present paper. 

%We further note that an opportunity to introduce a proper norm $\|\cdot\|$ is important in applications, e.g., related to graph-based signal processing, where a good choice of the guiding reconstruction subspace $\Tc$ may be not so evident, so our least-squares approach, minimizing the norm squared $\|\cdot\|^2$, can improve quality of the reconstruction, if the norm $\|\cdot\|$ properly measures a signal energy.   

\section{Uniqueness of Reconstructed Signal}\label{s:uniquness}
The following theorem gives a condition of our reconstruction $\hat{\fv}$ uniqueness.
\begin{theorem}{(Based on \cite[Lemma~4.2]{Knyazev-CMM-07})}
\label{thm:uniqueness}
Let $\hat\xv\in\Sc^\perp$ be a solution of \eqref{eq:recon_sym_eq_bl} and $\hat{\fv}=\hat{\xv} + \Sm\fv$ be a solution of \eqref{eq:recon_sym_eq_con_zero}.
The solutions $\hat\xv$ and $\hat{\fv}$ are unique if and only if $\Sc^\perp\cap\Tc = \{\zerov\}$.
Otherwise, all solutions form a plane 
$\hat{\xv} + \{\Sc^\perp\cap\Tc\}$ for \eqref{eq:recon_sym_eq_bl} 
and  a plane 
$\hat{\fv} + \{\Sc^\perp\cap\Tc\}$ for \eqref{eq:recon_sym_eq_con_zero}. There exists  unique normal solutions (with minimal norm in $\Hc$) $\hat\xv_n\in\Sc^\perp$ of \eqref{eq:recon_sym_eq_bl}
and  $\hat{\fv}_n$ of \eqref{eq:recon_sym_eq_con_zero}, which belong to the intersection of the corresponding plane and the closed subspace $ \left(\Sc^\perp\cap\Tc\right)^\perp = \overline{\Sc^\perp+\Tc}$, and where $\hat{\fv}_n=\hat{\xv}_n + \Sm\fv$. 
\end{theorem}

Theorem~\ref{thm:uniqueness} gives us enough information to prove \ref{AECR}.
\begin{theorem}
\label{thm:AECR}
Reconstruction method \eqref{eq:opt_energy} satisfies \ref{AECR}.
\end{theorem}\vspace{-0.1in}
\begin{proof}
According to \ref{AECR}, $\fv\in\Tc$, but then 
$\hat{\fv}=\hat{\xv} + \Sm\fv=\fv \in \Tc$ is a solution of \eqref{eq:recon_problem}, 
since the minimizing quantity turns into zero, which is its smallest possible value. 
By Theorem~\ref{thm:uniqueness}, all solutions of \eqref{eq:recon_sym_eq_con_zero} 
form the plane $\hat{\fv} + \{\Sc^\perp\cap\Tc\}\subseteq\Tc$,  since 
$\hat{\fv}=\fv\in\Tc$ and $ \{\Sc^\perp\cap\Tc\}\subseteq\Tc$.
\end{proof}\vspace{-0.1in}

By Theorem~\ref{thm:uniqueness}, if $\Sc^\perp\cap\Tc \neq \{\zerov\}$, the solution $\hat{\xv}$ to the reconstruction problem~\eqref{eq:recon_sym_eq_bl} and the reconstruction itself, $\hat{\fv}=\hat{\xv} + \Sm\fv$, determined by \eqref{eq:recon_sym_eq_con_zero}, are both not unique, and vice versa, consistently with the assumption \ref{AU}. This can happen, e.g.,\ if the number of samples is too small or when the guiding reconstruction space is too large. A~similar issue appears in \cite{Hirabayashi-Unser-2007}, dealing with non-unique strictly consistent reconstructions in $\Tc$ by choosing a subspace in $\Tc$, i.e.\ constraining the guiding reconstruction space. Here, we propose a different approach, constraining the orthogonal complement  $\Sc^\perp$ of the sampling subspace $\Sc$. 

The reconstruction $\hat{\fv}$ is determined up to an arbitrary signal from the intersection $\Sc^\perp\cap\Tc \neq \{\zerov\}$. In section~\ref{s:overview}, we treat the plane $\hat{\fv} + \{\Sc^\perp\cap\Tc\}$ as a unique element of the quotient space $\Hc/\{\Sc^\perp\cap\Tc\}$, factoring out the intersection $\Sc^\perp\cap\Tc \neq \{\zerov\}$. A quotient space is mathematically powerful, but may be impractical in some applications. In practice, it may be desired to choose, by imposing further restrictions on the reconstruction, a single solution representing the equivalence class---the whole plane of solutions. The minimum norm solution is one such choice of a unique representative obtained by restricting the solution to be in $\left(\Sc^\perp\cap\Tc\right)^\perp$, suggested in Theorem~\ref{thm:uniqueness}. However, the minimum norm requirement may not be relevant for properties of the signal to be reconstructed. 

Alternatively, we can obtain a well-defined unique reconstruction by choosing 
 the solution $\hat{\xv}$ to the reconstruction problem~\eqref{eq:recon_sym_eq_bl} and the reconstruction $\hat{\fv}=\hat{\xv} + \Sm\fv$ determined by \eqref{eq:recon_sym_eq_con_zero} in a given closed subspace $\Mc$. 
The normal solution is a special case, where $\Mc = \left(\Sc^\perp\cap\Tc\right)^\perp$. 
We note that if the uniqueness condition $\Sc^\perp\cap\Tc = \{\zerov\}$ is satisfied, then $\Mc$ has to be equal to $\Hc$ so that there is no restriction on the reconstruction, which is consistent, e.g., with the choice of the normal solution. 

It is known that in order to be  isomorphic to the quotient space $\Hc/\{\Sc^\perp\cap\Tc\}$ it is necessary and sufficient for the subspace $\Mc$ to be complimentary to $\Sc^\perp\cap\Tc$,~i.e.
\begin{equation}
\Mc + \left(\Sc^\perp\cap\Tc\right) = \Hc \text{ and } \Mc \cap \left(\Sc^\perp\cap\Tc\right) = \{\zerov\},
\label{eq:complimentary-M}
\end{equation}
which we assume to hold for the rest of the section.

Assumptions \eqref{eq:complimentary-M} imply that the solution $\xv_\Mc$ of \eqref{eq:recon_sym_eq_bl} in $\Mc$ is unique; cf. \cite[Proposition 2]{Hirabayashi-Unser-2007}. 
%
%\begin{proof}
%Let $\xv_1$ and $\xv_2$ be two solutions of \eqref{eq:recon_sym_eq_bl} in $\Mc$. This implies that $\xv_1 - \xv_2 \in N((\Sm^\perp\Tm^\perp)|_{\Sc^\perp}\right) = \Sc^\perp\cap\Tc$. But $\xv_1 - \xv_2 \in \Mc$. Hence, $\xv_1 - \xv_2 \in \Mc \cap \left(\Sc^\perp\cap\Tc) = \{\zerov\}$, which shows that $\xv_1 = \xv_2$.
%\end{proof}
%
With the additional constraint $\hat\xv\in\Mc$ that makes the reconstruction unique, reconstruction problem \eqref{eq:opt_energy} becomes
\begin{equation}
\inf_{\hat\xv \in \Mc\cap\Sc^\perp} \Braket{\left(\hat\xv+\Sm\fv\right), \Tm^\perp \left(\hat\xv+\Sm\fv\right)}.
\label{eq:recon-energy-non-unique}
\end{equation}
In order to write problem \eqref{eq:recon-energy-non-unique} in an unconstrained form similar to \eqref{eq:recon_sym_eq_bl}, we introduce orthogonal projectors $\Mm$ onto $\Mc$ and $\Fm$ onto the subspace $\Fc = \Mc \cap \Sc^\perp$.  The projector onto an  intersection of two subspaces has a closed form expression in terms of the projectors for the individual subspaces, given by the Anderson-Duffin formula~\cite{Anderson-JMA-69}, 
$
\Fm = 2\Mm \left(\Mm + \Sm^\perp\right)^\dagger\Sm^\perp.
$  
We note that the pseudo-inverse $\left(\Mm + \Sm^\perp\right)^\dagger$ above is bounded if and only if the minimal gap between subspaces $\Mc$ and $\Sc^\perp$ is positive; 
see, e.g.,\ \cite[Theorems 2.15 and 2.18]{Knyazev-JFA-10}.

Having the projector $\Fm$, we rewrite equation \eqref{eq:recon-energy-non-unique}, by analogy with \eqref{eq:recon_sym_eq_bl}, in an equivalent form
\begin{equation}
\left(\Fm \Tm^\perp\right) \big| _{\Fc}\, \xv = - \Fm\Tm^\perp\Sm\fv, 
\label{eq:recon-non-unique}
\end{equation}
which can be solved via a conjugate gradient (CG) method. 
Let us note that in the case of the normal solution, where  
$\Mc=\left(\Sc^\perp\cap\Tc\right)^\perp$, the CG method can find the normal solution being applied directly to \eqref{eq:recon_sym_eq_bl}, not needing \eqref{eq:recon-non-unique}.

Our reconstruction satisfies \ref{ASC} and  \ref{AECR}, which imply full conditional reconstruction, i.e.\ $\hat\fv-\fv\in\Sc^\perp\cap\Tc,~\forall\fv\in\Tc$, as we already know.
We now select a unique representative reconstruction $\hat\fv$ by restricting the solution $\hat{\xv}=\Sm^\perp \hat\fv$ to be in $\Mc$. But the original signal $\fv$ itself is only 
a representative of the plane $\fv+\Sc^\perp\cap\Tc$ of signals, which are indistinguishable within our assumptions. In order to match our unique representative $\hat\fv$ satisfying   
$\Sm^\perp \hat\fv \in \Mc$ to some representative of the plane $\fv+\Sc^\perp\cap\Tc$ of original signals, we make a natural assumption on the unmeasured component $\Sm^\perp \fv \in \Mc$, to make $\fv \in \Tc$ fully recoverable in the next theorem. 
\begin{theorem}\label{thm:AECRwithM}
Let $\fv \in \Tc$ and $\Sm^\perp \fv \in \Mc$. If $\hat{\xv}$ is the unique solution of \eqref{eq:recon-energy-non-unique}, then the reconstruction $\hat{\fv} = \hat{\xv} + \Sm\fv = \fv$.
\end{theorem}
\begin{proof}
Under given conditions $\fv \in \Tc$ and $\Sm^\perp \fv \in \Mc$, clearly $\hat{\xv}=\Sm^\perp \fv$ is a solution of  \eqref{eq:recon-energy-non-unique}. But since $\hat{\xv}$ is unique, we have $\hat{\xv} = \Sm^\perp \fv$. Thus, $\hat{\fv} = \hat{\xv} + \Sm\fv = \fv$.   
\end{proof}
Even if the unmeasured part of the true signal has some energy outside of $\Mc$, this formulation ensures that the components in $\Mc$ are fully recovered. This would be beneficial if $\Mc$ is chosen such that large portion of the signal energy is expected to be contained in it.   

Next, we discuss conditions of existence and stability of the reconstructed signal. 

\section{Existence and Stability}\label{sec:es}
We begin by stating conditions for wellposedness, i.e.\ existence and stability of a solution, of problem~\eqref{eq:recon_sym_eq_con} since it is later used to give us a bound on a reconstruction error. We denote operator $\Km = \left(\Sm^\perp\Am\right)\big|_{\Sc^\perp}$ obtaining
\[R(\Km) = \Sm^\perp\Am \Sc^\perp,~N(\Km) = N(\Sm^\perp \Am)\cap \Sc^\perp = N(\Am) \cap \Sc^\perp.\]
A normal solution of equation $\Km\xv=\bv$ depends continuously on $\bv\in R(\Km)$ if and only if the pseudo-inverse operator $\Km^\dagger:R(\Km) \rightarrow \Sc^\perp/N(\Km)$ is bounded. Here $\Sc^\perp/N(\Km)$ denotes the quotient space such that $\yv, \zv \in \Sc^\perp$ are equivalent if and only if $\yv - \zv \in N(\Km)$. The operator $\Km^\dagger$ is bounded iff $R(\Km)$ is closed. The following theorem restates these conditions in terms of $\Am$ and $\Sm$ for problem~\eqref{eq:recon_sym_eq_con}. 
\begin{theorem}{(Based on~\cite[Theorem~4.3]{Knyazev-CMM-07})}
A normal solution $\hat{\fv}_n=\hat{\xv}_n+\Sm\fv$ to \eqref{eq:recon_sym_eq_con} with $\hat{\xv}_n=\Sm^\perp\hat{\fv} \in \Sm^\perp R(\Am)$ exists and depends continuously on arbitrary $\hv \in R(\Am) + \Sc$ and $\fv\in\Hc$ if and only if 
\begin{equation}
\frac{1}{\rho}:=
\inf_{\xv\in\Sm^\perp R(\Am),~\xv\neq\zerov} 
\frac{\Braket{\xv,\Am\xv}}{\Braket{\xv,\xv}}>0
%\Braket{\phi,\Am\phi} \geq \frac{1}{\rho}\Braket{\phi,\phi},~\forall\phi\in\Sm^\perp R(\Am) \text{ with } \rho > 0.
\label{eq:wellposedness-1}
\end{equation} 
Moreover, condition \eqref{eq:wellposedness-1} implies
\begin{equation}
\|\hat{\xv}_n\|^2 \leq \rho^2 \|\Sm^\perp(\hv - \Am\Sm\fv)\|^2,
\label{eq:wellposedness-2}
\end{equation} 
that also leads to an upper bound for  
$\|\hat{\fv}_n\|^2 = \|\Sm\fv\|^2 + \|\hat{\xv}_n\|^2.$
\label{thm:existence}
\end{theorem}

 Taking $\hv = \zerov$ and $\Am = \Tm^\perp$, we obtain system~\eqref{eq:recon_sym_eq_con_zero} and $N(\Am)=\Tc$.
 Condition \eqref{eq:wellposedness-1} with $\Am = \Tm^\perp$
 is equivalent to  
 \begin{equation}
%\frac{1}
{\nu} := \inf_{\xv \in \Sm^\perp \Tc^\perp} \frac{\|\Tm^\perp \xv\|}{\|\xv\|} > 0, \text{ where } \rho=\frac{1}{\nu^2}, 
\label{eq:rho-projector-norm}
\end{equation}
which becomes the key assumption. Let us describe \eqref{eq:rho-projector-norm} via concepts of the minimal gap $\gamma$ and angles $\Theta$ between subspaces.

\begin{theorem}{(Based on~\cite[Lemma~4.6]{Knyazev-CMM-07})}\label{thm:kappa}
Let $\nu$ be defined by~\eqref{eq:rho-projector-norm}. Then 
\[
\nu = \gamma(\Sc,\Tc^\perp) = \cos\left(\theta_{\max}\right), 
\]
where 
 \begin{equation}
 \gamma\left(\Sc,\Tc^\perp\right) := \inf_{\fv \in \Sc, \fv \notin \Tc^\perp} \frac{\text{dist}\left(\fv,\Tc^\perp\right)}{\text{dist}\left(\fv, \Sc \cap \Tc^\perp\right)},
 \label{eq:def_mingap}
\end{equation}
 is the minimal gap between closed subspaces $\Sc$ and $\Tc^\perp$, and
 \begin{equation}
\theta_{\max} = \sup \{ \Theta(\Sc,\Tc) \setminus \{\pi/2\}\},
\label{eq:theta_max}
\end{equation}
is the largest non-trivial angle between closed subspaces $\Sc$ and $\Tc$.
 \end{theorem} 
 \begin{proof}
According to \cite[Lemma~4.6]{Knyazev-CMM-07}, 
\begin{equation}
\inf_{\xv \in \Sm^\perp \Tc^\perp} \frac{\|\Tm^\perp \xv\|}{\|\xv\|} = 
\gamma(\Sc,\Tc^\perp),
\label{eq:projector-norm-min-gap}
\end{equation}
 where $\gamma(\Sc,\Tc^\perp)$ is the \emph{minimal gap}~\cite[Sec.~IV-4]{kato}, equal to the sine the \emph{Friedrichs angle} between subspaces $\Sc$ and $\Tc^\perp$,
\begin{equation}
\gamma\left(\Sc,\Tc^\perp\right) = \sin\left(\inf \{\Theta\left(\Sc,\Tc^\perp\right) \setminus \{\zerov\}\}\right);
\label{eq:rho-angle}
\end{equation}
see, e.g.,\ \cite[Theorem~2.15]{Knyazev-JFA-10}. 
Relationships between $\Theta\left(\Fc,\Gc\right)$, $\Theta\left(\Fc,\Gc^\perp\right)$, and $\Theta\left(\Fc^\perp,\Gc^\perp\right)$ are given in~\cite[Theorem~2.7]{Knyazev-JFA-10}). In particular, 
\begin{equation}
\inf \{\Theta(\Sc,\Tc^\perp) \setminus \{0\}\} = \pi/2 - \theta_{\max}.
\label{eq:cos_theta_max}
\end{equation}
 \end{proof}

Let us also note that by \cite[Lemma~4.6]{Knyazev-CMM-07} we have
\begin{equation} 
\gamma(\Sc,\Tc^\perp) =  \gamma(\Tc^\perp,\Sc) = \gamma(\Sc^\perp,\Tc) = \gamma(\Tc,\Sc^\perp).
%\inf_{\xv \in \Tm\Sc} \frac{\|\Sm \xv\|}{\|\xv\|}.
\label{eq:projector-norm-min-gap1}
\end{equation}
The proof of Theorem~\ref{thm:kappa} is included in the supplementary material. The assumption $\gamma(\Tc,\Sc^\perp)>0$ is equivalent to assuming that the sum  $\Sc^\perp+\Tc$ is closed. The latter is automatically satisfied if $\Sc^\perp+\Tc=\Hc$ as traditionally assumed in reconstruction literature; see, e.g.,\ \cite{Unser-TSP-94, Eldar-JFA-03}.

If $\dim\Hc<\infty$, which is the case, e.g., in graph-based signal processing, every subspace is automatically closed, i.e., 
the assumptions in our existence theorems automatically hold, in contrast to, e.g., \cite[Theorem 4.1]{7480396} requiring 
that $\Sc^\perp\cap\Tc=\{\zerov\}$. 
There is no contradiction, however, since \cite[Theorem 4.1]{7480396} postulates an existence of an exact reconstruction, i.e.,\ $\hat{\fv}=\fv$,
and correctly argues that a signal $\fv\in\Sc^\perp\cap\Tc$ cannot be exactly reconstructed, since $\Sm\fv=\zerov,$ unless $\fv=\zerov.$ 
While we merely claim the existence of a solution to equation \eqref{eq:recon_sym_eq_con_zero} and deal with issues stemming from  $\Sc^\perp\cap\Tc\neq\{\zerov\}$ separately in Sec.~\ref{s:uniquness}. 

Theorems \ref{thm:existence} and \ref{thm:kappa} immediately imply 
\begin{theorem}\label{thm:rec}
If $\cos \theta_{\max} > 0$, then there exists a solution of the reconstruction problem \eqref{eq:recon_sym_eq_con_zero} for any signal $\fv$;
the normal solution $\hat{\fv}_n$ of  \eqref{eq:recon_sym_eq_con_zero} is unique and bounded by 
\begin{equation}
\|\hat{\fv}_n\|^2 \leq \|\Sm\fv\|^2 + \|\Sm^\perp\Tm^\perp\Sm\fv\|^2/\cos^4\theta_{\max}.
\label{eq:wellposedness-fnog}
\end{equation}
\end{theorem}

Let us note that Theorem \ref{thm:rec} applies Theorem \ref{thm:existence} with $g=\zerov $ and leaves open a question whether condition \eqref{eq:wellposedness-1} or
condition \eqref{eq:rho-projector-norm} is still necessary in this case. 
In the rest of the section, we go beyond the results presented in \cite{Knyazev-CMM-07} and address this question, 
using a powerful theory for a pair of two orthogonal projectors; see, e.g.,\
 \cite{Knyazev-JFA-10}.
\begin{theorem}\label{thm:Lidentity}
We denote by $\Hc_0$ the subspace of $\Hc$ that is orthogonal to all four subspaces 
$\Sc\cap\Tc$, $\Sc^\perp\cap\Tc$, $\Sc\cap\Tc^\perp$, and $\Sc^\perp\cap\Tc^\perp$,
as introduced in \cite{Halmos-1969}. Let $\Pm_0$ be the orthogonal projector onto the subspace $\Hc_0$.

The assumption $\cos \theta_{\max} > 0$ is necessary and sufficient for existence of a solution of the reconstruction problem \eqref{eq:recon_sym_eq_con_zero} for any signal $\fv$. 
A normal solution $\hat{\xv}_n$ to \eqref{eq:recon_sym_eq_bl}, giving the normal sample consistent reconstruction 
$\hat{\fv}_n=\hat{\xv}_n+\Sm\fv$ and the normal strictly guided reconstruction 
\[\hat{\tv}_n=\Tm\hat{\fv}_n= \hat{\fv}_n-\Pm_{\Sc\cap\Tc^\perp}\fv,\]  
 exists and depends continuously on arbitrary $\fv\in\Hc$ if and only if $\cos \theta_{\max} > 0$. 
If $\cos \theta_{\max} > 0$, bound \eqref{eq:wellposedness-fnog} holds and
\begin{equation}\label{eq:cosbound}
\|\hat{\xv}_n\| \leq \|\Tm^\perp \Sm \Pm_0\fv\|/\cos\theta_{\max},
\end{equation}  
as well as
\begin{equation}\label{eq:tanbound}
\|\hat{\xv}_n\| \leq \|\Sm \Pm_0\fv\|\tan\theta_{\max},
\end{equation}  
in  
$\|\hat{\fv}_n\|^2 = \|\hat{\tv}_n\|^2 + \|\Pm_{\Sc\cap\Tc^\perp}\fv\|^2=\|\Sm\fv\|^2 + \|\hat{\xv}_n\|^2$.
\end{theorem}

\begin{proof}
Recall that $\Hc_0$ is a subspace of $\Hc$ that is orthogonal to all four subspaces 
$\Sc\cap\Tc$, $\Sc^\perp\cap\Tc$, $\Sc\cap\Tc^\perp$, and $\Sc^\perp\cap\Tc^\perp$.

On the one hand, the right-hand side of equation \eqref{eq:recon_sym_eq_bl} is in $\Hc_0\cap\Sc^\perp$, 
i.e.\ $\Sm^\perp\Tm^\perp\Sm\fv\in\Hc_0\cap\Sc^\perp$, and the set of all possible right-hand sides 
$\Sm^\perp\Tm^\perp \Sc$ in equation \eqref{eq:recon_sym_eq_bl}
is a proper, in general, subspace of  $\Hc_0\cap\Sc^\perp$.

On the other hand, each of the five spaces, including $\Hc_0$, is invariant under 
both orthogonal projectors $\Sm$ and $\Tm$, and hence their complements $\Sm^\perp$ and $\Tm^\perp$.
Let us denote by
$\Sm_0$, $\Tm_0$, $\Sm_0^\perp,$ and $\Tm_0^\perp$ the corresponding restrictions on $\Hc_0$. 
The product $ \Sm^\perp\Tm^\perp$ is also $\Hc_0$-invariant. 
Thus, the closed subspace $\Hc_0\cap\Sc^\perp$ is invariant under the operator $\Km=\left(\Sm^\perp\Tm^\perp\right) \big| _{\Sc^\perp}$. Denoting the restriction of $\Km$ to $\Hc_0\cap\Sc^\perp$ by $\Km_\star,$ 
we observe that the operator $\Km$ is a sum of the operator $\Km_\star$ and an orthogonal projector onto $\Sc^\perp\cap\Tc^\perp$.
Both operators $\Km$ and $\Km_\star$ are bounded and self-adjoint, with the same spectrum, included in the interval $[0,1]$, except that $\Km$ has an extra eigenvalue $1$, if $\Sc^\perp\cap\Tc^\perp\neq\{\zerov \},$
and an extra eigenvalue $0$, if $\Sc^\perp\cap\Tc\neq\{\zerov \},$
The smallest point of the spectrum of $\Km_\star$ is $\nu^2=\cos^2 \theta_{\max}$ defined by \eqref{eq:rho-projector-norm} and characterized in Theorem~\ref{thm:kappa}, thus,
\begin{equation}
\left\|\Km^{\dagger}\right\|=\left\|\Km_\star^{-1}\right\|=1/\nu^2,
\label{eq:1overkappa}
\end{equation}
where the sign ${}^\dagger$ means the Moore–Penrose pseudoinverse.

Therefore, we can substitute $\Km_\star$  for $\Km$  in equation \eqref{eq:recon_sym_eq_bl}, where 
the normal solution of \eqref{eq:recon_sym_eq_bl} satisfies $\hat{\xv}_n \in \Hc_0\cap\Sc^\perp$,
if it exists. Assuming $\nu=\cos \theta_{\max}>0$, we obtain the bound 
$\left\|\hat{\xv}_n\right\|\leq\rho\|\Sm^\perp\Tm^\perp\Sm\fv\|$ with $\rho = 1/\nu^2$, which is sharp, due to \eqref{eq:1overkappa},
and is equivalent to \eqref{eq:wellposedness-fnog} by the Pythagorean theorem. 

We us now manipulate the expression of 
\begin{equation*}
\begin{array}{ll}
\hat{\xv}_n &= \Km_\star^{-1}\left(\Sm^\perp\Tm^\perp\Sm\fv\right)\\
 &= \Km_0^\dagger\left(\Sm_0^\perp\Tm_0^\perp\Sm_0\Pm_0\fv\right)\\
 &= \Km_0^\dagger\Sm_0^\perp\Tm_0^\perp\Sm_0\Pm_0\fv,
\end{array}
\end{equation*}
where the newly introduced 
operator $\Km_0:\Hc_0\to\Hc_0$ is defined as a bounded extension by zero of $\Km_\star$ from 
$\Hc_0\cap\Sc^\perp$ to $\Hc_0$, i.e. $\Km_0u=\Km_\star u,\, \forall u \in \Hc_0\cap\Sc^\perp$
and $\Km_0u=\zerov,\, \forall u \in \Hc_0$ orthogonal to the subspace $\Hc_0\cap\Sc^\perp$.

We then have
\begin{equation*}
\begin{array}{ll}
	\|\hat{\xv}_n\| &= \|\Km_0^\dagger\Sm_0^\perp\Tm_0^\perp\Sm_0\Pm_0\fv\| \\
	&= \|\Km_0^\dagger\Sm_0^\perp\Tm_0^\perp\Tm_0^\perp\Sm_0\Pm_0\fv\| \\
			&\leq \|\Km_0^\dagger\Sm_0^\perp\Tm_0^\perp\|\|\Tm_0^\perp\Sm_0\Pm_0\fv\| \\
			&= {\|\Tm^\perp \Sm \Pm_0 \fv\|}/{\cos\theta_{\max}},
\end{array}
\end{equation*}
since $\|\Km_0^\dagger\Sm_0^\perp\Tm_0^\perp\|^2$ can be written as
\[
\begin{array}{ll}
\|\Km_0^\dagger\Sm_0^\perp\Tm_0^\perp\|^2 &=
\|\Km_0^\dagger\Sm_0^\perp\Tm_0^\perp \Tm_0^\perp\Sm_0^\perp\Km_0^\dagger\|  \\
		& = \|\left.\Km_0\right|_{\Hc_0 \cap \Sc^\perp}^\dagger \left.(\Sm_0^\perp\Tm_0^\perp)\right|_{\Hc_0 \cap \Sc^\perp} \left.\Km_0\right|_{\Hc_0 \cap \Sc^\perp}^\dagger\| \\
		& = \|\Km_\star^{-1} \|. 
\end{array}
\]
The second equality above arises from writing the orthogonal decomposition $\Hc_0 = (\Hc_0\cap\Sc) \oplus (\Hc_0\cap\Sc^\perp)$ and noting that $\Sm_0^\perp\Tm_0^\perp\Sm_0^\perp$ vanishes on $\Hc_0\cap\Sc$. 

Alternatively, we may split the product in step two above as follows
\begin{equation*}
\begin{array}{ll}
	\|\hat{\xv}_n\| &=\|\Km_0^\dagger\Sm_0^\perp\Tm_0^\perp\Sm_0\Pm_0\fv\| \\
	&= \|\Km_0^\dagger\Sm_0^\perp\Tm_0^\perp\Sm_0\Sm_0\Pm_0\fv\| \\
	&\leq \|\Km_0^\dagger\Sm_0^\perp\Tm_0^\perp\Sm_0\|\|\Sm_0\Pm_0\fv\| \\
	& =  \|\Sm \Pm_0\fv\|\tan\theta_{\max}.
\end{array}
\end{equation*}
The last equality follows from writing 
\[
\begin{array}{ll}
\Km_0^\dagger &= (\Sm_0^\perp\Tm_0^\perp\Sm_0^\perp)^\dagger \\
			&= (\Sm_0^\perp\Tm_0^\perp\Tm_0^\perp\Sm_0^\perp)^\dagger,
\end{array}
\]
hence 
\[
\begin{array}{ll}
\Km_0^\dagger\Sm_0^\perp\Tm_0^\perp &= (\Tm_0^\perp\Sm_0^\perp)^\dagger(\Sm_0^\perp\Tm_0^\perp)^\dagger\Sm_0^\perp\Tm_0^\perp\\
			&= (\Tm_0^\perp\Sm_0^\perp)^\dagger\Tm_0^\perp\\
			&= (\Sm_0^\perp\Tm_0^\perp\Tm_0^\perp\Sm_0^\perp)^\dagger\Sm_0^\perp\Tm_0^\perp\Tm_0^\perp\\
			&= (\Sm_0^\perp\Tm_0^\perp\Tm_0^\perp\Sm_0^\perp)^\dagger\Sm_0^\perp\Tm_0^\perp\\
			&= (\Tm_0^\perp\Sm_0^\perp)^\dagger.
\end{array}
\] 
Finally, the $\tan\theta_{\max}$ follows from \cite[Theorem 4.1]{Zhu2013nov} which shows that the positive singular values of the operator $(\Tm_0^\perp\Sm_0^\perp)^\dagger\Sm_0$ are equal to the tangent of the angles between the subspaces $\Tc_0$ and $\Sc_0$. 

If $\nu=\cos \theta_{\max}=0$, it remains to show that the solution of the reconstruction problem \eqref{eq:recon_sym_eq_con_zero} may fail to exist for some signal $\fv$,  i.e. the equation 
$\Km_\star\xv=-\Sm^\perp\Tm^\perp\Sm\fv$ may have no solution. Since the operator $\Km_\star$ is bounded, then it is closed and its inverse $\Km_\star^{-1}$ is closed. If $\nu = 0$, then $\Km_0^\dagger\Sm_0^\perp\Tm_0^\perp\Sm_0$ is closed and unbounded. Hence, basic results in functional analysis state that if an operator is closed and unbounded, then its range is not closed. Thus, it is a proper subset of $\Hc_0$, and consequently a solution fails to exist for some $\Pm_0\fv$.

We complete the proof by noting that the theorem claims for the normal (with the smallest norm) strictly guided reconstruction $\hat{\tv}_n$ follow from 
\eqref{eq:mixed} and \eqref{eq:tf}.
\end{proof}

We finally underline that none of the bounds \eqref{eq:wellposedness-fnog}, \eqref{eq:cosbound}, and \eqref{eq:tanbound} can be derived from the other one, i.e. not one of them is in general sharper than the other. 

%%%%%%%%%%%%%%%%%%%%%%%%%%%%%%%%%%%%%%%%%%%%%%%%%%%%%%%%%%%%%

\section{Reconstruction Error Bounds}
If the original signal satisfies $\fv \in \Tc$ and $\Sc^\perp\cap\Tc = \{\zerov\}$, then the proposed reconstruction \eqref{eq:recon_sym_eq_con_zero} perfectly recovers it. Suppose now that we obtain a reconstruction $\hat{\fv}$ of some $\fv \notin \Tc$ by solving~\eqref{eq:recon_sym_eq_con_zero}.  
An important question in this context is to bound the error $\hat{\fv} - \fv$. 

If $\Sc^\perp\cap\Tc \neq \{\zerov\}$ then the solution to reconstruction problem \eqref{eq:recon_sym_eq_con_zero} is evidently not unique. In this case, it is still possible to bound the reconstruction error, but in the factor space \hbox{$\Hc/\left(\Sc^\perp\cap\Tc\right)$}. Let $\Mm$ be an orthogonal projector onto $\left(\Sc^\perp\cap\Tc\right)^\perp = \overline{\Sc+\Tc^\perp}$, such that $\Mm = \Pm_{\Hc} - \Pm_{\Sc^\perp\cap\Tc}$. Then the norm of the error in the factor space equals the norm of a projection of the error on the subspace $\overline{\Sc+\Tc^\perp}$, representing the factor space \hbox{$\Hc/\left(\Sc^\perp\cap\Tc\right)$}. In other words, we need to bound above the quantity $\left\|M\left(\hat{\fv} - \fv\right)\right\|$, 
removing from the consideration the $\Pm_{\Sc^\perp\cap\Tc}\fv$ part of the original signal $\fv$ and ignoring 
the non-unique part $\Pm_{\Sc^\perp\cap\Tc}\hat{\fv}$ of the reconstructed signal $\hat{\fv}$.  
If the uniqueness condition holds, we have $\left(\Sc^\perp\cap\Tc\right)^\perp = \Hc$ and $\Mm\left(\hat{\fv} - \fv\right) = \hat{\fv} - \fv$. 

The unique normal solution $\hat{\fv}_n$ of problem \eqref{eq:recon_sym_eq_con_zero} simply drops the $\Pm_{\Sc^\perp\cap\Tc}\fv$ part of the original signal $\fv$, Thus, the term $\left\|\Pm_{\Sc^\perp\cap\Tc}\fv\right\|$ appears in the upper bound for $\|\hat{\fv}_n-\fv\|$, but not for $\left\|M\left(\hat{\fv} - \fv\right)\right\|$. 

The $\Pm_{\Sc^\perp\cap\Tc^\perp}\fv$ part of the original signal $\fv$ is visible neither in the sample $\Sm\fv$, nor to the guiding orthoprojector $\Tm$, thus the term $\|\Pm_{\Sc^\perp\cap\Tc^\perp}\fv\|$ is expected in any error bound. 

The following theorem gives reconstruction error bounds.
\begin{theorem}\label{thm:errbnds}
Let $\cos \theta_{\max} > 0$.
In the notation of Theorem \ref{thm:Lidentity}, let us consider the normal solution $\hat{\xv}_n$ to \eqref{eq:recon_sym_eq_bl}, giving the normal reconstruction $\hat{\fv}_n=\hat{\xv}_n+\Sm\fv$ as well as any 
reconstruction $\hat{\fv}$, obtained by solving~\eqref{eq:recon_sym_eq_con_zero}. Let $\Mm$ be the orthoprojector onto $\overline{\Sc+\Tc^\perp}$ and $\Pm_0$ be defined as in Theorem \ref{thm:Lidentity}. Then, 
\[
 \left\|M\left(\hat{\fv} - \fv\right)\right\|^2 =  \left\|\Pm_{\Sc^\perp\cap\Tc^\perp}\fv\right\|^2 +\left\|\hat{\xv}_n-\Sm^\perp\Pm_0\fv\right\|^2
\]
and 
\[
 \left\|\hat{\fv}_n - \fv\right\|^2 =  \left\|\Pm_{\Sc^\perp\cap\Tc}\fv\right\|^2+ \left\|\Pm_{\Sc^\perp\cap\Tc^\perp}\fv\right\|^2 +\left\|\hat{\xv}_n-\Sm^\perp\Pm_0\fv\right\|^2,
\]
and the following bounds hold
\begin{equation}
\left\|\hat{\xv}_n-\Sm^\perp\Pm_0\fv\right\|\leq\|\Sm^\perp \Tm^\perp\Pm_0\fv\|/\cos^{2}\theta_{\max},
\label{eq:error_bound_Me}
\end{equation}
and
\begin{equation}
 \left\|\hat{\xv}_n-\Sm^\perp\Pm_0\fv\right\|\leq\|\Tm^\perp\Pm_0\fv\|/\cos\theta_{\max}.
\label{eq:error_bound_Me_2}
\end{equation}
Bounds \eqref{eq:error_bound_Me} and \eqref{eq:error_bound_Me_2} are sharp. 
\end{theorem}

\begin{proof}
All the reconstructions are sample consistent, i.e. we have $\Sm\fv=\Sm\hat{\fv}_n=\Sm\hat{\fv}_n$. Using 
\[\Mm = \Pm_{\Hc} - \Pm_{\Sc^\perp\cap\Tc} =\Pm_0 + \Pm_{\Sc^{\perp}\cap\Tc^\perp} + \Pm_{\Sc\cap\Tc^\perp} + \Pm_{\Sc\cap\Tc},\]
we obtain the orthogonal decomposition 
\[M\left(\hat{\fv} - \fv\right)= -\Pm_{\Sc^\perp\cap\Tc^\perp}\fv+\Sm^\perp\Pm_0\left(\hat{\fv} - \fv\right).\]
Similarly, the orthogonal decomposition of the error of the normal reconstruction is
\[\hat{\fv}_n - \fv= -\Pm_{\Sc^\perp\cap\Tc}\fv-\Pm_{\Sc^\perp\cap\Tc^\perp}\fv+\Sm^\perp\Pm_0\left(\hat{\fv}_n - \fv\right).\]
In the last term of the both identities above, we have 
\[
\Sm^\perp\Pm_0\hat{\fv} = \Sm^\perp\Pm_0\hat{\fv}_n = \Sm^\perp\Pm_0\hat{\xv}_n=\hat{\xv}_n\in\Hc_0\cap\Sc^\perp.
\]
The Pythagorean theorem thus proves both identities in the statement of the theorem. 

Following algebraic transformations from the proof of Theorem~\ref{thm:Lidentity}, where $\Km_\star = \left.\left(\Sm^\perp\Tm^\perp\right)\right| _{\Sc^\perp\cap \Hc_0}$, we get
\begin{equation*}
\begin{array}{ll}
	\hat{\xv}_n - \Sm^\perp\Pm_0\fv &= - \Km_\star^{-1}\left(\Sm^\perp\Tm^\perp\Sm\Pm_0\fv\right) - \Sm^\perp\Pm_0\fv\\
				&= - \Km_\star^{-1}\left(\Sm^\perp\Tm^\perp\Sm\Pm_0\fv + \Km_\star\Sm^\perp\Pm_0\fv\right)\\
				&= - \Km_\star^{-1}\left(\Sm^\perp\Tm^\perp\Sm\Pm_0\fv + \Sm^\perp\Tm^\perp\Sm^\perp\Pm_0\fv\right)\\
				&= - \Km_\star^{-1}\left(\Sm^\perp\Tm^\perp\Pm_0\fv\right)\\
				&= - \Km_0^{\dagger}\Sm_0^\perp\Tm_0^\perp\Pm_0\fv.
\end{array}
\end{equation*}
Finally, using arguments similar to those in the proof of Theorem~\ref{thm:Lidentity}, we obtain the bounds
\begin{equation*}
\begin{array}{ll}
 \left\|\Km_\star^{-1}\left(\Sm^\perp\Tm^\perp\Pm_0\fv\right)\right\| &\leq \left\|\Km_\star^{-1}\right\|\left\|\Sm^\perp\Tm^\perp\Pm_0\fv\right\|\\
 &={\|\Sm^\perp\Tm^\perp\Pm_0 \fv\|}/{\cos^2\theta_{\max}}
\end{array},
\end{equation*}
\begin{equation*}
\begin{array}{ll}
 \left\|\Km_0^{\dagger}\Sm_0^\perp\Tm_0^\perp\Pm_0\fv\right\| &\leq \left\|\Km_0^{\dagger}\Sm_0^\perp\Tm_0^\perp\right\|\left\|\Tm_0^\perp\Pm_0\fv\right\| \\
 &={\|\Tm^\perp\Pm_0 \fv\|}/{\cos\theta_{\max}}
\end{array},
\end{equation*}
which complete the proof. 
The sharpness of bounds \eqref{eq:error_bound_Me} and \eqref{eq:error_bound_Me_2} is shown in Sec. \ref{sec:moreexamples}.
\end{proof}
The error bounds  of Theorem~\ref{thm:errbnds} based on \eqref{eq:error_bound_Me_2}, improve and extend to the most general case the bound $\|\Tm^\perp \fv\|/{\cos\theta_{\max}}$ obtained with the consistent reconstruction method  presented in~\cite{Unser-TSP-94, Eldar-TSP-06}, dropping all unnecessary assumptions on the sampling and guiding subspaces made in ~\cite{Unser-TSP-94, Eldar-TSP-06}. The error bounds of Theorem~\ref{thm:errbnds} based on \eqref{eq:error_bound_Me} are new. Neither of the bounds \eqref{eq:error_bound_Me} and \eqref{eq:error_bound_Me_2} can be derived from the other one.

%%%%%%%%%%%%%%%%%%%%%%%%%%%%%%%%%%%%%%%%%%%%%%%%%%%%%%%%%%%%%
\section{Alternative Equivalent Formulations}\label{s:GenRec} 
We assume  $\Sc^\perp\cap\Tc=\{\zerov\}$ for uniqueness in this section.
\subsection{Quotient Space Reconstruction} 

The oblique projector onto the subspace  $\Tc$ along the subspace $\Sc^\perp$, we denote by $\Pm_{\Tc\perp\Sc}$, 
is conventionally used to compute the reconstructed signal constrained to~$\Tc$.
The  existence of $\Pm_{\Tc\perp\Sc}$ relies on
the traditional assumption $\Sc^\perp+\Tc = \Hc$, made in \cite{Unser-TSP-94, Eldar-JFA-03},
which is equivalent to $\theta_{\max}<\pi/2$ and 
\begin{equation}
\{\zerov\}= \left\{\Sc^\perp+\Tc\right\}^\perp=\Tc^\perp\cap\Sc=\Sc\cap\Tc^\perp,
\label{eq:nonzero_assumption}
\end{equation}
The spectral norm of the oblique projector $\Pm_{\Tc\perp\Sc}$, determining stability of the reconstruction $\Pm_{\Tc\perp\Sc}\,\fv$  is 
equal (cf. \cite[Eq.~(6.2), attributed to Del Pasqua, 1955]{Szyld-06}) in this case to $1/\gamma(\Tc,\Sc^\perp)=1/\cos\theta_{\max}$; see Theorem \ref{thm:kappa} and~\eqref{eq:projector-norm-min-gap1}.

Oversampling can make the intersection $\Sc\cap\Tc^\perp$ nontrivial, i.e.\ $\Sc\cap\Tc^\perp\neq\{\zerov\}$, so  
there is a nontrivial orthogonal decomposition $\Hc=\overline{\left(\Sc^\perp+\Tc\right)}\oplus\left(\Sc\cap\Tc^\perp\right)$. 
In~this case, 
the oblique projector $\Pm_{\Tc\perp\Sc}$ cannot be defined in the whole space $\Hc$, but
it can be instead defined within the subspace $\Sc^\perp+\Tc\subseteq\overline{\Sc^\perp+\Tc}=\Hc\ominus\left(\Sc\cap\Tc^\perp\right)$, where the latter represents the quotient space $\Hc/\{\Sc \cap \Tc^\perp\}$.

        \begin{figure}
                \includegraphics[width=0.95\linewidth]{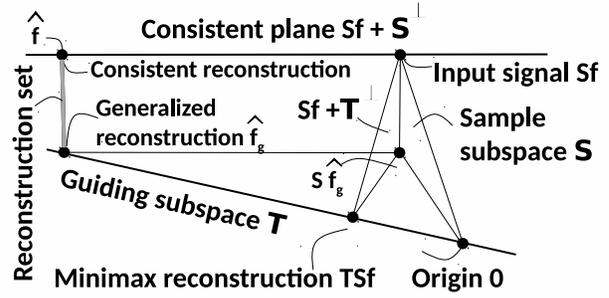}
                \vspace{-0.8in}
                \caption{3D example in details}\label{fig:3}
        \end{figure}

A specific reconstruction algorithm, implementing this idea, as illustrated in Fig. \ref{fig:3}, can be as follows. Let $\Pm_{\Sc\cap\Tc^\perp}\fv$ be an orthogonal projection of the original signal $\fv$
on the subspace $\Sc\cap\Tc^\perp$, then the difference $\fv-\Pm_{\Sc\cap\Tc^\perp}\fv\in\overline{\Sc^\perp+\Tc}$ is a reduced signal, representing the
original signal in the   quotient space $\Hc/\{\Sc \cap \Tc^\perp\}$.
The~oblique projector $\Pm_{\Tc\perp\Sc}$ onto the subspace  $\Tc$ along the subspace $\Sc^\perp$, defined within $\overline{\Sc^\perp+\Tc}$, 
acting on the reduced signal $\fv-\Pm_{\Sc\cap\Tc^\perp}\fv\in\overline{\Sc^\perp+\Tc}$, gives   
\begin{equation}
\hat{\fv}_g\equiv\Pm_{\Tc\perp\Sc}\left(\fv-\Pm_{\Sc\cap\Tc^\perp}\fv\right)\in\Tc\cap\overline{\Sc^\perp+\Tc},
\label{eq:fg}
\end{equation}
which is sample consistent with $\fv-\Pm_{\Sc\cap\Tc^\perp}\fv$. 
We call $\hat{\fv}_g$ a \emph{generalized reconstruction}~of~$\fv$, 
since it is the same as the generalized reconstruction in~\cite{Adcock:2013:BCR}; see the next section.

The sample consistent with $\fv$ reconstructed signal is finally obtained by adding the subtracted term $\Pm_{\Sc\cap\Tc^\perp}\fv$ back, i.e. by  
$\hat{\fv}_g+\Pm_{\Sc\cap\Tc^\perp}\fv.$
We next prove that this quotient space reconstruction method results in our previously defined sample consistent reconstructed signal $\hat{\fv} = \hat{\tv} + \Pm_{\Sc\cap\Tc^\perp}\fv$ in \eqref{eq:tf}, and that $\hat{\fv}_g=\hat{\tv}$, i.e., we obtain the same reconstructions as before.

\begin{theorem}
Let $\Sc^\perp\cap\Tc=\{\zerov\}$. The quotient space reconstruction method is equivalent to and gives the same reconstructed signal  $\hat\fv$ %(modulo  $\Sc^\perp\cap\Tc$) 
as solving problem~\eqref{eq:recon_problem}, while $\hat{\fv}_g=\hat{\tv}$.
\label{thm:QSCR_equivalence}
\end{theorem}
\begin{proof}
We first discuss that the conditions of the reconstructed signal uniqueness are the same in both approaches. 
As in \cite[Proposition 2]{Hirabayashi-Unser-2007}, the assumption $\Sc^\perp\cap\Tc=\{\zerov\}$ is necessary and sufficient for the quotient space constrained reconstruction uniqueness, since the subspaces $\Sc^\perp\cap\Tc$ and $\Sc\cap\Tc^\perp$ are orthogonal, thus the 
the former is not affected by vanishing of the latter in the quotient space $\Hc/\{\Sc \cap \Tc^\perp\}$.
All our arguments of Sec.~\ref{s:uniquness} are  applicable as well for the quotient space constrained reconstruction, 
and can be viewed as extensions of the arguments from \cite{Hirabayashi-Unser-2007} to the quotient space $\Hc/\{\Sc \cap \Tc^\perp\}$.
%We simply assume $\Sc^\perp\cap\Tc=\{\zerov\}$ here, to avoid dealing with multi-valued operators.

Second, we compare the conditions of the reconstructed signal  $\hat\fv$ existence and continuous dependence on the original signal  $\fv$. The subspace $\Sc^\perp+\Tc$, not necessarily closed, is a domain of the oblique projector $\Pm_{\Tc\perp\Sc}$, but the reduced signal $\fv-\Pm_{\Sc\cap\Tc^\perp}\fv$, which we need to apply $\Pm_{\Tc\perp\Sc}$ to, can be arbitrary in the closure $\overline{\Sc^\perp+\Tc}$. 

Thus, it is necessary and sufficient for the existence of the reconstructed signal, using the quotient space constrained reconstruction, for an arbitrary original signal $\fv\in\Hc$ that  $\Sc^\perp+\Tc=\overline{\Sc^\perp+\Tc}.$ 
A sum of two closed subspaces is closed iff the \emph{minimal gap}~\cite[Sec.~IV-4]{kato} $\gamma$ between them is positive. 
In our case, $\Sc^\perp+\Tc=\overline{\Sc^\perp+\Tc}$ iff $\gamma\left(\Tc,\Sc^\perp\right)>0,$ where 
$\gamma(\Tc,\Sc^\perp)=\cos\theta_{\max}$ by Theorem~\ref{thm:kappa} and identities~\eqref{eq:projector-norm-min-gap1}.

Moreover, by definition \eqref{eq:def_mingap}, the minimal gap  $\gamma(\Sc,\Tc^\perp)$
is essentially defined in a quotient space $\Hc/\{\Sc \cap \Tc^\perp\}$, factoring out the intersection $\Sc\cap\Tc^\perp$, 
if it is nontrivial, $\Sc\cap\Tc^\perp\neq\{\zerov\}$, which we allow. 
This implies that the formula $1/\gamma(\Tc,\Sc^\perp)$ from \cite[Equation (6.2)]{Szyld-06} of the spectral norm of the oblique projector $\Pm_{\Tc\perp\Sc}$, defined within the subspace $\Sc^\perp+\Tc$, remains valid even if $\Sc\cap\Tc^\perp\neq\{\zerov\}$.

% Following the definition of the reconstructed signal $\hat\fv=\Pm_{\Tc\perp\Sc} \left(\fv-\Pm_{\Sc\cap\Tc^\perp}\fv\right)+\Pm_{\Sc\cap\Tc^\perp}\fv$ via the quotient space reconstruction, 
% and taking into account orthogonality of signals $\fv-\Pm_{\Sc\cap\Tc^\perp}\fv$ and $\Pm_{\Sc\cap\Tc^\perp}\fv$, we obtain the following stability bound
% \begin{equation}
% \|\hat\fv\|^2\leq \frac {\|\fv-\Pm_{\Sc\cap\Tc^\perp}\fv\|^2}{\cos^2\theta_{\max}}+\|\Pm_{\Sc\cap\Tc^\perp}\fv\|^2.
% \label{eq:QSCR_stability}
% \end{equation}

We conclude that the assumption $\theta_{\max}<\pi/2$ is necessary and sufficient for  existence of the reconstructed signal  $\hat\fv$ 
using the quotient space reconstruction, for an arbitrary original signal $\fv\in\Hc$, as well as it guarantees the stability of the reconstruction.  
Comparing this assumption to those of Theorem~\ref{thm:existence}, while taking into account 
\eqref{eq:rho-projector-norm} and \eqref{eq:projector-norm-min-gap}, we also conclude that $\theta_{\max}<\pi/2$ is necessary and sufficient for  existence of the reconstructed signal  $\hat\fv$ in both approaches, the quotient space reconstruction and minimization in~\eqref{eq:recon_problem}. 

It remains to prove that both approaches also give the same reconstructed signal  $\hat\fv$, if it exists. Let   $\hat\fv$ be the reconstructed signal obtained by the quotient space constrained reconstruction. We analyze the square of the function 
$\|\Tm^\perp\hat{\fv}\|$
minimized in \eqref{eq:recon_problem}, using the following identities,
\begin{align}
\left\|\Tm^\perp\hat{\fv}\right\|^2&=\left\|\Tm^\perp\Pm_{\Sc\cap\Tc^\perp}\fv\right\|^2+\left\|\Tm^\perp\left(\hat{\fv}-\Pm_{\Sc\cap\Tc^\perp}\fv\right)\right\|^2 \nonumber \\
&=\left\|\Pm_{\Sc\cap\Tc^\perp}\fv\right\|^2+\left\|\Tm^\perp\left(\hat{\fv}-\Pm_{\Sc\cap\Tc^\perp}\fv\right)\right\|^2.
\label{eq:constant_term}
\end{align}
The first identity in \eqref{eq:constant_term} holds, because both vector sums $\Pm_{\Sc\cap\Tc^\perp}\fv+\left(\fv-\Pm_{\Sc\cap\Tc^\perp}\fv\right)$ and $\Pm_{\Sc\cap\Tc^\perp}\fv+\left(\hat{\fv}-\Pm_{\Sc\cap\Tc^\perp}\fv\right)$ are orthogonal, where 
$\Pm_{\Sc\cap\Tc^\perp}\fv\in\Sc\cap\Tc^\perp$, while also $\fv-\Pm_{\Sc\cap\Tc^\perp}\fv\in\overline{\Sc^\perp+\Tc}$ and $\hat{\fv}-\Pm_{\Sc\cap\Tc^\perp}\fv\in\overline{\Sc^\perp+\Tc},$
consistently with the orthogonal decomposition of the Hilbert space $\Hc=\left(\Sc\cap\Tc^\perp\right)\oplus\overline{\left(\Sc^\perp+\Tc\right)}$, 
by construction of the quotient space reconstruction. Moreover, the subspace $\Sc\cap\Tc^\perp$ is trivially invariant 
with respect to the orthogonal projector $\Tm^\perp$, consequently, its orthogonal complement $\overline{\left(\Sc^\perp+\Tc\right)}$
is also  $\Tm^\perp$-invariant, as can be directly verified. Therefore, we conclude that  the sum
$\Tm^\perp\Pm_{\Sc\cap\Tc^\perp}\fv+\Tm^\perp\left(\hat{\fv}-\Pm_{\Sc\cap\Tc^\perp}\fv\right)$ in \eqref{eq:constant_term} is
also orthogonal, and the Pythagorean theorem is applicable. The second identity  in \eqref{eq:constant_term} trivially follows from 
$\Tm^\perp\Pm_{\Sc\cap\Tc^\perp}=\Pm_{\Sc\cap\Tc^\perp}$ since $\Sc\cap\Tc^\perp\subseteq\Tc^\perp$.

We observe that in identity \eqref{eq:constant_term}, the first term in the sums is a constant, not changing in minimization~\eqref{eq:recon_problem}, since 
 $\Pm_{\Sc\cap\Tc^\perp}\fv$ is simply the orthogonal projection of the original signal $\fv$ on the subspace $\Sc\cap\Tc^\perp$. 
 We now show that the second term vanishes on the minimizer $\hat\fv$. Indeed, we have by the definition of the quotient space reconstruction that 
  $\hat\fv=\hat{\fv}_g+\Pm_{\Sc\cap\Tc^\perp}\fv$, where by \eqref{eq:fg}
  $\hat{\fv}_g=\Pm_{\Tc\perp\Sc}\left(\fv-\Pm_{\Sc\cap\Tc^\perp}\fv\right)\in\Tc\cap\overline{\Sc^\perp+\Tc}$ 
is sample consistent with $\fv-\Pm_{\Sc\cap\Tc^\perp}\fv,$
i.e. the following holds,
$\Sm\hat{\fv}_g=\Sm\left(\fv-\Pm_{\Sc\cap\Tc^\perp}\fv\right)=\Sm\fv-\Pm_{\Sc\cap\Tc^\perp}\fv.$
We conclude that $\Tm^\perp\hat{\fv}_g=\zerov $ and the orthogonal sum $\Tm^\perp\fv=\Pm_{\Sc\cap\Tc^\perp}\fv+\Tm^\perp\hat{\fv}_g$
both have the smallest possible norms, while $\Sm\hat{\fv}=\Sm\hat{\fv}_g+\Pm_{\Sc\cap\Tc^\perp}\fv=\Sm\fv$, i.e. the reconstructed signal  $\hat\fv$ obtained by the quotient space reconstruction is a valid minimizer in \eqref{eq:recon_problem}. 

Finally, comparing the identity $\hat\fv=\hat{\fv}_g+\Pm_{\Sc\cap\Tc^\perp}\fv$ to \eqref{eq:tf}, i.e.,\ $\hat{\fv} = \hat{\tv} + \Pm_{\Sc\cap\Tc^\perp}\fv$,   
immediately implies that $\hat{\fv}_g=\hat{\tv}$. 
\end{proof}

\subsection{Comparison with Generalized Reconstruction}\label{sec:cgr}
An equivalent to the quotient space approach is proposed in~\cite{Adcock:2013:BCR}, where the~oblique projector $\Pm_{\Tc\perp\Sc}$ onto the subspace  $\Tc$ along the subspace $\Sc^\perp$, defined within $\overline{\Sc^\perp+\Tc}$, is substituted with the oblique projector $\Pm_{\Tc\perp\Sm\Tc}$ onto the subspace $\Tc$ along $\left(\Sm\Tc\right)^{\perp}$, %denoted in \cite{Adcock:2013:BCR,Berger-2013} by $\Pm_{\Tc, (\Sm^\Tc)^\perp}$, 
resulting in the same \emph{generalized reconstruction} \[\hat{\fv}_g=\Pm_{\Tc\perp\Sm\Tc}\fv=\Pm_{\Tc\perp\Sc}\left(\fv-\Pm_{\Sc\cap\Tc^\perp}\fv\right).\]
Indeed, the traditional assumption $\Sc^\perp+\Tc = \Hc$ of \cite{Unser-TSP-94, Eldar-JFA-03} made for the oblique projector $\Pm_{\Tc\perp\Sc}$,
onto the subspace  $\Tc$ along the subspace $\Sc^\perp$, transforms here into the assumption $\left(\Sm\Tc\right)^\perp+\Tc = \Hc$ for the  oblique projector $\Pm_{\Tc\perp\Sm\Tc}$, which is equivalent to $\theta_{\max}<\pi/2$ and, by analogy with \eqref{eq:nonzero_assumption}, 
\begin{equation*}
\left\{\left(\Sm\Tc\right)^\perp+\Tc\right\}^\perp=\overline{\Sm\Tc}\cap\Tc^\perp=\Sc\cap\left(\Sc\cap\Tc^\perp\right)^\perp\cap\Tc^\perp=\{\zerov\},
\label{eq:nonzero_assumptiong}
\end{equation*}
but where the latter is automatically satisfied, in contrast to assumption \eqref{eq:nonzero_assumption}. 
The spectral norm of the oblique projector $\Pm_{\Tc\perp\Sm\Tc}$ is equal 
(cf. \cite[Lemma 4.4]{Adcock:2013:BCR}) in this case to $1/\gamma(\Tc,(\Sm\Tc)^\perp)=1/\cos\theta_{\max}$; see again Sec.~\ref{sec:es}.
We thus conclude that  $\Pm_{\Tc\perp\Sm\Tc}\fv=\hat{\tv}$.

By Theorem  \ref{thm:rcsize}, the reconstruction error is
\[\left\| \hat{\fv}_g - \fv\right\|^2 = \left\| \hat{\fv} - \fv\right\|^2 + \left\|\Pm_{\Sc\cap\Tc^\perp}\fv\right\|^2,\]
where $\left\| \hat{\fv} - \fv\right\|$ is bounded in Theorem~\ref{thm:errbnds}, extending and improving 
the bound $\|\hat{\fv}_g - \fv\| \leq \|\Tm^{\perp}\fv\|/\cos \theta_{\max}$ of  \cite{Adcock:2013:BCR}.

\section{Comparison with Regularized Reconstruction}\label{sec:regularization}

Regularization-based methods, suggested in \cite{Narang-GlobalSIP-13}, in our notation can be formulated using the following unconstrained quadratic minimization problem 
\begin{equation}
\inf_{\hat{\fv}_\rho} \quad \left\|\Sm\hat{\fv} _\rho-\Sm\fv\right\|^2 +\rho\left\|\Hm\hat{\fv}_\rho\right\|^2 ,\quad\rho>0,
\label{eq:reg_problem_orig}
\end{equation}
where the operator $\Hm$ is interpreted as a filter, e.g., it may approximate our $\Tm^\perp$,
in which case problem \eqref{eq:reg_problem_orig} approximates 
\begin{equation}
\inf_{\hat{\fv}_\rho} \quad \left\|\Sm\hat{\fv}_\rho-\Sm\fv\right\|^2 +\rho\left\|\left(\hat{\fv}_\rho-\Tm\hat{\fv}_\rho\right)\right\|^2 .
\label{eq:reg_problem}
\end{equation}
Problem  \eqref{eq:reg_problem} can be viewed as a relaxation of our \eqref{eq:recon_problem}. 

The authors of  \cite{Narang-GlobalSIP-13}  assume that there exists a unique intersection of the sample-consistent closed plane $\Sm\fv+\Sc^\perp$ and the guiding closed subspace $\Tc$ and claim, without proof, that the minimizer $\hat{\fv}_\rho$ of \eqref{eq:reg_problem} and $\rho\to\infty$, approximates this intersection. We prove below a surprising result that, under the assumption of the unique intersection, the minimizer $\hat{\fv}_\rho$ of \eqref{eq:reg_problem} is equal to this intersection, for any $\rho>0$, i.e.\ $\hat{\fv}_\rho$ does not actually depend on~$\rho$. This will be a trivial consequence of an even more stunning result that the set of all solutions of \eqref{eq:reg_problem} for varying $\rho>0$ in general is nothing but our reconstruction set with removed end points, belonging the sample-consistent plane $\Sm\fv+\Sc^\perp$ and the guiding subspace $\Tc$. 

\begin{theorem}\label{thm:reg}
Let our reconstruction set be given by formula $\hat{\fv}_\alpha=\alpha\hat{\fv}+(1-\alpha)\Tm\hat{\fv}$, where $0\leq\alpha\leq1,$ and $\hat\fv$ solves
 \eqref{eq:recon_problem}.
Then $\hat{\fv}_\alpha$ solves 
problem  \eqref{eq:reg_problem} with $\rho=(1-\alpha)/\alpha$.
\end{theorem}
\begin{proof}
On the one hand, minimization problem \eqref{eq:reg_problem} is equivalent to the following linear equation
$\left(\Sm+\rho\Tm^\perp\right)\hat{\fv}_\rho=\Sm\fv.$
On the other hand, the consistent reconstruction $\hat\fv$ solves \eqref{eq:recon_sym_eq_con_zero}, i.e.\
$\Sm^\perp\Tm^\perp \hat{\fv} = \zerov$ and $\Sm\hat{\fv} = \Sm\fv$. Taking $\rho=(1-\alpha)/\alpha$ and substituting
$\hat{\fv}_\alpha$ 
 for  $\hat{\fv}_\rho$, we obtain by elementary calculations
 \[
 \left(\Sm+\frac{1-\alpha}{\alpha}\Tm^\perp\right)\left(\alpha\hat{\fv}+(1-\alpha)\Tm\hat{\fv}\right) = \Sm\fv
 \]
using properties of $\Sm$ and $\Tm$ as projectors.
\end{proof}

Theorem~\ref{thm:reg} can be extended to the case, where the filter $\Hm$ approximates the orthoprojector $\Tm^\perp$, but may 
fail for more general filters, e.g.,\ for some practically important in graph-based setup polynomial \cite{Chebyshev:icme2014} 
and nonlinear \cite{KnyazevM15e} filters. 

If there exists a unique intersection of the sample-consistent closed plane $\Sm\fv+\Sc^\perp$ and the guiding closed subspace $\Tc$, as assumed in \cite{Narang-GlobalSIP-13}, then the intersection is $\hat{\fv}=\Tm\hat{\fv}$ and our reconstruction set is thus trivially reduced to this single element 
$\hat{\fv}=\Tm\hat{\fv}$, so, by Theorem~\ref{thm:reg}, the minimizer  $\hat{\fv}_\rho$ in \eqref{eq:reg_problem} is simply
 $\hat{\fv}_\rho=\hat{\fv}=\Tm\hat{\fv}$, no matter what the value of $\rho>0$ is. 

If our reconstruction set is nontrivial, we can intentionally move the reconstructed signal away from the sample-consistent reconstruction plane $\Sm\fv+ \Sc^\perp$ toward the guiding subspace $\Tc$, e.g.,\ assuming that the sampling procedure is noisy.  
The sum in  \eqref{eq:reg_problem} penalizes for moving the reconstructed signal away from the sample-consistent reconstruction plane $\Sm\fv+\Sc^\perp$ and from the guiding subspace~$\Tc$.  
A~specific value of the~regularization parameter need to be chosen {\it a~priori}, e.g.,\ according to a noise level, if 
problem  \eqref{eq:reg_problem} is solved directly. 

Theorem \ref{thm:reg} allows us to choose the value $\rho=(1-\alpha)/\alpha$ {\it a~posteriori}, after determining the 
 reconstruction set, as well as to try a variety of choices at no extra costs. For example, let the reconstruction set be the closed interval with the end points $\hat{\fv}\in\Sm\fv+\Sc^\perp$ and $\hat{\tv}=\Tm\hat{\fv}\in\Tc$
If we trust that the sample-consistent closed plane $\Sm\fv+\Sc^\perp$ is actually accurate, we can choose our reconstruction to be sample consistent,  $\hat{\fv}\in\Sm\fv+\Sc^\perp$ that solves, e.g., minimization problem \eqref{eq:recon_problem}. 
If there is noise in sample measurements, we may decide to trust the guiding closed subspace $\Tc$ more than the sample $\Sm\fv$ 
and choose as our output reconstruction a convex linear combination $\alpha\hat{\fv}+(1-\alpha)\Tm\hat{\fv}$ within the reconstruction set, 
where $0<\alpha<1,$ or use the extreme choice $\alpha=0$ that results in the strictly  guided reconstruction $\Tm\hat{\fv}$ of~\cite{Adcock:2013:BCR}.

Specifically, for reconstruction with noisy or otherwise inaccurate samples, where $\Sm\fv$ is substituted by $\Sm\fv+\nv$, 
and $\nv$ represents a deviation from the true sample $\Sm\fv$, we can select 
\begin{equation}
1-\alpha=\frac{\|\nv\|}{\|\hat{\fv}-\Tm\hat{\fv}\|}. \label{eqn:alpha_selection}
\end{equation}
In \eqref{eqn:alpha_selection}, the numerator $\|\nv\|$ may be known from specifications of a sampling sensor.
The denominator $\|\hat{\fv}-\Tm\hat{\fv}\|$ is easily computable directly.

In the next section, we present conjugate gradient based methods to solve the proposed reconstruction problem.

%%%%%%%%
\section{Iterative Reconstruction Algorithms}\label{sec:CG}
An iterative algorithm based on projection on convex sets (POCS) for
reconstructing a band-limited graph signal is presented in~\cite{Narang-GlobalSIP-13}. Starting with an initial guess, at each iteration the algorithm projects the signal on $\Tc$ and then resets the signal samples on $\Sc$ to the given samples. The POCS method can be interpreted as a Richardson iterative method for solving~\eqref{eq:recon_sym_eq_bl},
\begin{equation}\label{POCS}
\xv_m = \left(\mathbf{I} - \Km\right)\xv_{m-1} + \bv.
\end{equation}
When $\Km = \left.\left( \Sm^\perp \Tm^\perp\right) \right|_{\Sc^\perp}$ and $\bv = - \Sm^\perp \Tm^\perp \Sm\fv$, as in the present context, this iteration becomes
\begin{equation*}
\xv_m = \Sm^\perp \Tm \fv_{m-1}, \text{ where } \fv_{m-1} = \Sm\fv + \xv_{m-1},
\end{equation*} 
which is POCS method in~\cite{Narang-GlobalSIP-13}.

Conjugate gradient (CG) is the optimal iterative method for solving linear systems $\Km\xv = \bv$, if $\Km$ is a linear self-adjoint non-negative operator with bounded (pseudo)inverse. The basics of CG are reviewed in Appendix~\ref{app:CG}. We would like to use CG to solve \eqref{eq:recon_sym_eq_bl}. The difficulty  lies in the fact that $\Sm^\perp \Tm^\perp$ is not self-adjoint in general. However, as shown below, the restriction $\Km$ of $\Sm^\perp \Tm^\perp$ to its invariant subspace $\Sc^\perp$ is self-adjoint and positive semi-definite. 
\begin{proposition}\label{p}
Let $\Sm^\perp$ and $\Tm^\perp$ be two orthoprojectors. Then the operator $K=\left.\left(\Sm^\perp \Tm^\perp\right)\right|_{\Sc^\perp}0$ is self-adjoint and the operator lower and upper bounds $0\leq\Km\leq \mathbf{I}$ hold. 
\end{proposition}
\begin{proof}
For any $\uv \in \Sc^\perp$, $\Sm^\perp \uv = \uv$. Since $\Sm^\perp$ is an orthoprojector, it is self-adjoint. Thus, for any $\uv, \vv \in \Sc^\perp$ we have 
$\Braket{\Sm^\perp \Tm^\perp \uv, \vv} = \Braket{\Tm^\perp\uv, \Sm^\perp \vv} 
= \Braket{\Tm^\perp\uv, \vv} = \Braket{\uv, \Tm^\perp\vv} = \Braket{\Sm^\perp\uv, \Tm^\perp\vv} = \Braket{\uv, \Sm^\perp\Tm^\perp\vv}$.

Taking above $\uv=\vv$ proves the both operator bounds, since  $0\leq\Braket{\uv, \Km\uv}=\Braket{\Tm^\perp\uv, \uv}=\Braket{\Tm^\perp\uv, \Tm^\perp\uv}\leq\Braket{\uv,\uv}.$
\end{proof}
We can use CG for solving \eqref{eq:recon_sym_eq_bl} thanks to Proposition \ref{p}. When the solution is not unique, CG converges to the unique normal solution $\hat{\xv}$ (with minimum norm), but it needs to be initialized with some $\xv_0 \in \Sc^\perp$. Since CG is the optimal iterative method, it computes the most efficient signal reconstruction. The solution $\xv_m$ after $m$ iterations of CG satisfies
\begin{equation*}
\xv_m = \argmin_{\xv \in \Sc^\perp \cap \bar{\Kc}_m} \Braket{\left(\xv+\Sm\fv\right), \Tm^\perp \left(\xv+\Sm\fv\right)},
\end{equation*} 
where $\bar{\Kc}_m$ is a plane defined as in \eqref{eq:krylov_hp} with $\Km = \left(\Sm^\perp \Tm^\perp\right)\big| _{\Sc^\perp}$ and $\bv = - \Sm^\perp \Tm^\perp \Sm\fv$.

We note that super-resolution using preconditioned CG has been suggested in \cite{yang2002super} for image reconstruction from multiple low-resolution frames in a video sequence, assuming explicitly known imaging models. 

\subsection*{Convergence Analysis}

Convergence speed of iterative methods for solving the linear system $\Km \xv = \bv$ depends on a condition number $\kappa$.
Since our operator  $\Km = \left(\Sm^\perp \Tm^\perp\right)\big| _{\Sc^\perp}$ is self-adjoint and positive semi-definite, but has 
a possibly non-trivial null-space $\Sc^\perp\cap\Tc$, special considerations apply; see, e.g., \cite{bk94} and references there. As in the proof of Theorem \ref{thm:Lidentity}, we have
$\bv = - \Sm^\perp \Tm^\perp \Sm\fv\in\Sm^\perp\Tm^\perp \Sc$  and can substitute $\Km_\star$, defined as the restriction of $\Km$ to $\Hc_0\cap\Sc^\perp$, for $\Km$ in equation \eqref{eq:recon_sym_eq_bl}. Even though in practical implementations of the Richardson iterative method~\eqref{POCS}, as well as CG, one simply multiplies vectors by $\Sm^\perp \Tm$, the convergence analysis can be 
based on $\Km_\star \xv = \bv$, as soon as all the iterative errors $\xv_m-\xv^*$ stay within the subspace $\Hc_0\cap\Sc^\perp$; see again \cite{bk94}. The latter can be easily achieved by choosing simply $\xv_0 = \zerov$ to initiate the iterative method. 
Then the convergence speed is determined by the spectral condition number $\kappa$ of the operator $\Km_\star$. 

From the proof of Theorem \ref{thm:Lidentity}, the smallest point $\sigma_{{\min}}$ of the spectrum of $\Km_\star$ is $\nu^2=\cos^2 \theta_{\max}$ defined by \eqref{eq:rho-projector-norm} and characterized in Theorem~\ref{thm:kappa}, while the largest point $\sigma_{{\min}}$ is bounded above by one. 
%Let $\Sigma(\Km)$ be the spectrum of $\Km$ with $\sigma_{{\min}} = \inf \{\Sigma(\Km) \setminus \{0\}\}\geq0$ and $\sigma_{{\max}} = \sup \Sigma(\Km)= \max \Sigma(\Km)\leq1$ by Proposition \ref{p}, 
%since the spectrum of a bounded  self-adjoint operator is a closed set on the real line. 
%Using the definition of angles and the fact that cosine is a decreasing function on $[0,\pi/2]$, we can write $\sigma_{\min}=\cos^2 \left(\sup \{\Theta(\Sc^\perp, \Tc^\perp)\setminus \{\pi/2\}\}\right)$. From \eqref{eq:angle-relations}, we get 
%\[\sup \{\Theta(\Sc^\perp, \Tc^\perp)\setminus \{\pi/2\}\} = \sup \{\Theta(\Sc, \Tc)\setminus \{\pi/2\}\} = \theta_{\max}.\] 
%Since $\Km$ in this case is a restriction of a product of two projectors, we have $\sigma_{\max} \leq 1$. 
Hence, 
$\|\mathbf{I}-\Km_\star\|\leq  \cos^2 \theta_{\max}$ 
and the spectral condition number ${\kappa}$ of $\Km_\star$ can be bounded as
\begin{equation*}
{\kappa} = \frac{\sigma_{\max}}{\sigma_{\min}} 
\leq \frac{1}{\cos^2 \theta_{\max}}.
\end{equation*}
%Thus, to minimize the condition number, the angle $\theta_{\max}$ needs to be as small as possible. This observation can help one choose good guiding and sampling subspaces. In the ideal case $\theta_{\max}=0$, one subspace is simply a subset of the other. 

If the iteration is initialized with $\xv_0 = \zerov$, then the relative error in the solution $\xv_m$ after $m$ iterations of \eqref{POCS} satisfies
\begin{equation*}
\frac{\|\xv_m-\xv^*\|}{\|\xv^*\|} \leq \left(1-\cos^2 \theta_{\max}\right)^m,
\end{equation*}
where $\xv^*$ is the actual normal solution, since
$\|\xv_m-\xv^*\|=\|(\mathbf{I}-\Km_\star)\left(\xv_{m-1}-\xv^*\right)\|\leq \|\mathbf{I}-\Km_\star\|\|\xv_{m-1}-\xv^*\|.$

If CG is initialized with $\xv_0 = \zerov$, then it can be shown that the relative error in the solution $\xv_m$ obtained after $m$ CG iterations satisfies
\begin{equation*}
\frac{\|\xv_m-\xv^*\|_{\Km}}{\|\xv^*\|_{\Km}} \leq 2\left(\frac{\sqrt{\kappa}-1}{\sqrt{\kappa}+1}\right)^m
\leq 2\left(\frac{1-{\cos\,\theta_{\max}}}{1+{\cos\,\theta_{\max}}}\right)^m
;
\end{equation*}  
see, e.g.,\ \cite{bk94,daniel1967conjugate}.
The relative error with POCS and CG decreases geometrically at a rate that depends on $\theta_{\max}$, where 
CG is always faster than POCS. The acceleration provided by CG becomes more pronounced when $\Km_\star$ is ill-conditioned. i.e. $\theta_{\max}$
is not small enough.

%%%%%%%%%%%%%%%%%%%%%%%%%%%%%%%%%%%%%%%%%%%%%%%%%%%%%%%%%

\section{Reconstruction of Bandlimited Graph Signals}\label{sec:gs} 
\subsection{Notation and Preliminaries}
An undirected, weighted graph $G = \left(\Vc, \Ec\right)$ is a collection of nodes (or vertices) $\Vc = \{1,2, \ldots, n\}$ which are connected to each other by a set of edges (or links) $\Ec = \{\left(i,j,w_{ij}\right)\}_{i,j \in \Vc}$. $\left(i,j,w_{ij}\right)$ denotes an edge between nodes $i$ and $j$ with weight $w_{ij}$. The adjacency matrix $\Tm$ of the graph is a $n \times n$ matrix with entries $\Tm(i,j) = w_{ij}$. The degree $d_i$ of node $i$ is the sum of the weights of edges incident on $i$, i.e. $d_i = \sum_j w_{ij}$. The degree matrix is a diagonal matrix $\Km = \text{diag}\{d_1, d_2, \ldots, d_n\}$. The combinatorial Laplacian matrix of the graph is defined as $\Lm = \Km - \Tm$. 
We use the normalized form of the Laplacian given by $\Lcb = \Km^{-1/2}\Lm \Km^{-1/2}$. It is a symmetric positive semi-definite matrix and has a set of real eigenvalues $0 = \lambda_1 \leq \lambda_2 \leq \dots \leq \lambda_n \leq 2$ and a corresponding orthogonal set of eigenvectors denoted as $\Um = \{\uv_1,\uv_2,\dots,\uv_n\}$~\cite{Chung-AMS-97}.
A graph signal is a function $f:\Vc \rightarrow \mathbb{R}$ defined on the nodes of the graph, such that $f(i)$ is the value of the signal at node $i$. Thus, a graph signal can also be represented as a vector $\fv$ in $\mathbb{R}^n$, with indices corresponding to the nodes in the graph.
%The inner product of two signals $\fv,\hv$ has its usual meaning $\braket{\fv,\hv} = \fv^t \hv$, with the norm $||\fv|| = \left(\braket{\fv,\fv}\right)^{1/2}$. 
%
We denote a subset of nodes of the graph as a collection of indices $\Sc \subset \Vc$, with $\Scc = \Vc \setminus \Sc$ denoting its complement set. 
A downsampled signal $\fv(\Sc)$, which is a vector of reduced length $|\Sc|$, is obtained by taking samples of $\fv$ on subset $\Sc$ of $\Vc$.
We denote the space of signals which may have non-zero values on $\Sc$ but are identically zero on $\Scc$ by $l_2(\Sc)$.
%
%A submatrix of $\Am$ obtained by taking its rows in set $\Sc_1$ and columns in set $\Sc_2$ is denoted by $\left(\Am)_{\Sc_1,\Sc_2}$ and for the sake of brevity $\Am_{\Sc,\Sc} = \Am_\Sc$. We use $\mathbf{I}$ to denote an identity matrix of appropriate size. $\zerov$ and $\onev$ denote matrices, of appropriate sizes, with all zeros and ones respectively. $R(\Am)$ denotes the range of operator $\Am$ and $N(\Am)$ denotes its null space.

It is known that the eigenvalues and eigenvectors of $\Lcb$ provide a spectral interpretation (i.e. a notion of frequency) for a graph signal, similar to the Fourier transform in traditional signal processing.  The eigenvalues of $\Lcb$ can be thought of as frequencies:
%and always lie in the real interval $[0,2]$. 
a high eigenvalue implies higher variation in the corresponding eigenvector \cite{Shuman-SPM-13}.
Every graph signal can be represented in the eigenvector basis as $\fv = \sum_i \tilde{f}\left(\lambda_i\right) \uv_i$, where $\tilde{f}\left(\lambda_i\right) = \braket{\fv,\uv_i}$ (or more compactly, $\tilde{\fv} = \Um^T\fv$) is the \emph{Graph Fourier Transform} (GFT). 
In this setting, an \emph{$\omega$-bandlimited signal} on a graph is defined as a signal with zero GFT coefficients at frequencies greater than its bandwidth $\omega$, i.e. its spectral support is restricted to the set of frequencies $[0,\omega]$. In other words, an $\omega$-bandlimited signal has its energy contained within the subspace spanned by the eigenvectors of the Laplacian with eigenvalues less than $\omega$.
The space of all $\omega$-bandlimited signals is known as the \emph{Paley-Wiener} space and is denoted by $PW_\omega(G)$ \cite{Pesenson-AMS-08}. 
% Note that $PW_\omega(G)$ is a subspace of $\mathbb{R}^n$. 
%We denote the bandwidth of a signal $\fv$ by $\omega(\fv)$.
%%%%%%%%%%%%%%%%%%%%%%%%%%%%%%%%%%%%%%%%%%%%%%%%%%%%%%%%%%%

\subsection{Reconstruction Problem}
We consider the problem of reconstructing a graph signal $\fv$ from its subsampled version $\fv(\Sc)$ under the assumption that $\fv$ is band-limited, i.e. $\fv \in PW_{\omega}(G)$. 
%Later on we will also consider the problem of reconstructing the missing samples such that a predefined quadratic energy function (which measures the variation in the signal on the graph) is minimized. 
%
Thus, $l_2(\Sc) = \Sc$ is the sampling subspace and $PW_\omega(G) = \Tc$ is the reconstruction subspace. Under a permutation which groups together nodes in $\Sc$ (and $\Scc$), we can represent the projectors $\Sm$ and $\Sm^\perp$ of $l_2(\Sc)$ and $l_2(\Scc)$ as
\begin{equation}
\Sm = 
\begin{pmatrix}
\mathbf{I} & \zerov \\
\zerov & \zerov
\end{pmatrix}, \quad
\Sm^{\perp} = 
\begin{pmatrix}
\zerov & \zerov \\
\zerov & \mathbf{I}
\end{pmatrix},
\end{equation}
so that $\Sm\fv = [\fv^\top(\Sc), \zerov]^\top$. $\Sm$ preserves the samples of a signal on $\Sc$ and sets the samples on $\Scc$ to zero. 
% $\Sm^{\perp}$, on the other hand, preserves the samples on $\Scc$ and sets the samples on $\Sc$ to zero. 

The projector $\Tm$ for $PW_{\omega}(G)$ is a low-pass filter which can be written in the graph spectral domain as 
\begin{equation}
\Tm = \sum_{i=1}^n h(\lambda_i) \uv_i\uv_i^t, \quad \text{where } 
h(\lambda) = 
\begin{cases}
   1 & \text{if } \lambda \leq \omega \\
   0 & \text{if } \lambda > \omega 
\end{cases}.
\end{equation}
The projector on the orthogonal complement of $PW_{\omega}(G)$ will be a high pass filter, $\Tm^{\perp} = \mathbf{I} - \Tm$.
The condition under which a bandlimited graph signal can be uniquely recovered from its samples on $\Sc$ is given in~\cite{Anis-ICASSP-14} as $PW_{\omega}(G) \cap l_2(\Scc) = \{\zerov\}$, which is equivalent to the one presented in Theorem~\ref{thm:uniqueness}.

% need to add something about localized interpolation

%%%%%%%%%%%%%%%%%%%%%%%%%%%%%%%%%%%%%%%%%%%%%%%%%%%%%%

\section{Numerical illustrations}
In this section, we apply the proposed reconstruction approaches to the image magnification problem.

\subsection{Problem set-up}
Let $\fv$ be the high resolution image of size $w \times w$. We assume that the samples (i.e., the low resolution version) of $\fv$ are obtained by a sampling operator $\Bm^*_\Sc$ which downsizes the image by a factor of $r$ using $r \times r$ averaging and then downsampling. Its adjoint $\Bm_\Sc$ upsamples a low resolution image by simply copying each pixel value in a $r \times r$ block to get back a $w \times w$ image. Thus, the sampling subspace $\Sc \subset \mathbb{R}^{w \times w}$ is a space of images which take a constant value in each $r \times r$ block. Note that $\dim \Sc = w/r$. The projection $\Sm\fv = \Bm_\Sc\Bm^*_\Sc \fv$ of $\fv$ on $\Sc$ is obtained replacing the values in each of its $r\times r$ blocks by their average. Our goal is to estimate $\fv$ having the input signal $\Sm\fv$.

%A downsampling $\Bm^*_\Sc$ on an $w \times w$ ideal image $f$ by a ratio $r$ is to average the pixels in a $r \times r$ block to a single pixel. A upsampling $\Sm^T$ is to repeat each pixel in a $r \times r$ block so as to restore the image to the full resolution. In the end, a consistent space projector given by $\Sm=\Bm_\Sc \Bm^*_\Sc$ can be proved to be orthogonal. The corresponding consistent subspace is $\Sc$.

We know that the DCT captures most of the energy of natural images into a first few low frequency coefficients. Thus, a reasonable guiding subspace $\Tc$ is a space of images which are bandlimited to the lowest $k \times k$ frequencies. The projector $\Tm$ for this subspace is simply a low pass filter which sets the higher frequency components of the image to zero. $\Tm$ can also be decomposed as $\Bm_\Tc\Bm^*_\Tc$. Here $\Bm^*_\Tc\fv$ involves taking the DCT of $\fv$ and setting the high frequency coefficients to zero whereas $\Bm_\Tc$ converts these DCT coefficients to spatial domain to get a low frequency image. 
%
%should be selected to drive the image reconstruction toward a high quality level. Given a well-known fact that DCT can strongly compact energy in an image by concentrating the signal in a few low-frequency components, we utilize a transform $\Bm^*_\Tc$ defined by a DCT transform followed by resetting all coeffcients to be zero outside of the range of $k \times k$ low frequencies. 
%In such a way, high frequency noises can be filtered. Moreover, the inverse transform $\Bm_\Tc$ is further defined by an inverse DCT transform. 
%In the end, we can prove that the guiding subpace projector $\Tm = \Bm_\Tc \Bm^*_\Tc$ is orthogonal.

In our experiments, we study the effect of $\dim \Tc = k \times k$ on quality of reconstruction.  
%Using the proposed guiding subspace orthogonal projector $\Tm$ above, the guiding subspace $\Tc$ is spanned by $k \times k$ DCT coefficients and its dimension is $k \times k$. 
%In order to facilitate the following studies, 
We define $k_{\text{scale}} = (w / r)/k$ which compares the dimensionality of the sampling and guiding subspace. $k_{\text{scale}} < 1$ corresponds to an undersampling problem,  while $k_{\text{scale}} > 1$, corresponds to an oversampling scenario.
A shorthand $\fv_d$ is used to denote the low resolution image $\Bm^*_\Sc \fv$ and $\fv_{du}$ to denote the projection $\Sm\fv$.
%The ideal image $f$ can be assumed to be composed of two parts, $f = \Sm f + \xv = \fv_{du} + \xv$, where $\fv_{du} = \Sm f = \Bm_\Sc \Bm^*_\Sc f$ is the input image to be interpolated. It is clear that $\fv_{du} \in \Sc$ and $\xv \in \Sc^{\perp}$.
%
We also consider the scenario where the samples are contaminated by noise, i.e., $\fv^n_{d} = \Bm^*_\Sc \fv + \ev$, where $\ev$ is i.i.d. Gaussian noise. As a result, the input image becomes $\fv^n_{du} = \Bm_\Sc \fv_d^n$.

\subsection{Approaches under study}

We compare four reconstruction approaches, namely, the consistent reconstruction $\hat{\fv}_c$, the generalized reconstruction $\hat{\fv}_g$, the regularized reconstruction $\hat{\fv}_\alpha$ and the minimax regret reconstruction $\hat{\fv}_m = \Tm\fv_{du}$.

Consistent reconstruction $\hat{\fv}_c$ is calculated as $\hat{\fv}_c = \hat{\xv} + \fv_{du}$, where $\hat{\xv}$ is the solution to problem $\Sm^\perp \Tm^\perp \xv = - \Sm^\perp \Tm^\perp \fv_{du}$ obtained using the conjugate gradient method. 

Generalized reconstruction $\hat{\fv}_g$ is computed using three different implementations. In the first implementation, we solve the problem $\Bm^*_\Tc \Sm \Bm_\Tc \yv = \Bm^*_\Tc \fv_{du}$ using a conjugate gradient method to obtain $\hat{\yv}$. The final reconstruction is then given by $\hat{\fv}_{g1} = \Bm_\Tc \hat{\yv}$. The second implementation uses the projector $\Tm$ instead of the sampling operator $\Bm^*_\Tc$, and the reconstruction $\hat{\fv}_{g2}$ is the conjugate gradient solution to the problem $\Tm \Sm \Tm \fv = \Tm \fv_{du}$. In the third implementation, $\hat{\fv}_c$ is supposed to be available, and the generalized reconstruction is then computed by $\hat{\fv}_{g3} = \Tm \hat{\fv}_c$. Mathematically, it can be proved that all these implementations would produce identical reconstructions when the conjugate gradient algorithm converges. However, these methods are algorithmically distinct and may converge at different rates as shown in the tests later.

Regularized reconstruction $\hat{\fv}_r$, as posed in \eqref{eq:reg_problem}, can be computed by solving $(\Sm + \rho \Tm^\perp) \fv = \fv_{du}$ via conjugate gradient. If $\hat{\fv}_c$  and $\hat{\fv}_g$ are available, we can simply take the convex combination $\hat{\fv}_\alpha = \alpha \hat{\fv}_c + (1-\alpha) \hat{\fv}_g$ with $\rho = ({1-\alpha})/{\alpha}$ and because of Theorem~\ref{thm:reg}, we have $\hat{\fv}_r = \hat{\fv}_\alpha$. Although these two solutions are mathematically equivalent (upon convergence of conjugate gradient), they are not similar algorithmically and show different behavior and robustness against noise for a small fixed number of CG steps. 

%``Guided consistent" reconstruction $\hat{\fv}_\alpha$, is computed as a weighted results between the consistent reconstruction and generalized reconstruction, $\hat{\fv}_\alpha = \alpha \hat{\fv}_c + (1-\alpha) \hat{\fv}_g$, with $\alpha \in [0,1]$ and $\hat{\fv}_{g1}$ is used for $\hat{\fv}_g$. In particularly, $\hat{\fv}_c = \hat{\fv}_{\alpha=1}$, and $\hat{\fv}_g = \hat{\fv}_{\alpha=0}$.
%
%Regularized reconstruction $\hat{\fv}_r$ is from $\hat{\fv}_r = \underset{f}{\arg\min} \frac{1}{2} \left\|\Sm f - \Sm \fv_{du}\right\|^2 + \frac{\rho}{2} \left\| \Tm^\perp f \right\|^2$. We could use a conjugate gradient algorithm to compute $\hat{\fv}_r$ by solving $(\Sm + \rho \Tm^\perp) f = \fv_{du}$. It can be easily proven in mathematics that the two approaches will result in identical reconstructions. By plugging $\hat{\fv}_\alpha = (1-\alpha) \hat{\fv}_g + \alpha \hat{\fv}_c$ into $(\Sm + \rho \Tm^\perp) \hat{\fv}_\alpha$, it would lead to $\fv_{du}$ exactly when $\rho = \frac{1-\alpha}{\alpha}$. It is interesting to note that the regularizer parameter $\rho$ can be linked to the weighting factor $\alpha$.

\subsection{Experiments and observations}

We conduct four sets of experiments to study different aspects of the reconstruction methods such as the effect of under/oversampling, effect of noise and convergence behavior.

\subsubsection{Experiment 1}

In the first experiment, we take a noise free signal $\fv_{du}$ as input and observe the PSNR of reconstruction for different methods as the value of $k_{\text{scale}}$ (i.e., amount of under/oversampling) varies. For computing $\fv_\alpha$, we first fix $\alpha = 0.7$. Fig.~\ref{fig:e1}(a) shows the plot of PSNR against $k_{\text{scale}}$. We observe that in the undersampling regime, i.e. when $k_{\text{scale}} < 1$, $\hat{\fv}_c$ equals $\hat{\fv}_g$ and performs better than $\hat{\fv}_m$. In case of oversampling, however, it shows $\hat{\fv}_c$ offers better PSNR than $\hat{\fv}_g$  which, in turn, performs better than $\hat{\fv}_m$.
This is because sampling is noise free and a method which keeps the samples unchanged is expected to perform better. 
% For heavily over-sampling cases, $\hat{\fv}_c$ is clearly the preferred method in terms of reconstruction quality. 
Example reconstructed images $\hat{\fv}_g$ and $\hat{\fv}_c$ with $\dim \Hc = 256 \times 256$, $\dim \Sc = 128 \times 128$ and $k_{\text{scale}}=4$ are shown in Fig.~\ref{fig:e2}
%the dimension of the guiding subspace is $\dim({\Tc}) = 32 \times 32$, which is a heavily over-sampling case. 
%In this example, the reconstructed images from $\hat{\fv}_g$ ($=\hat{\fv}_{\alpha=0}$) and $\hat{\fv}_c$ ($=\hat{\fv}_{\alpha=1}$) are shown in Fig. \ref{fig:e2}. 
The effect of $\alpha$ on the reconstruction quality is illustrated in Fig.~\ref{fig:e1}(b). Once again we observe that as $\alpha$ increases (i.e., the samples are trusted more), the reconstruction quality improves.

%\begin{figure}\center
% \includegraphics[width=.3\textwidth]{figs/sim/lena_k_scale.pdf}
% \caption{Reconstructed image qualities vs. $k_{\text{scale}}$ for guiding subspace $\Tc$}
% \label{fig:e1}
%\end{figure}

\begin{figure}
\centering
       \begin{subfigure}[b]{0.2\textwidth}\center
               \includegraphics[height=.125\textheight]{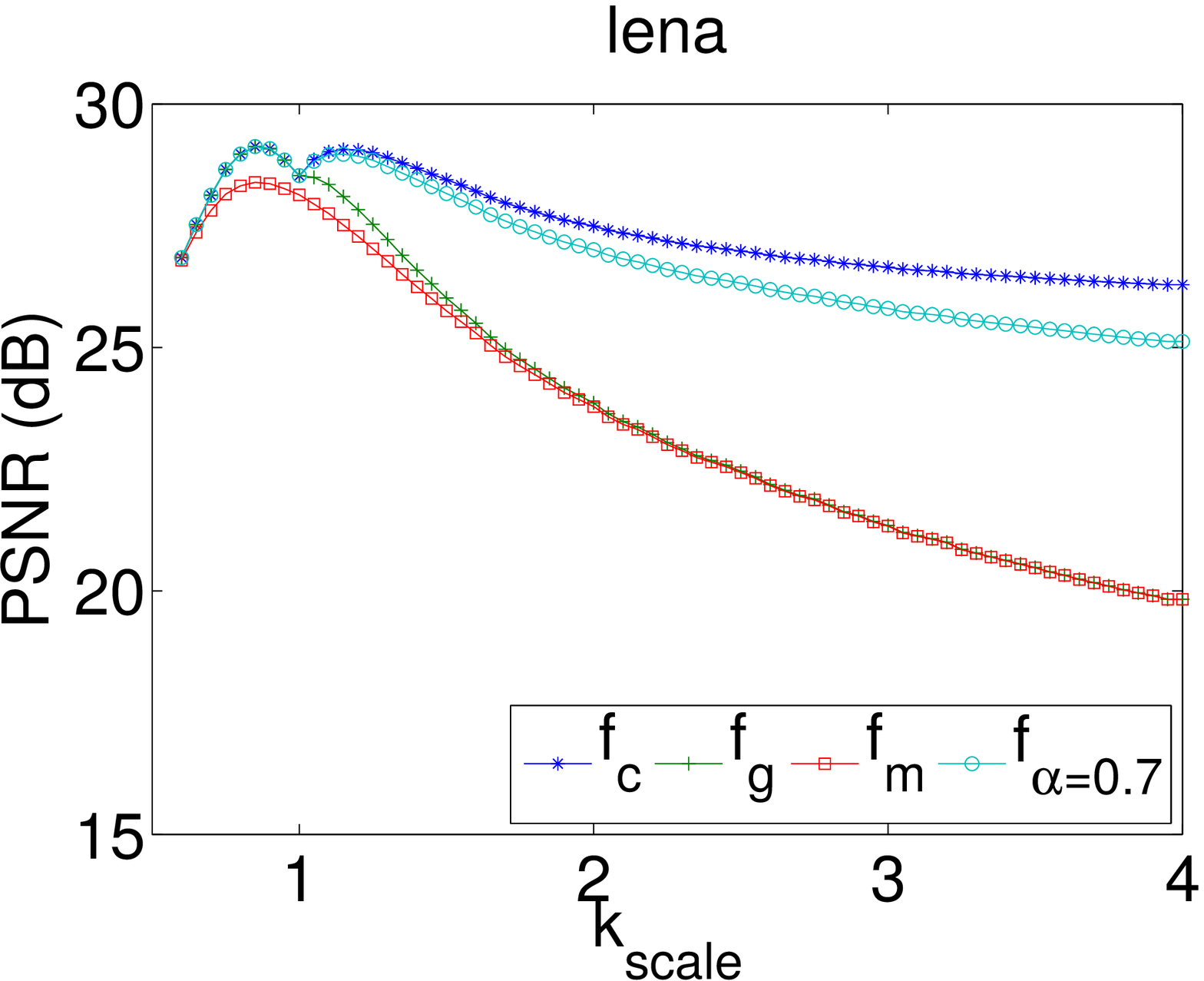}
               \caption{$\alpha=0.7$}
       \end{subfigure}
       \hspace{0.03cm}
       \begin{subfigure}[b]{0.2\textwidth}\center
               \includegraphics[height=.125\textheight]{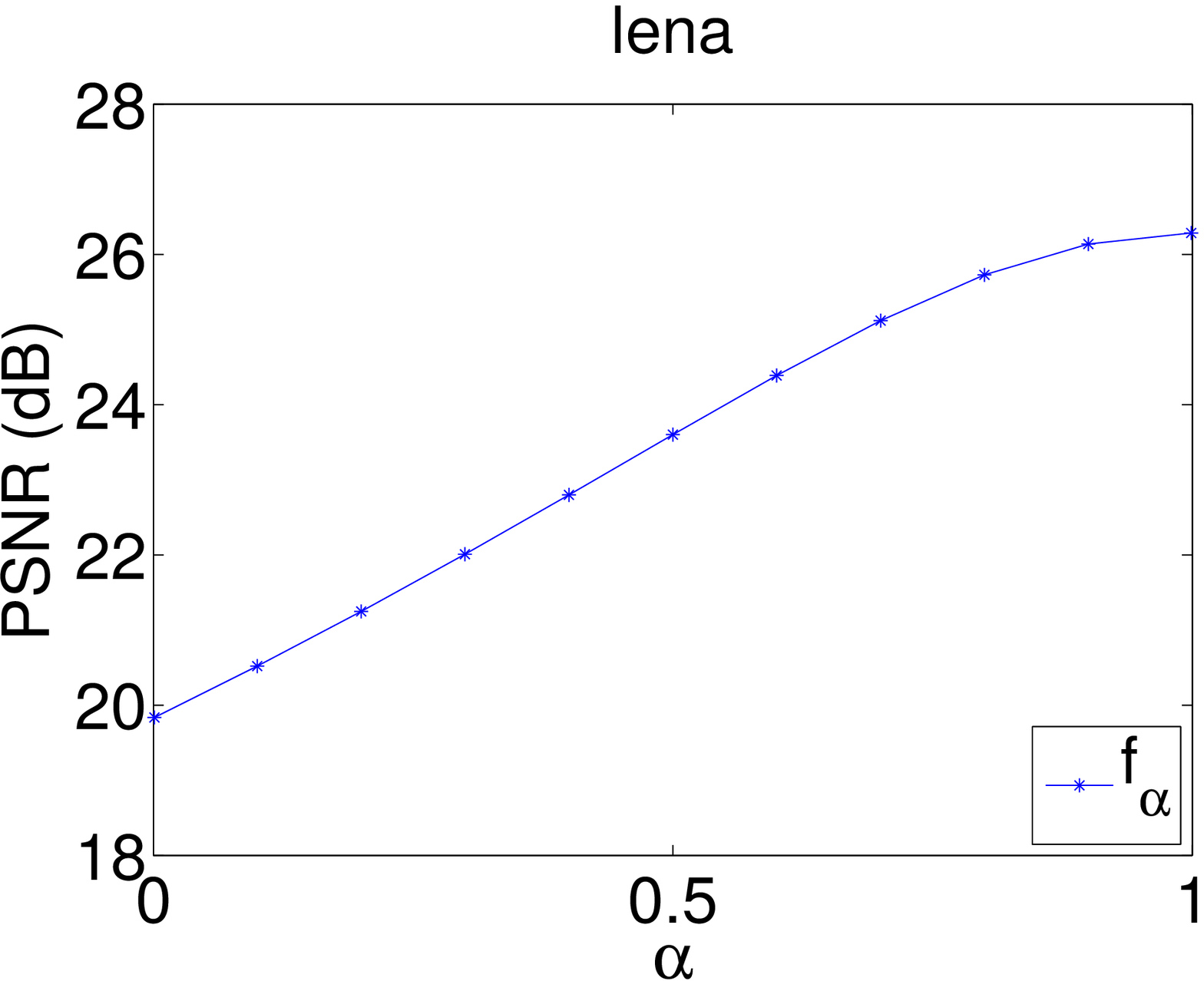}
               \caption{$k_{\text{scale}} = 4$}
       \end{subfigure}
\caption{Effects of $k_{\text{scale}}$ and $\alpha$ on noise-free reconstruction}
\label{fig:e1}
\end{figure}

\begin{figure}
\centering
       \begin{subfigure}[b]{0.2\textwidth}\center
               \includegraphics[height=.125\textheight]{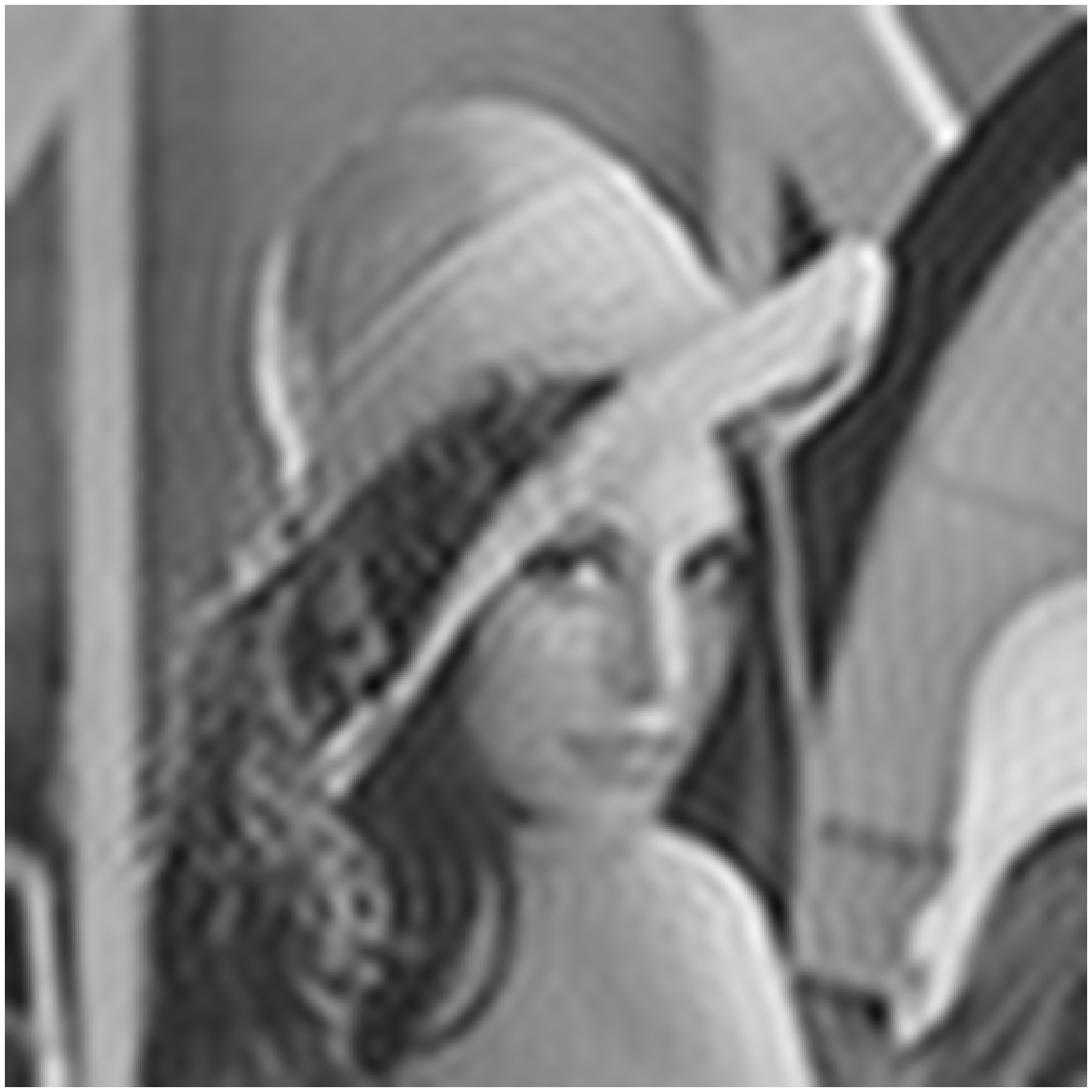}
               \caption{$\hat{\fv}_g$, PSNR=19.83dB}
       \end{subfigure}
       \begin{subfigure}[b]{0.2\textwidth}\center
               \includegraphics[height=.125\textheight]{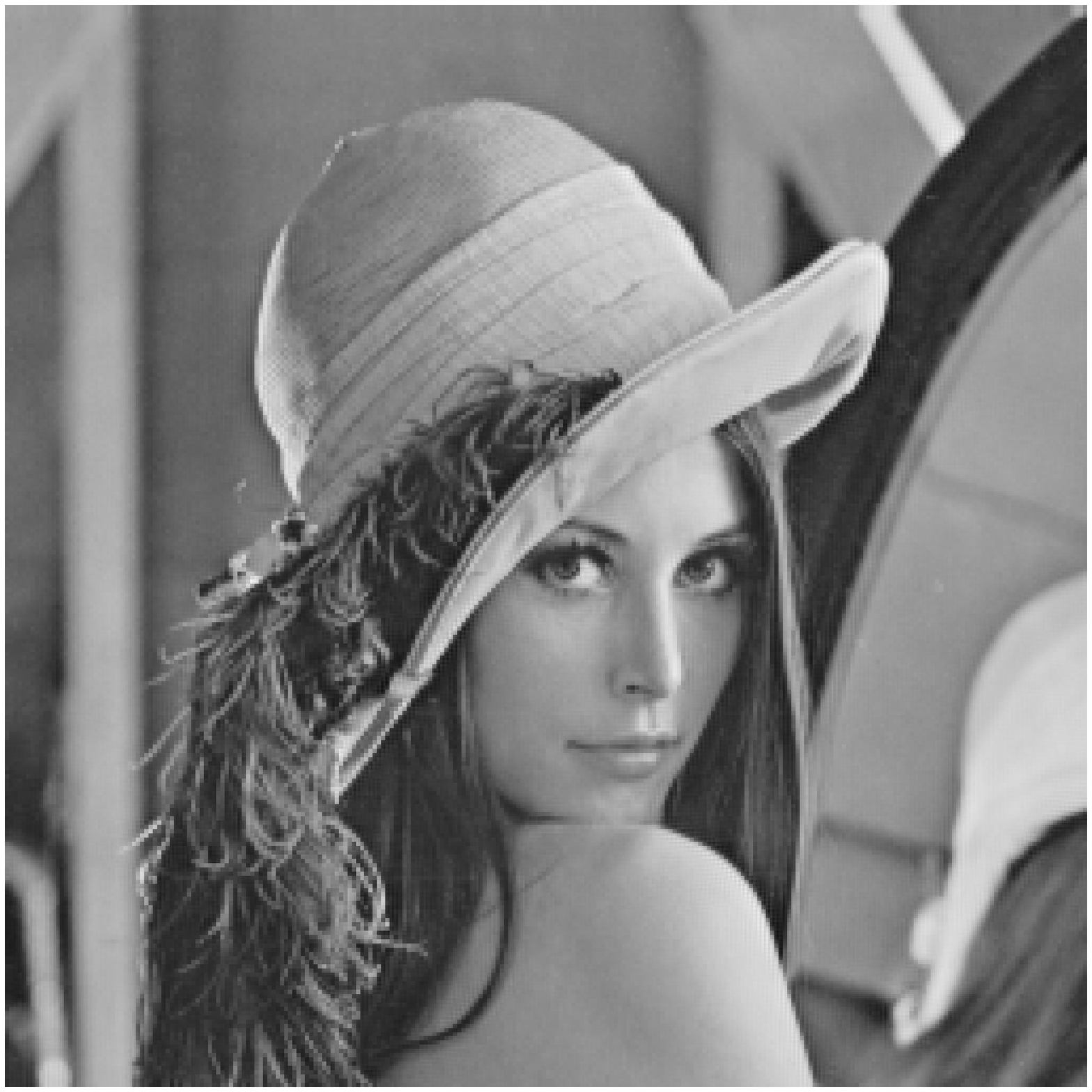}
               \caption{$\hat{\fv}_c$, PSNR=26.29dB}
       \end{subfigure}
\caption{Noise-free reconstruction with $\dim \Sc=128\times 128$, $\dim \Tc = 32 \times 32$, $k_{\text{scale}} = 4$}
\label{fig:e2}
\end{figure}

\subsubsection{Experiment 2}
%
%For noise free inputs $\fv_{du}$, it was shown above that $\hat{\fv}_c = \hat{\fv}_{\alpha=1}$ is the best algorithm. 
In this experiment, we assume that the input $\fv_{du}^n = \Sm \fv + \ev$ is noisy, where $\ev$ is i.i.d. Gaussian  with zero mean and variance $0.001$.

\begin{figure}
\centering
       \begin{subfigure}[b]{0.2\textwidth}\center
               \includegraphics[height=.125\textheight]{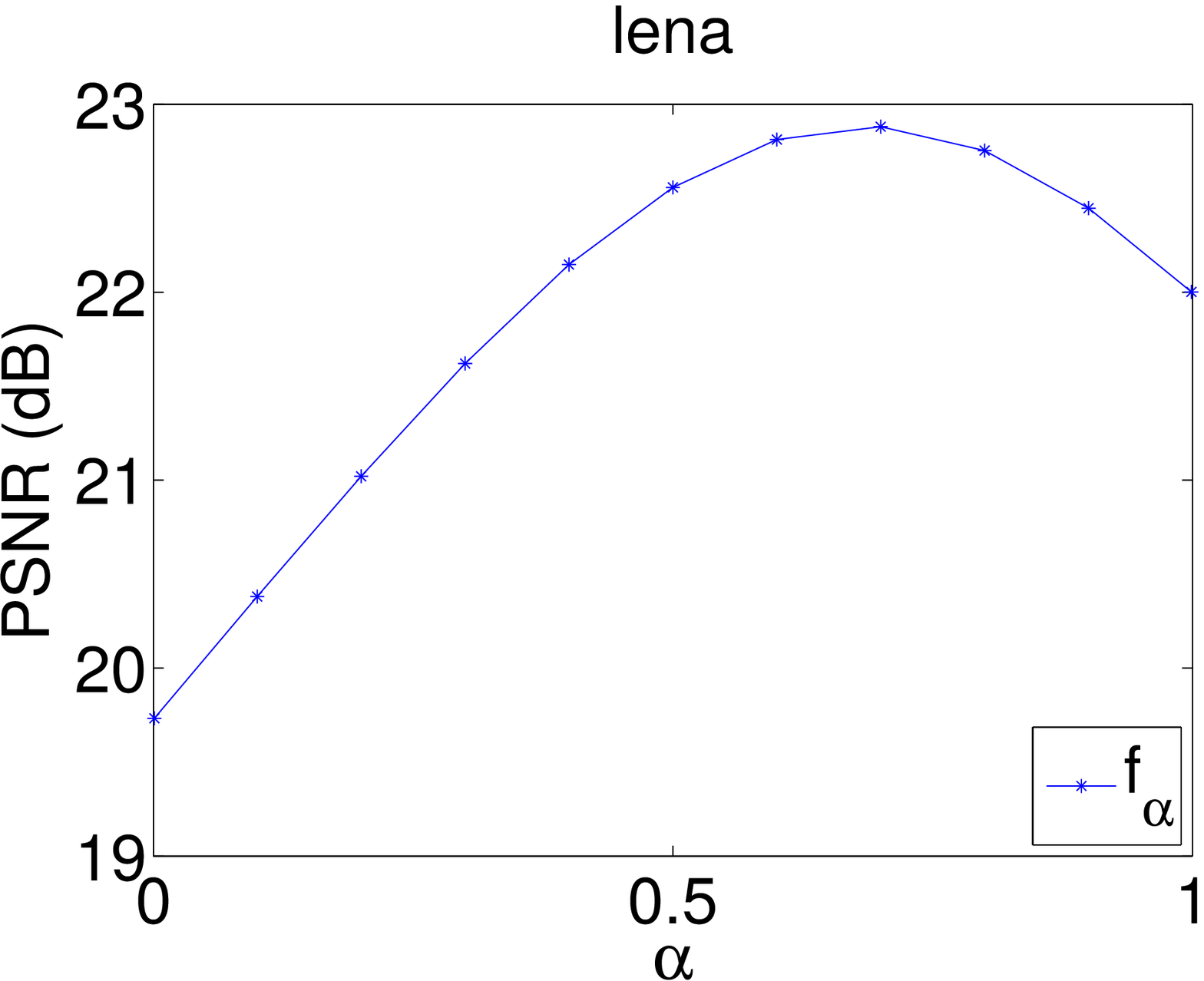}
               \caption{$k_{\text{scale}} = 4$}
       \end{subfigure}
       \hspace{0.03cm}
       \begin{subfigure}[b]{0.2\textwidth}\center
               \includegraphics[height=.125\textheight]{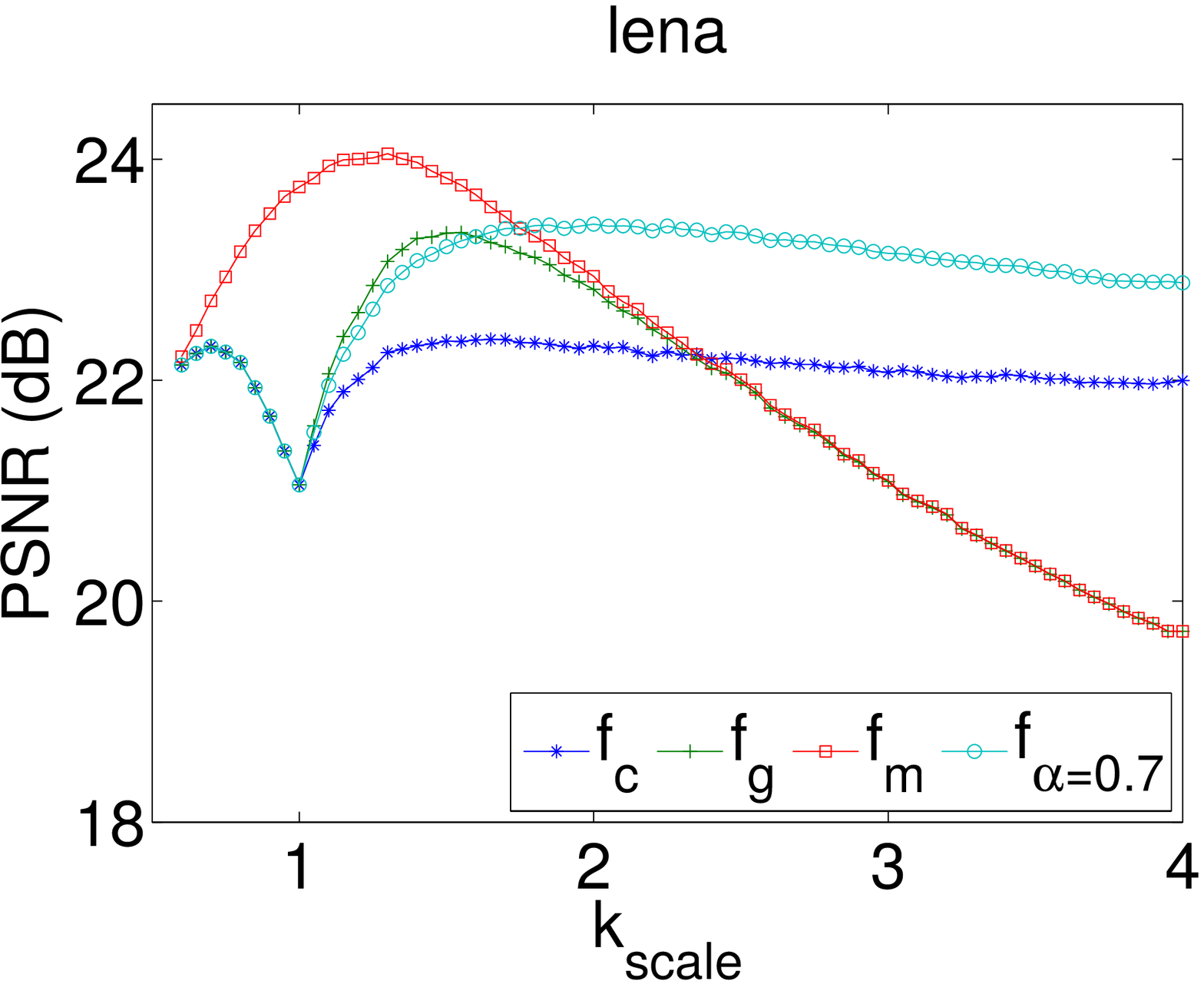}
               \caption{$\alpha=0.7$}
       \end{subfigure}
\caption{Effects of $k_{\text{scale}}$ and $\alpha$ on noisy reconstruction}
\label{fig:e3}
\end{figure}

We first focus on the performance of $\hat{\fv}_\alpha$ as $\alpha$ varies in case of oversampling by a factor $k_{\text{scale}} = 4$. From the results shown in Fig. \ref{fig:e3}(a), the best reconstruction is obtained with $\alpha = 0.7$. This observation agrees with the theoretically suggested optimal value $\alpha_{\text{opt}} = 1 - \|\ev\|^2/\|\hat{\fv}_g - \hat{\fv}_c\|^2 = 0.7$.  

We next analyze performance of $\hat{\fv}_g$, $\hat{\fv}_c$, $\hat{\fv}_m$ and $\hat{\fv}_{\alpha=0.7}$ for different values of $k_\text{scale}$, in Fig. \ref{fig:e3}(b). The minimax regret reconstruction $\hat{\fv}_m = \Tm\fv_{du}$, in contrast to the noise-free case displayed in Fig.~\ref{fig:e2}, 
produces the best PSNR  if $k_{\text{scale}}< 1.8$, which can be easily explained since $\Tm$ is a low-pass filter, performing image denoising. It can be thus recommended to combine the reconstruction procedure with pre- and possibly post-denoising, e.g., using \cite{KnyazevM15e}. 

As opposed to the previous noise free experiment, we notice in Fig.~\ref{fig:e3}(b) that $\hat{\fv}_c$ cannot always beat $\hat{\fv}_g$ when noise is present. $\hat{\fv}_c$ only performs better than $\hat{\fv}_g$ in the heavy oversampling regime, in this example ($k_{\text{scale}} > 2.5$). This observation indicates that, in case of slight oversampling, the noise filtering effect of the projection on guiding subspace offsets the loss due to sample inconsistency. On the other hand, for heavy oversampling, the sample consistency requirement is more important. 
%This implies that with noisy signal, increasing the dimension of sampling subspace will in general benefits the reconstruction quality. 
We also observe that $\hat{\fv}_\alpha$ which is a weighted combination of $\hat{\fv}_c$ and $\hat{\fv}_g$ can beat both $\hat{\fv}_c$ and $\hat{\fv}_g$ for  $k_{\text{scale}} > ~1.5$ for this example image. This is because it offers some noise suppression while not deviating much from the consistency requirement. 
%This is a different conclusion as in the noise free use case where $\hat{\fv}_c$ is always better than other choices. 
%With noises introduced, we confirmed in the experiments that a combination between the consistency subspace and the guding subspace using "guided consistent" reconstruction approach could bring additional benefits with certain level of over-sampling schemes. 
%Last, from the experiments, it is worthy to point out that the selection of optimal $\alpha$ could be selected according to the noise level, $\alpha_{optimal} = 1 - \|\ev\|^2/\|\hat{\fv}_g - \hat{\fv}_c\|^2 = 0.7$. With $0$ noise, $\alpha_{optimal} = 1$, which is also consistent to noise free tests in previoius experiment since $\hat{\fv}_{\alpha=1} = \hat{\fv}_c$.
%
Fig.~\ref{fig:e4} shows an example of the noisy input image and reconstructed images.
% from $\hat{\fv}_g$, $\hat{\fv}_c$, and $\hat{\fv}_c$.

\begin{figure}
\centering
       \begin{subfigure}[b]{0.2\textwidth}\center
               \includegraphics[height=.125\textheight]{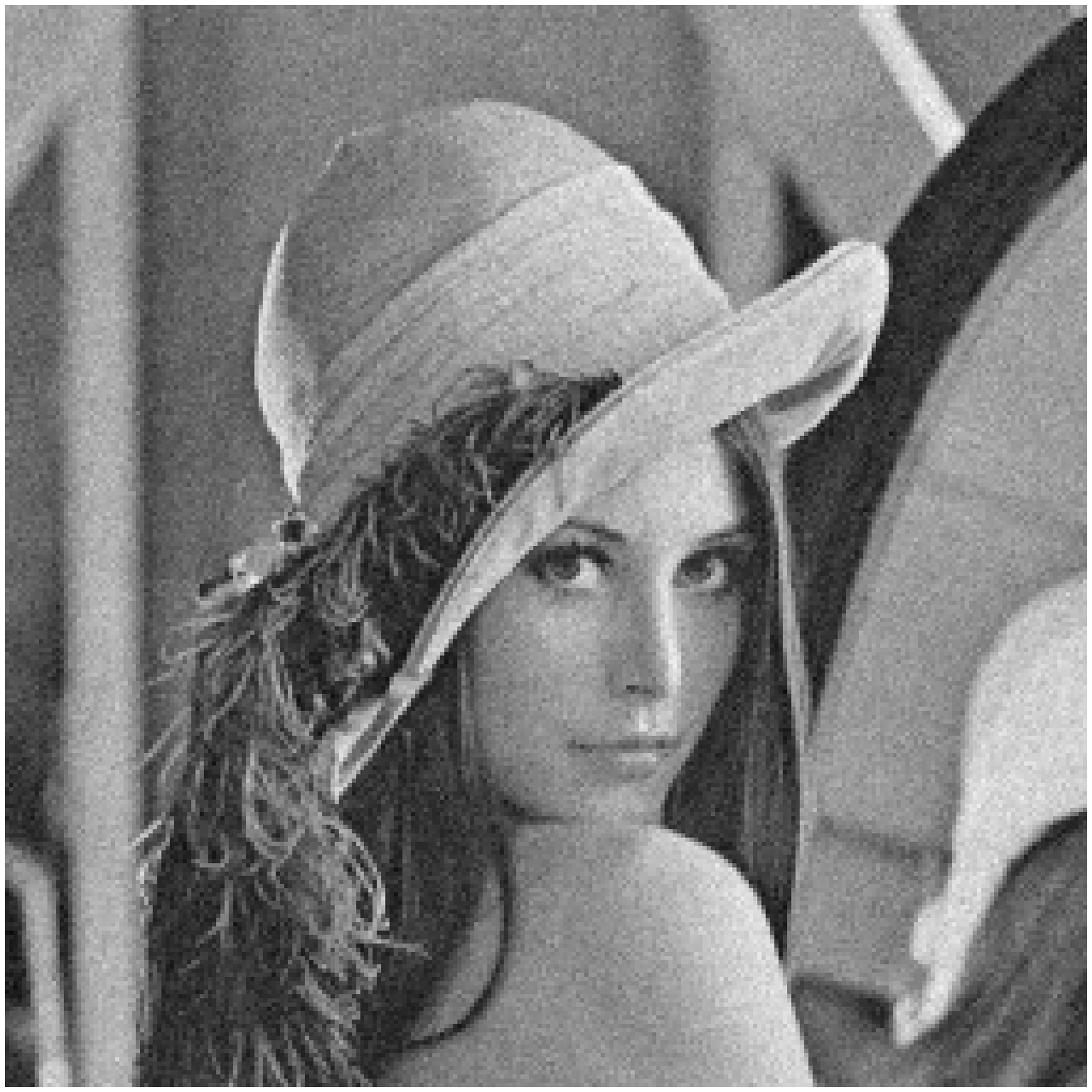}
               \caption{$\fv_{du}^*$, PSNR=21.69dB}
       \end{subfigure}
       \begin{subfigure}[b]{0.2\textwidth}\center
               \includegraphics[height=.125\textheight]{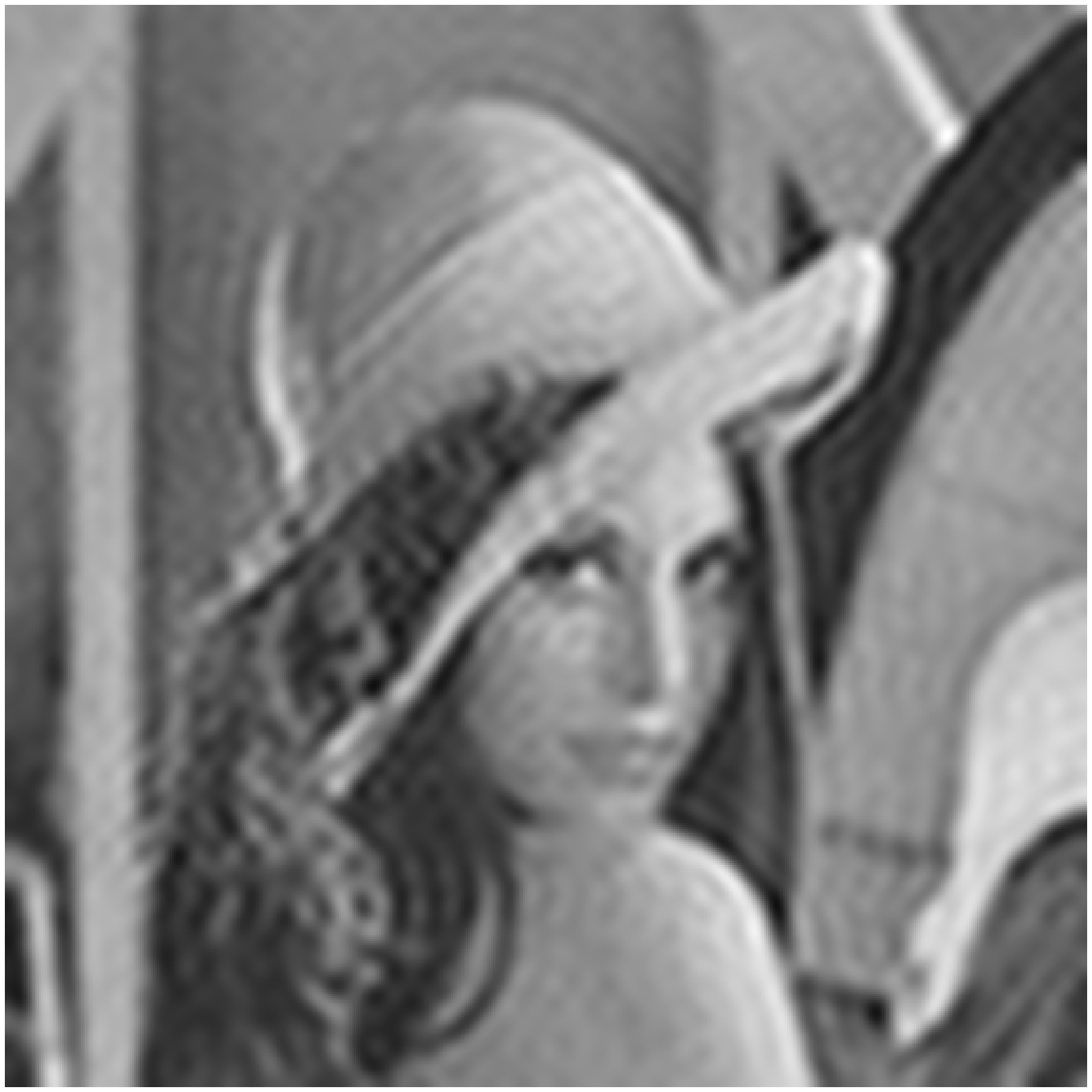}
               \caption{$\hat{\fv}_g$, PSNR=19.73dB}
       \end{subfigure}

       \begin{subfigure}[b]{0.2\textwidth}\center
               \includegraphics[height=.125\textheight]{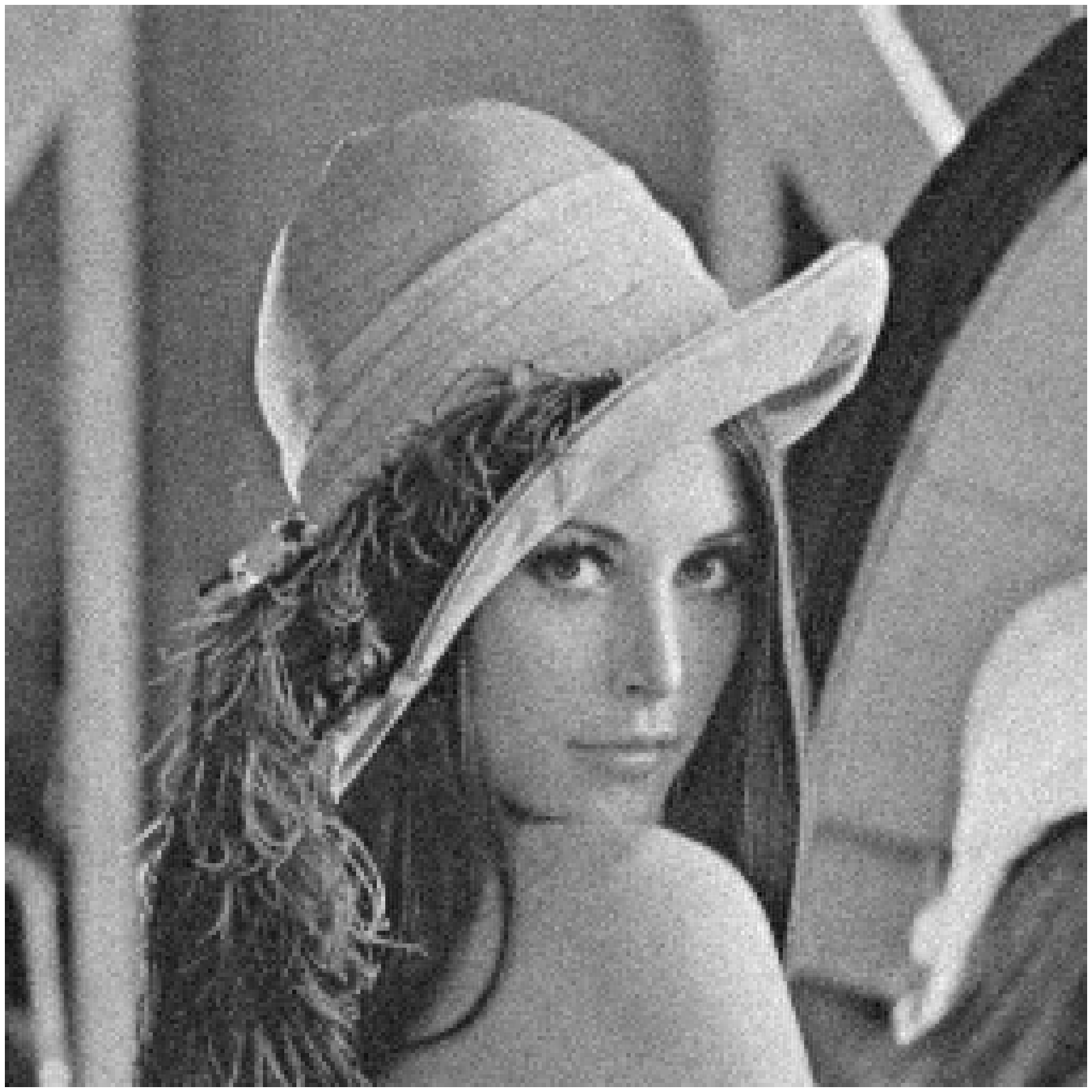}
               \caption{$\hat{\fv}_c$, PSNR=22.00dB}
       \end{subfigure}
       \begin{subfigure}[b]{0.2\textwidth}\center
               \includegraphics[height=.125\textheight]{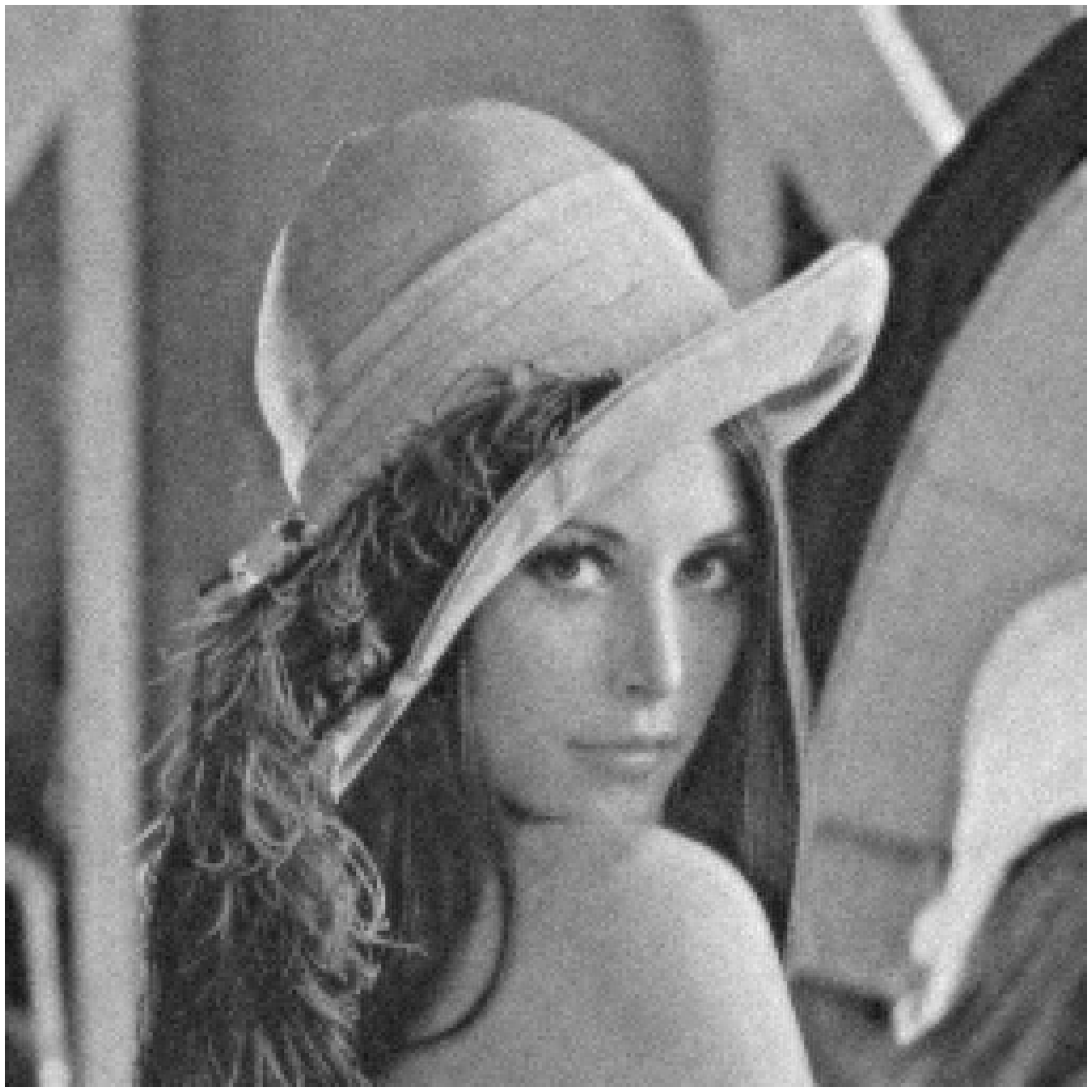}
               \caption{$\hat{\fv}_{\alpha=0.7}$, PSNR=22.88dB}
       \end{subfigure}
\caption{Reconstruction results with noisy inputs, $k_{\text{scale}} = 4$}
\label{fig:e4}
\end{figure}

\subsubsection{Experiment 3}

In this experiment, we study the relationship between $\hat{\fv}_\alpha$ and $\hat{\fv}_r$ in case of noisy inputs. 
%since for the noise free case, there is not need to look for $\hat{\fv}_\alpha$ or $\hat{\fv}_r$, because $\hat{\fv}_c$ is always the best. 
Numerical results confirm that if the parameter $\rho$ or $\alpha$ is known beforehand and are fixed, the two approaches, despite having different implementations, give identical reconstructions. 
%For example, in previous subsections, we observed that an optimial $\alpha$ can be determined by the noise level via $\alpha_{optimal} = 1 - \|\ev\|^2/\|\hat{\fv}_g - \hat{\fv}_c\|^2$. In this case, we need only solve one least square problem $\Am \xv = \bv$ to get the final reconstruction.
%
However, if the parameter $\rho$ or $\alpha$ needs to be determined on the fly in an application, $\hat{\fv}_\alpha$ is clearly favorable than $\hat{\fv}_r$ in terms of computation complexity.
For computing the whole set of solution $\{\hat{\fv}_\alpha\}$, for $\alpha \in (0, 1)$, only \emph{one} least squares problem needs to be solved which is to compute $\hat{\fv}_c$. All other candidate solution points can be calculated by $\alpha \hat{\fv}_c + (1-\alpha)\Tm\hat{\fv}_c$ since $\hat{\fv}_g = \Tm \hat{\fv}_c$.
On the other hand to search through the full set of $\{\hat{\fv}_r\}$, for $\rho \in (0, \infty)$, \emph{one} least squares problem needs to be solved for \emph{each} candidate solution which may not be computationally feasible.

\subsubsection{Experiment 4}

In the previous experiments, all the conjugate gradient algorithms used to solve a least squares problem are allowed to converge. The purpose of this experiment is to compare how the reconstruction methods perform each iteration of conjugate gradient.

As described before, $\hat{\fv}_g$ has three different implementations, represented by $\hat{\fv}_{g1}$, $\hat{\fv}_{g2}$, and $\hat{\fv}_{g3}$. In Fig.~\ref{fig:e5}, with noisy input, we compare the three implementations with maximum number of CG iterations, $MaxIter$,  set to 1 and 2. We observe that $\hat{\fv}_{g1} = \hat{\fv}_{g2}$ in both cases. Although $\hat{\fv}_{g3}$ is different when the number of iterations is 1, as seen in Fig.~\ref{fig:e5}(a), the difference becomes very minor when the number of iterations equals 2. This observation also holds for noise free inputs.

\begin{figure}
\centering
       \begin{subfigure}[b]{0.2\textwidth}\center
               \includegraphics[height=.125\textheight]{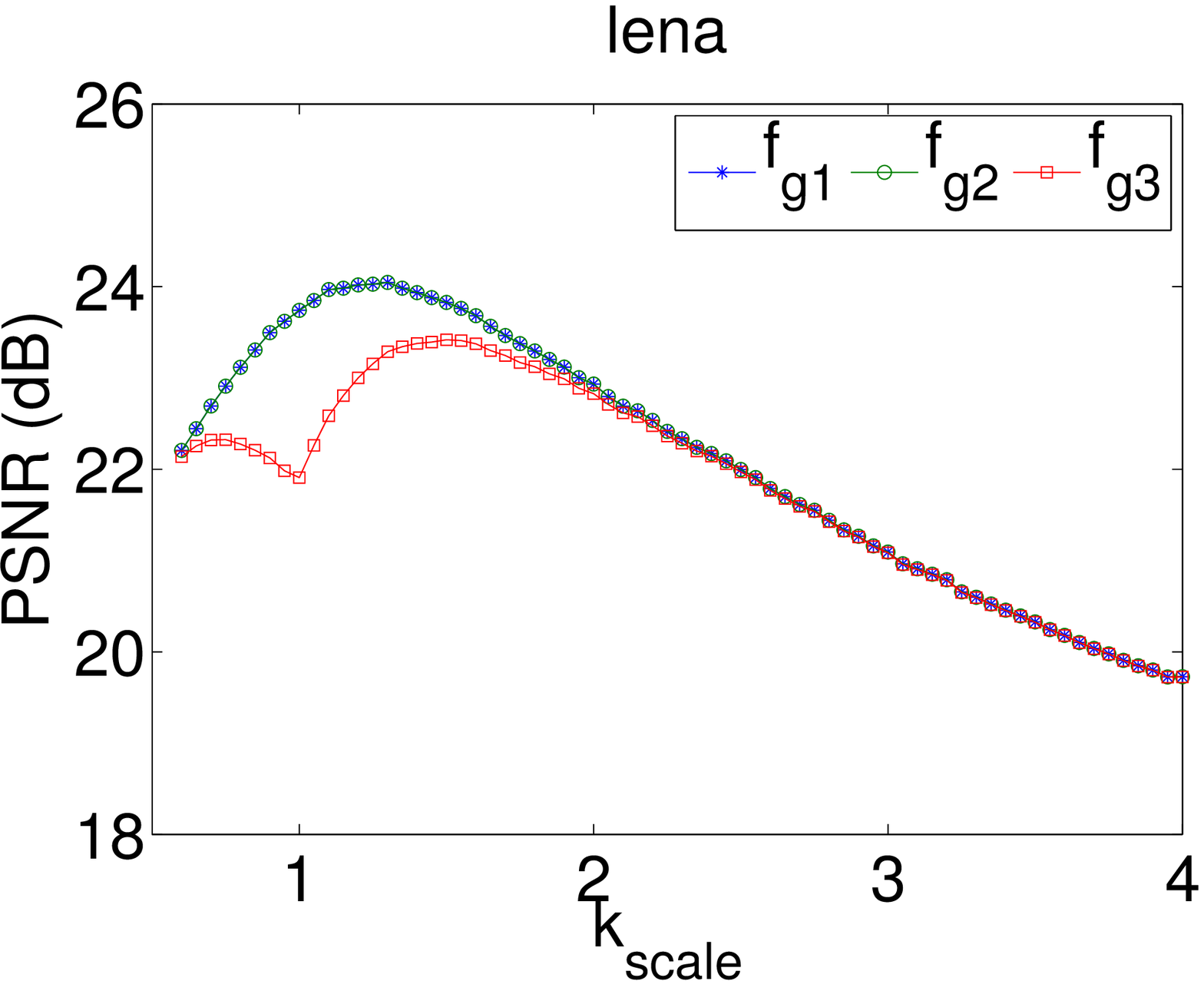}
               \caption{$MaxIter=1$}
       \end{subfigure}
       \hspace{0.03cm}
       \begin{subfigure}[b]{0.2\textwidth}\center
               \includegraphics[height=.125\textheight]{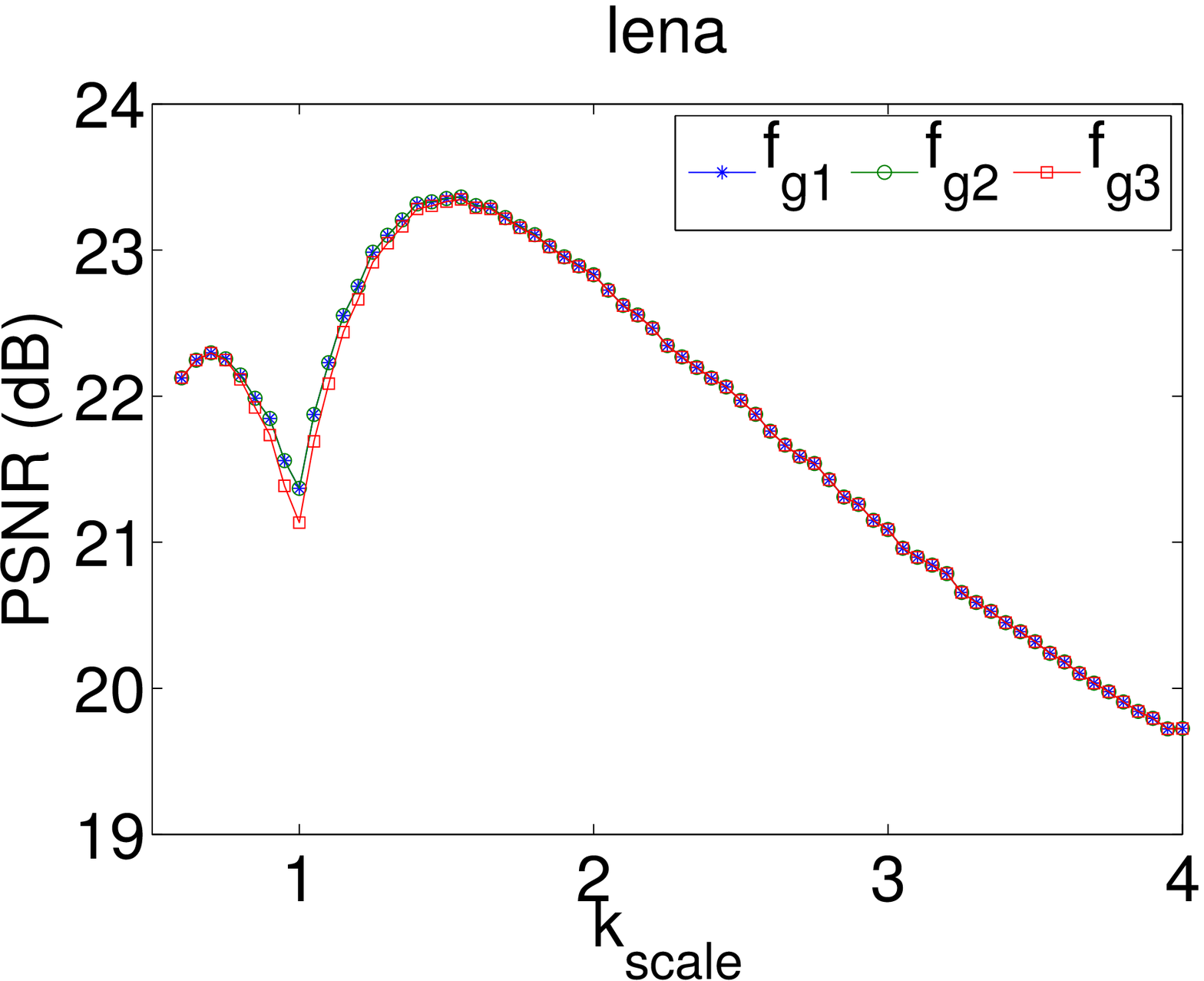}
               \caption{$MaxIter=2$}
       \end{subfigure}
\caption{Performance of different implementations of $\hat{\fv}_{g}$}
\label{fig:e5}
\end{figure}

Since $\hat{\fv}_g$ has three implementations, $\hat{\fv}_\alpha$ can also have different corresponding implementations, given by $\hat{\fv}_{\alpha i}=\alpha \hat{\fv}_c + (1-\alpha) \hat{\fv}_{gi}$ with $i=1,2,3$. The performance of all the reconstruction methods with different implementations is shown in Fig. \ref{fig:e6}. All the algorithms are configured to use $MaxIter$ number of CG iterations (except $\hat{\fv}_m$ since it does not need to solve a least square problem). $\hat{\fv}_{\alpha 2}$ is omitted as it is always equal to $\hat{\fv}_{\alpha 1}$. 
We observe that $\hat{\fv}_{\alpha1} = \hat{\fv}_{\alpha2}$ performs better than $\hat{\fv}_{\alpha3}$. In case of heavier oversampling, $\hat{\fv}_\alpha$ is more favorable than $\hat{\fv}_r$. Finally, $\hat{\fv}_r$  shows worse performance compared to other approaches.

\begin{figure}
\centering
       \begin{subfigure}[b]{0.2\textwidth}\center
               \includegraphics[height=.125\textheight]{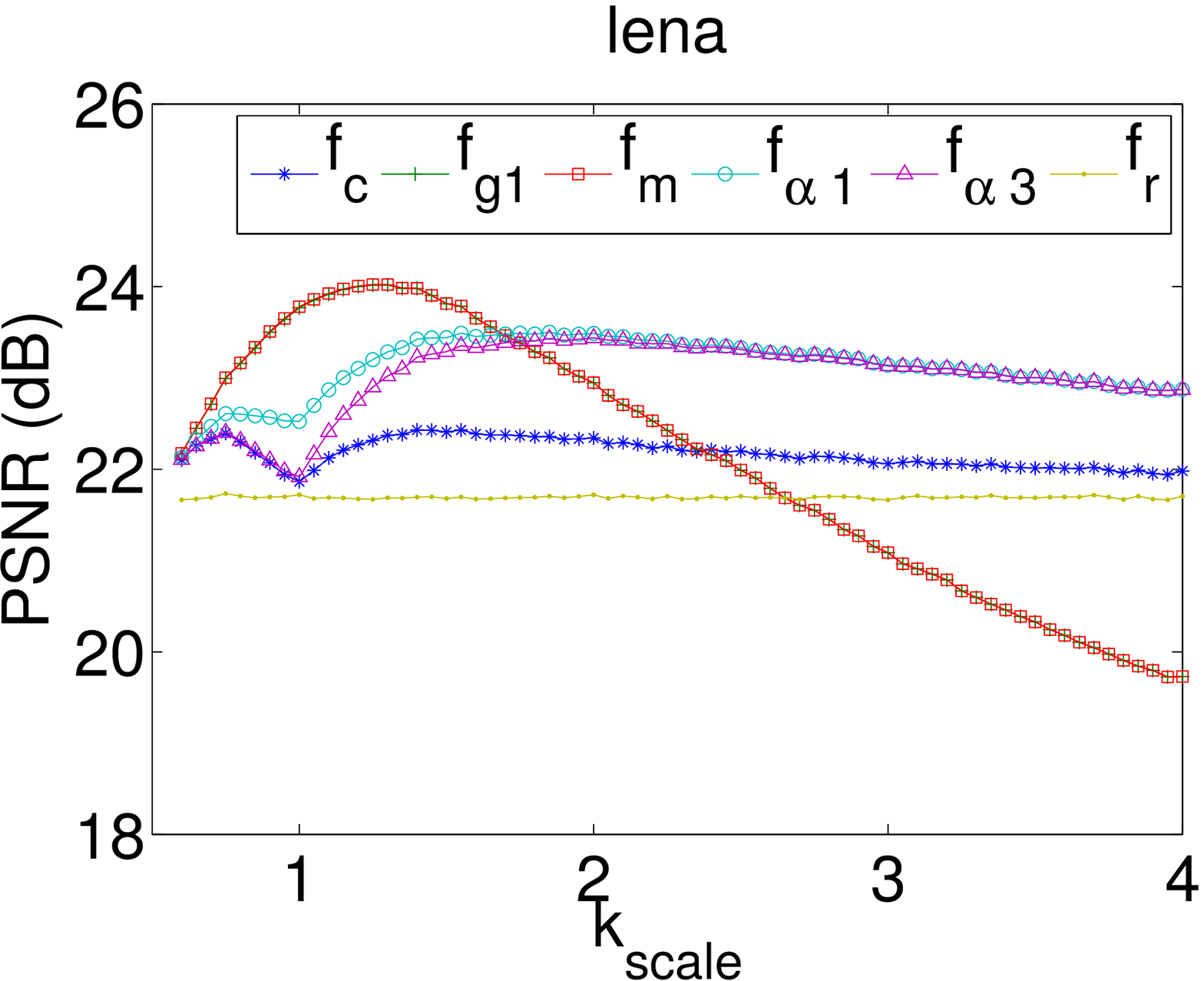}
               \caption{$MaxIter=1$}
       \end{subfigure}
       \hspace{0.03cm}
       \begin{subfigure}[b]{0.2\textwidth}\center
               \includegraphics[height=.125\textheight]{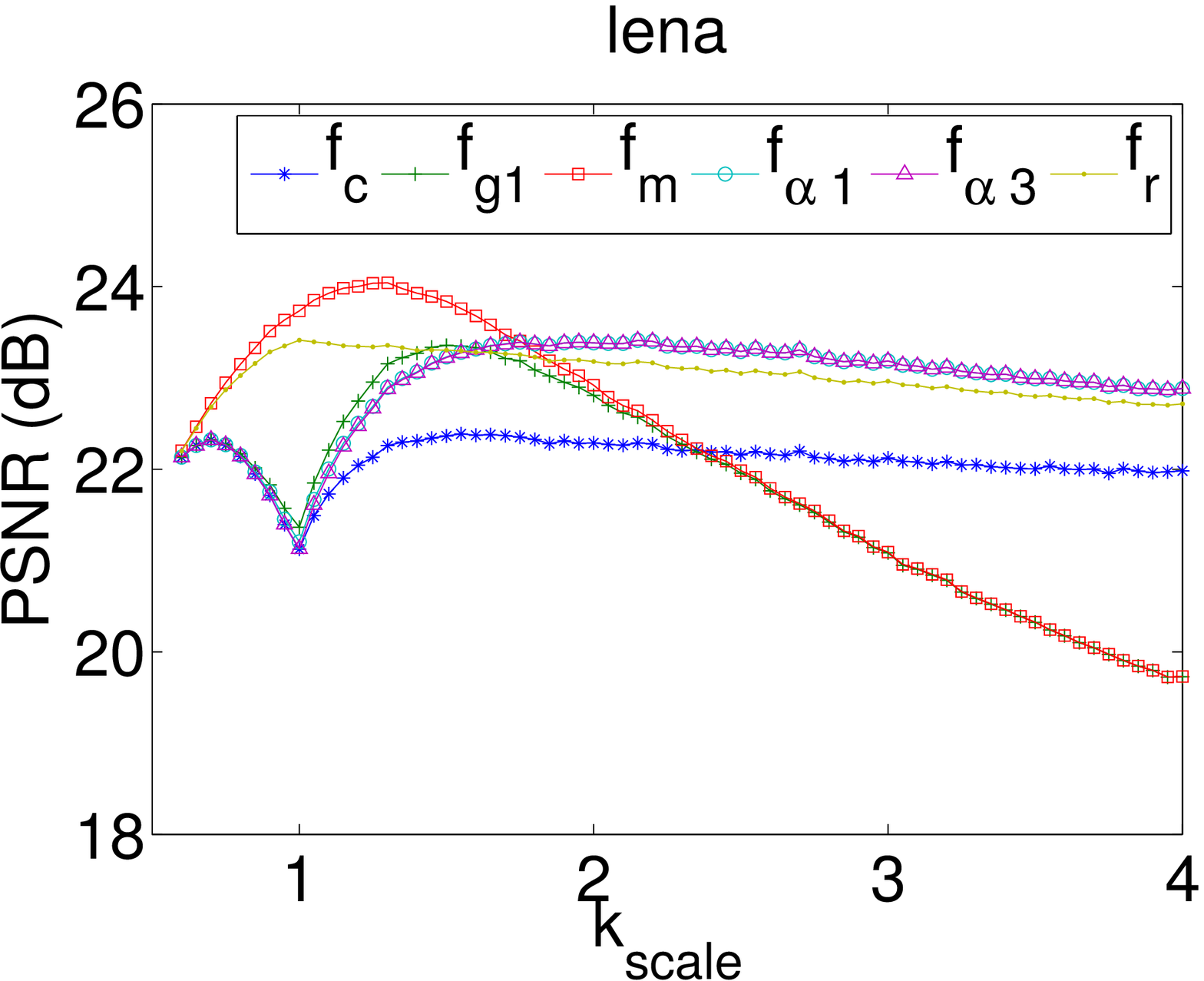}
               \caption{$MaxIter=2$}
       \end{subfigure}
\caption{Reconstructed image qualities}
\label{fig:e6}
\end{figure}

\section{Conclusion}
Signal reconstruction problems appear in many application areas, under various names. In image and video processing, a signal may include sets of images, video sequences, depth and spectral maps, their patches, as well as image-related feature vectors. 
Common image and video processing tasks, such as super-resolution, upscaling, magnification, in-painting, depth recovery, increasing image dynamic range, adding video frames for faster refresh rate, etc., can be posed as signal reconstruction problems. 
Even some seemingly unrelated tasks can be framed as signal reconstruction problems, e.g.,\ classification, or object tracking and motion prediction. 
In audio processing, a signal may include audio sequences, audio spectral maps, and audio feature vectors. Reconstruction can be used, e.g.,\ for upsampling, increasing audio frequency or dynamic ranges, adding synthetic audio channels, depth reconstruction, and audio restoration, including real-time removal of impulse noise. 
In data mining applications, signal reconstruction appears in a form of data completion or interpolation, estimating missing data and predicting future data, e.g.,\ time series data reconstruction can be used to deal with faulty sensors, and data extrapolation can help to predict future system failures.

Our efficient iterative reconstruction algorithms allow reconstructing signals with desired properties given by a guiding subspace. 
Numerical examples for noise-free and noisy image magnification demonstrate the advantages of our technology. 
Although our tests in this paper are limited to one specific example of 
signal reconstruction in imaging, the proposed methodology is general and
expected to be effective for a wide range of applications, in video and sound processing, data mining, real time security, and artificial intelligence systems.  

%\section*{Acknowledgment}

%%%%%%%%%%%%%%%%%%%%%%%%%%%%%%%%%%%%%%%%%%%%%%%%%%%%%%%%%%%

\appendices
\section{Basics of angles between subspaces}
%\appendix[Angles between subspaces]
\label{app:Angles}
\begin{definition}
The minimum gap between two closed subspaces $\Fc$ and $\Gc$ is defined as 
\begin{equation*}
\gamma(\Fc,\Gc) = \inf_{\fv \in \Fc, \fv \notin \Gc} \frac{\text{dist}\left(\fv,\Gc\right)}{\text{dist}\left(\fv, \Fc \cap \Gc\right)}.
\end{equation*}
\end{definition}
% something about what it means. 
%
\begin{definition}
Let $\Fc$ and $\Gc$ be two closed subspaces of $\Hc$ with projectors $\Tm_\Fc$ and $\Tm_\Gc$ respectively. Let $\Sigma\left(\left(\Tm_\Fc\Tm_\Gc\right)|_\Fc\right)$ denote the spectrum of $\left(\Tm_\Fc\Tm_\Gc\right)|_\Fc$. Then,
\begin{equation*}
\hat{\Theta}\left(\Fc,\Gc\right) = \{\theta: \theta = \cos^{-1}\sigma, \sigma \geq 0, \sigma^2 \in \Sigma\left(\left(\Tm_\Fc\Tm_\Gc\right)|_\Fc\right)\}
\end{equation*}
is called the set of angles from subspace $\Fc$ to the subspace $\Gc$. Angles $\Theta\left(\Fc,\Gc\right) = \hat{\Theta}\left(\Fc,\Gc\right) \cap \hat{\Theta}\left(\Gc,\Fc\right)$ are called the angles between the subspaces $\Fc$ and $\Gc$.
\label{def:angles}
\end{definition}
The minimum gap can also be expressed in terms of the angles between the subspaces as~in~\cite[Theorem~2.15]{Knyazev-JFA-10}:
\begin{equation}
\gamma(\Fc,\Gc) = \sin\left(\inf \{\Theta(\Fc,\Gc) \setminus \{0\}\}\right).
\label{eq:gap-angle}
\end{equation}

Principal angles between two subspaces in $\mathbb{R}^n$ can be defined more simply as follows.
\begin{definition}
Let $\Fc$ and $\Gc$ be two subspaces in $\mathbb{R}^n$ with dimensions $q$ and $p$ respectively. Let $q \leq p$. Then the principal angles $\theta_1, \ldots, \theta_q \in [0, \pi/2]$ between $\Fc$ and $\Gc$ are defined recursively for $i = 1,\ldots, q$ by
\begin{equation}
\cos \theta_i = \max_{\uv \in \Fc, \vv \in \Gc} \frac{\Braket{\uv,\vv}}{\|\uv\|\|\vv\|} 
\end{equation} 
subject to $\uv \perp \uv_j, \vv\perp \vv_j$ for $j = 1, \ldots, i-1$.
\end{definition}
If columns of two matrices $\Fm$ and $\Gm$ span $\Fc$ and $\Gc$, then cosines of the principal angles are also called canonical correlations between $\Fm$ and $\Gm$. %~\cite{Bjorck-MOC-73}.
Let $\Tm_{\Fc}$ and $\Tm_{\Gc}$ be the projectors for $\Fc$ and $\Gc$ respectively; then the eigenvalues $\{\sigma^2\}$ of $\Tm_\Fc \Tm_\Gc|_{\Fc}$ are related to  the angles $\{\theta\}$ by~\cite{Knyazev-JFA-10} 
\begin{equation}
\sigma = \cos\theta.
\end{equation} 
The condition of positiveness of the infimum of non-zero angles between $\Sc$ and $\Tc^\perp$ is evidently always satisfied in finite dimensional spaces, although it may 
approach zero as the dimension increases. 
However, in infinite dimensional spaces a sequence of non-zero angles may converge to zero, leading to the zero infimum. 

Relationship between angles $\Theta(\Fc,\Gc)$, $\Theta\left(\Fc,\Gc^\perp\right)$, and $\Theta\left(\Fc^\perp,\Gc^\perp\right)$ is given in~\cite[Theorem~2.7]{Knyazev-JFA-10}.
\begin{figure*}[!t]
\normalsize
\setcounter{MYtempeqncnt}{\value{equation}}
% Set the equation number to one less than the one desired for the first equation here. The value here will have to changed if equations are added or removed prior to the place these equations are referenced in the main text.
\setcounter{equation}{14}
\begin{align}
\Theta(\Fc,\Gc) \setminus \left(\{0\} \cup \{\pi/2\}\right) &= \{ \pi/2-\Theta(\Fc,\Gc^\perp)\} \setminus \left(\{0\} \cup \{\pi/2\}\right)\nonumber \\ 
\Theta(\Fc,\Gc) \setminus \{0\} &= \Theta(\Fc^\perp,\Gc^\perp) \setminus \{0\}
\label{eq:angle-relations}
\end{align}
\hrulefill
%\vspace*{4pt}
\end{figure*}
%
%\small 
%\begin{align}
%\Theta(\Fc,\Gc) \setminus \left(\{0\} \cup \{\pi/2\}\right) &= \{ \pi/2-\Theta(\Fc,\Gc^\perp)\} \setminus \left(\{0\} \cup \{\pi/2\}\right)\nonumber \\ 
%\Theta(\Fc,\Gc) \setminus \{0\} &= \Theta(\Fc^\perp,\Gc^\perp) \setminus \{0\}
%\label{eq:angle-relations}
%\end{align}
%\normalsize

\section{The conjugate gradient method introduction}
%\appendix[The conjugate gradient method]
\label{app:CG}
The conjugate gradient method is one of the most widely used methods for solving $\Km\xv = \bv$ when $\Km$ is a linear, bounded, self-adjoint, non-negative operator. It is easy to see that solving $\Km \xv = \bv$ is equivalent to
\begin{equation*}
\min_{\xv} E(\xv) = \frac{1}{2}\Braket{\xv,\Km\xv} - \Braket{\bv,\xv}.
\end{equation*}
CG is the optimal method for solving the above problem among all polynomial iterative methods which involve multiplication of a vector by $\Km$ as the main step in each iteration. To put it more formally, let us first define a plane
\begin{equation}
\bar{\Kc}_m = \xv_0 + \text{span}\{\bv - \Km \xv_0, \ldots, \Km^{m-1}\left(\bv - \Km \xv_0\right)\},
\label{eq:krylov_hp}
\end{equation}
where $\xv_0$ is the initial guess for the solution. When $\xv_0 =$~$ \zerov$, $\bar{\Kc}_m$ equals the the Krylov subspace of order $m$ which is defined as 
\begin{equation*}
\Kc_m = \text{span}\{\bv, \Km\bv, \ldots, \Km^{m-1}\bv\}.
\end{equation*}
% 
% REVISE. BEST SOLUTION IN HYPERPLANE = KRYLOV SUBSPACE + INITIAL VAL
The solution $\xv_m$ at $m$-th iteration of CG satisfies
%~\cite{Luenberger-69, Shewchuk-94} %references not for Hilbert spaces!
\begin{equation*}
\xv_m = \argmin_{\xv \in \bar{\Kc}_m} E(\xv) 
= \argmin_{\xv \in \bar{\Kc}_m} \|\xv - \xv^*\|_{\Km}
\end{equation*}  
where, $\|\zv\|_{\Km} = \Braket{\zv,\Km\zv}$ denotes the induced $\Km$-norm  and $\xv^*$ denotes the actual solution of $\Km\xv = \bv$. This shows that CG gives the best possible solution after $m$ iterations and thus, is the most efficient iterative method. 
% The steps involved in the basic CG method are summarized in Algorithm~\ref{alg:CG}. 

%%%%%%%%
\section{Simple matrix examples}
\label{sec:examples}
To clarify, illustrate, and verify our somewhat abstract arguments in Hilbert spaces, in this section we present several 
matrix examples of increasing complexity, in 2D, 3D, and, finally, the most representative case of 4D subspaces in 8D space, where 
all important subspaces used in the paper are non-trivial, while, at the same time, all the important quantities are explicitly analytically derived.
We start with 2D and 3D cases, because they can also be illustrated geometrically, intuitively appealing. 

\subsection{2D case}
First, we consider $\Hc$ as the 2D plane such that $\Hc = \text{span}(\ev_1, \ev_2)$, where $\ev_1 = [1, 0]^T$ and $\ev_2 = [0, 1]^T$ are the standard basis vectors. Let the sampling subspace $\Sc = \text{span}(\ev_1)$ and the guiding subspace $\Tc = \text{span}(\ev_1 + a\ev_2)$ for some real scalar $a$. Without loss of generality, assume the signal $\fv = [2, 3]^T$ and $a = 2$. Consequently, the sampled signal $\Sm\fv = [2, 0]^T$ and the sample consistent space $\Sm\fv + \Sc^\perp = \text{span}([2, 0]^T)$. Fig.~\ref{fig:2D_illust} illustrates the 2D example showing the subspaces $\Sc$ and $\Tc$ as well as the signal $\fv$, it's sampling $\Sm\fv$, and the reconstruction $\hat{\fv}$. Here, the space $\Hc_0 = \Hc/\{\zerov \} =\Hc$. Let us also notice that $\Tc$ and $\Sm\fv + \Sc^\perp$ intersect at the unique reconstruction point $\hat{\fv}$, in this example.
\begin{figure}
\centering
\includegraphics[width = 3in]{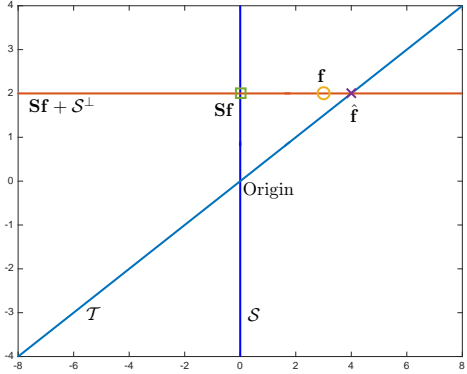}
\caption{2D example with a unique reconstruction point.}
\label{fig:2D_illust}
\end{figure}

For general $a \in \mathbb{R}$, we have
\[
	\Sm\fv = \left[
	\begin{array}{c}
	f_1 \\ 0 
	\end{array}
	\right], \	
	\Tm^\perp \Sm\fv = \left[
	\begin{array}{c}
	f_1 \sin^2 \theta_a\\ -f_1 \sin\theta_a\cos\theta_a 
	\end{array}
	\right].
\]
The subspace $\Hc_0=\text{span}(\ev_1,\ev_2)/ \{\zerov \}$. Therefore, restricting $\Km = (\Sm^\perp\Tm^{\perp})|_{\Sm^\perp}$ to $\Hc_0\cap \Sc^\perp$ reduces the 2-by-2 matrix form of the operator $\Km$ to the scalar form $\cos^2 \theta_a$ of the operator $\Km_\star=\left.\Km\right|_{\Hc_0\cap \Sc^\perp}$. The extension operator $\Km_0$ has the form
\[
	\left[
	\begin{array}{cc}
	0 & 0  \\
	0 & \cos^2 \theta_a 
	\end{array}
	\right]
\]
Consequently, the sample consistent reconstruction results in 
\[
	\hat{\xv}_n = \left[
	\begin{array}{c}
	0 \\ a f_1 
	\end{array}
	\right], \	
	\hat{\fv} = \left[
	\begin{array}{c}
	f_1 \\ a f_1 
	\end{array}
	\right].
\]

Next, we show that the norms of operators $\Km_0^\dagger \Sm_0^\perp \Tm_0^{\perp}$ and $\Km_0^\dagger \Sm_0^{\perp} \Tm_0^{\perp} \Sm_0$ are in fact $1/\cos\theta_{\max}$ and $\tan \theta_{\max}$, respectively. The matrix form of $\Tm_0^\perp$ is 
\[
	\left[
	\begin{array}{cc}
	\sin^2\theta_a & -\sin\theta_a\cos\theta_a \\
	-\sin \theta_a \cos \theta_a & \cos^2 \theta_a  
	\end{array}
	\right].
\]
We have a matrix form of $\Sm_0^\perp \Tm_0^{\perp}$ as
\[
	\left[
	\begin{array}{cc}
	0 & 0 \\
	-\sin \theta_a \cos \theta_a & \cos^2 \theta_a  
	\end{array}
	\right],
\]
and the matrix form of $\Sm_0^\perp \Tm_0^{\perp}\Sm_0$ as
\[
	\left[
	\begin{array}{cc}
	0 & 0 \\
	-\sin \theta_a \cos \theta_a & 0 
	\end{array}
	\right],
\]
with the corresponding singular values $1/\cos\theta_a$ and $\sin\theta_a \cos\theta_a$, respectively. Therefore, the operator $\left(\Sm_0^\perp \Tm_0^{\perp}\right)^\dagger \Sm_0^\perp \Tm_0^{\perp}$ has the form
\[
\left[
	\begin{array}{cc}
	\sin^2\theta_a & -\sin\theta_a\cos\theta_a \\
	-\sin \theta_a \cos \theta_a & \cos^2 \theta_a  
	\end{array}
	\right],
\]
which is the same form as $\Tm_0^\perp$. Then, $\Km_0^\dagger \Sm_0^\perp \Tm_0^{\perp}$ has the matrix form of
\[
	\left[
	\begin{array}{cc}
	0 & 0 \\
	-\tan \theta_a & 1 
	\end{array}
	\right],
\]
whose nonzero singular value is $1/\cos\theta_{a}$. On the other hand, $\Km_0^\dagger \Sm_0^\perp \Tm_0^{\perp}\Sm_0$ has the matrix form
\[ 
	\left[
	\begin{array}{cccc}
	0 & 0 \\
	-\tan \theta_a & 0 
	\end{array}
	\right],
\]
with singular value $\tan \theta_a$.
%%%
\subsection{2D subspaces in 3D space}
Next, we consider the 3D space $\Hc = \text{span}(\ev_1, \ev_2, \ev_3)$ with the sampling plane  $\Sc = \text{span}(\ev_1, \ev_2)$ and the guiding subspace $\Tc = \text{span}(\ev_1 + a\ev_3)$ for $a = 2$. Denote by $\theta_a$ be the angle between the subspaces $\Sc$ and $\Tc^\perp$, then 
\[\cos \theta_a = \frac{1}{\sqrt{1 + a^2}},\, \sin \theta_a = \frac{a}{\sqrt{1+a^2}},\, \tan \theta_a = a.\] The projection operators $\Sm$ and $\Tm$ are given by 
$$
	\Sm = \left[\begin{array}{ccc}
	1 & 0 & 0\\
	0 & 1 & 0\\
	0 & 0 & 0
	\end{array}\right], \quad
	\Tm = \left[\begin{array}{ccc}
	\frac{1}{1+a^2} & 0 & \frac{a}{1+a^2}\\
	0 & 0 & 0\\
	\frac{a}{1+a^2} & 0 & \frac{a^2}{1+a^2}
	\end{array}\right]. 
$$
A signal $\fv = [2, 1, 6]^T$ then results in the sampling $\Sm\fv = [2, 1, 0]$. Since the reconstruction is restricted to the subspace $\Sc^\perp = \text{span}(\ev_3)$, in this example 
$$
\Km = (\Sm^\perp\Tm^{\perp})|_{\Sm^\perp} = \left[\begin{array}{ccc}
	0 & 0 & 0\\
	0 & 0 & 0\\
	0 & 0 & \frac{1}{1+a^2} 
	\end{array}\right].
$$
Moreover, the subspace $\Sc\cap\Tc^\perp = \text{span}(\ev_2)$ is nontrivial. Therefore, the reconstruction subspace $\Hc_0 = \Hc/\{\text{span}(\ev_2)\}$ and $\Km_\star = \Km|_{\Hc_0\cap \Sc^\perp} = {1}/({1+a^2})$. 

Fig.~\ref{fig:3D_illust} illustrates the geometry of the subspaces. Notice that in this example, the guiding subspace $\Tc$ does not intersect the sample consistent space $\Sm\fv + \Sc^\perp$. Therefore, a reconstruction interval exists between the sample consistent reconstruction $\hat{\fv}_c = [2,1,4]^T$ and the generalized reconstruction $\hat{\fv}_g = \Tm\hat{\fv}_c = [2, 0, 4]^T$. The proposed reconstruction $\hat{\fv}_{\alpha}$ can exist anywhere on the reconstruction interval and is parametrized by $\alpha \in [0, 1]$. Here we plot the reconstructed signal $\hat{\fv}_{\alpha} = [2, 0.7, 4]^T$ corresponding to $\alpha = 0.7$.

\begin{figure}
\centering
\includegraphics[width = 3in]{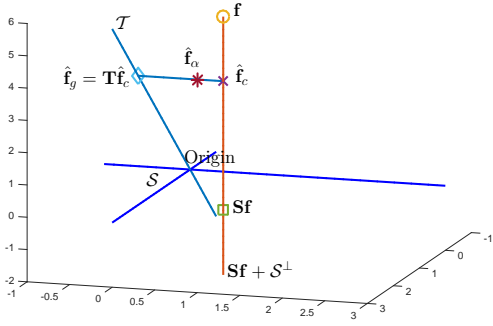}
\caption{3D example showcasing the reconstruction interval between the sample consistent reconstruction $\hat{\fv}_c$ and the generalized reconstruction $\hat{\fv}_g = \Tm\hat{\fv}_c$. The proposed reconstruction $\hat{\fv}_{\alpha}$ exists anywhere on the reconstruction interval.}
\label{fig:3D_illust}
\end{figure}

\subsection{4D subspaces in 8D space}\label{sec:moreexamples}

Finally, we illustrate a example in the eight dimensional space $\Hc = \text{span}(\ev_1, \ev_2,\dots, \ev_8)$. Consider the symbolic signal $\fv = [f_1, f_2, f_3, f_4, f_5, f_6, f_7, f_8]^T$ with the sampling space $\Sc = \text{span}(\ev_1, \ev_3, \ev_5, \ev_6)$ and target space 
\[\Tc = \text{span}(\ev_1 + a\ev_2, \ev_3 + b\ev_4, \ev_5, \ev_7),\, a\geq b>0.\] 
There are four principle angles $\pi/2>\theta_a\geq\theta_b>0$ between the subspaces $\Sc$ and $\Tc$, where 
\[\cos \theta_a = \frac{1}{\sqrt{1 + a^2}},\, \sin \theta_a = \frac{a}{\sqrt{1+a^2}},\, \text{and } \tan \theta_a = a,\] 
\[\cos \theta_b = \frac{1}{\sqrt{1 + b^2}},\, \sin \theta_b = \frac{b}{\sqrt{1+b^2}},\, \text{and } \tan \theta_b = b.\] 
Consequently, the following identities hold
\[
	\Sm\fv = \left[
	\begin{array}{c}
	f_1 \\ 0 \\ f_3 \\ 0 \\ f_5 \\ f_6 \\ 0 \\ 0
	\end{array}
	\right], \	
	\Tm^\perp \Sm\fv = \left[
	\begin{array}{c}
	f_1 \sin^2 \theta_a\\ -f_1 \sin\theta_a\cos\theta_a \\ f_3 \sin^2 \theta_b\\ -f_3 \sin\theta_b\cos\theta_b \\ 0 \\ f_6 \\ 0 \\ 0
	\end{array}
	\right].
\]
The subspace $\Hc_0=\text{span}(\ev_1,\ev_2,\ev_3,\ev_4)$ excludes the intersections $\Sc \cap \Tc = \text{span}(\ev_5)$, $\Sc \cap \Tc^\perp = \text{span}(\ev_6)$, $\Sc^\perp \cap \Tc = \text{span}(\ev_7)$, and $\Sc^\perp \cap \Tc^\perp = \text{span}(\ev_8)$. Therefore, restricting $\Km = (\Sm^\perp\Tm^{\perp})|_{\Sm^\perp}$ to $\Hc_0\cap \Sc^\perp$ reduces the 4-by-4 matrix form of the operator $\Km$ to the following 2-by-2 matrix form
\[
	\left[
	\begin{array}{cc}
	\cos^2 \theta_a & 0 \\
	0 & \cos^2 \theta_b 
	\end{array}
	\right]
\]
of the operator $\Km_\star=\left.\Km\right|_{\Hc_0\cap \Sc^\perp}$. The extension operator $\Km_0$ has the form
\[
	\left[
	\begin{array}{cccc}
	0 & 0 & 0 & 0 \\
	0 & \cos^2 \theta_a & 0 & 0 \\
	0 & 0 & 0 & 0 \\
	0 & 0 & 0 & \cos^2 \theta_b 
	\end{array}
	\right]
\]
Consequently, the sample consistent reconstruction results in 
\[
	\hat{\xv}_n = \left[
	\begin{array}{c}
	0 \\ a f_1 \\ 0 \\ b f_3 \\ 0 \\ 0 \\ 0 \\ 0 \\ 0
	\end{array}
	\right], \	
	\hat{\fv} = \left[
	\begin{array}{c}
	f_1 \\ a f_1 \\ f_3 \\ b f_3 \\ f_5 \\ f_6 \\ 0 \\ 0
	\end{array}
	\right].
\]

Next, we show that the norms of operators $\Km_0^\dagger \Sm_0^\perp \Tm_0^{\perp}$ and $\Km_0^\dagger \Sm_0^{\perp} \Tm_0^{\perp} \Sm_0$ are in fact $1/\cos\theta_{\max}$ and $\tan \theta_{\max}$, respectively. We have a matrix form of $\Sm_0^\perp \Tm_0^{\perp}$ as
\[
	\left[
	\begin{array}{cccc}
	0 & 0 & 0 & 0 \\
	-\sin \theta_a \cos \theta_a & \cos^2 \theta_a & 0 & 0 \\
	0 & 0 & 0 & 0 \\
	0 & 0 & -\sin \theta_b \cos \theta_b & \cos^2 \theta_b 
	\end{array}
	\right],
\]
and the matrix form of $\Sm_0^\perp \Tm_0^{\perp}\Sm_0$ as
\[
	\left[
	\begin{array}{cccc}
	0 & 0 & 0 & 0 \\
	-\sin \theta_a \cos \theta_a & 0 & 0 & 0 \\
	0 & 0 & 0 & 0 \\
	0 & 0 & -\sin \theta_b \cos \theta_b & 0 
	\end{array}
	\right],
\]
with the corresponding singular values $(1/\cos\theta_a, 1/\cos\theta_b)$ and $(\sin\theta_a \cos\theta_a, \sin\theta_b \cos\theta_b)$, respectively.
Then, $\Km_0^\dagger \Sm_0^\perp \Tm_0^{\perp}$ has the matrix form of
\[
	\left[
	\begin{array}{cccc}
	0 & 0 & 0 & 0 \\
	-\tan \theta_a & 1 & 0 & 0 \\
	0 & 0 & 0 & 0 \\
	0 & 0 & -\tan\theta_b & 1 
	\end{array}
	\right],
\]
whose nonzero singular values are $1/\cos\theta_{a}$ and $1/\cos\theta_b$. On the other hand, $\Km_0^\dagger \Sm_0^\perp \Tm_0^{\perp}\Sm_0$ has the matrix form
\[ 
	\left[
	\begin{array}{cccc}
	0 & 0 & 0 & 0 \\
	-\tan \theta_a & 0 & 0 & 0 \\
	0 & 0 & 0 & 0 \\
	0 & 0 & -\tan\theta_b & 0 
	\end{array}
	\right],
\]
with singular values $\tan \theta_a$ and $\tan \theta_b$. Thus, on the one hand,
\[
\begin{array}{lll}	
 \|\hat{\xv}_n\|^2&=&a^2|f_1|^2+b^2|f_3|^2\\
 &\leq& (|f_1|^2+|f_3|^2)a^2\\
 &=&\|\Sm \Pm_0\fv\|^2\tan^2\theta_{\max},
 \end{array}
\]
since $\theta_{\max}=\theta_{a}$.
On the other hand, 
\[
\begin{array}{lll}	
 \|\hat{\xv}_n\|^2&=&|f_1|^2\tan^2\theta_{a}+|f_3|^2\tan^2\theta_{b}\\
 &\leq&{\|\Tm^\perp \Sm \Pm_0 \fv\|^2}/{\cos^2\theta_{\max}},
 \end{array}
\]
since $\|\Tm^\perp \Sm \Pm_0 \fv\|^2=|f_1|^2\sin^2\theta_{a}+|f_3|^2\sin^2\theta_{b}$.
Finally,
\[
\begin{array}{lll}	
 \|\hat{\xv}_n\|^2&=&|f_1|^2\tan^2\theta_{a}+|f_3|^2\tan^2\theta_{b}\\
 &\leq&{\|\Sm^\perp\Tm^\perp \Sm \Pm_0 \fv\|^2}/{\cos^4\theta_{\max}},
 \end{array}
\]
since
\[
 \|\Sm^\perp\Tm^\perp \Sm \Pm_0 \fv\|^2 = |f_1|^2\sin^2\theta_{a}\cos^2\theta_{a}+|f_3|^2\sin^2\theta_{b}\cos^2\theta_{b}.
\]
The three inequalities above illustrate all three bounds proved in Theorem~\ref{thm:Lidentity}. Moreover, the bounds on the reconstruction error in Theorem~\ref{thm:errbnds} are equal to 
\[
\begin{array}{ll}
	\|\fv - \hat{\fv}\|^2 &= (f_2 - \tan\theta_a f_1)^2 + (f_4 - \tan\theta_b f_3)^2 + f_7^2 + f_8^2 \\
				&\leq \|\Tm^\perp\Pm_0\fv\|^2/\cos^2\theta_{\max}  + f_7^2 + f_8^2\\
	\textrm{and}	& \leq \|\Sm^{\perp}\Tm^\perp\Pm_0\fv\|^2/\cos^4\theta_{\max} + f_7^2 + f_8^2,
\end{array}
\]
where $\left\|\Pm_{\Sc^\perp\cap\Tc}\fv\right\|^2 = f_7^2$, $\left\|\Pm_{\Sc^\perp\cap\Tc^\perp}\fv\right\|^2 = f_8^2$, 
$$
\begin{array}{lll}
\|\Tm^\perp\Pm_0\fv\|^2 &= &(f_2\sin\theta_a\cos\theta_a - f_1\sin^2\theta_a)^2 \\
					&& + (f_2\cos^2\theta_a - f_1\sin\theta_a\cos\theta_a)^2 \\
					&& + (f_4\sin\theta_b\cos\theta_b - f_3\sin^2\theta_b)^2 \\
					&& + (f_4\cos^2\theta_b - f_3\sin\theta_b\cos\theta_b)^2\\
				    &= & \cos^2\theta_a(f_2 - f_1\tan\theta_a)^2 \\
				    && + \cos^2\theta_b(f_4 - f_3\tan\theta_b)^2,
\end{array}				
$$ 
and 
$$
\begin{array}{lll}
\|\Sm^\perp\Tm^\perp\Pm_0\fv\|^2 &= &(f_2\cos^2\theta_a - f_1\sin\theta_a\cos\theta_a)^2 \\
					&& + (f_4\cos^2\theta_b - f_3\sin\theta_b\cos\theta_b)^2\\
				    &= & \cos^4\theta_a(f_2 - f_1\tan\theta_a)^2 \\
				    && + \cos^4\theta_b(f_4 - f_3\tan\theta_b)^2.
\end{array}				
$$
The above derivation shows that every bound is sharp, e.g.,\ turns into an equality if $a=b$.  

% Next, we show that the norms of operators $\Tm_0\Km_0^\dagger$ and $\Tm_0\Km_0^\dagger \Sm_0^{\perp} \Tm_0^{\perp}$ used in Theorem~\ref{thm:fgerrbnds} are in fact $\tan\theta_{\max}/\cos\theta_{\max}$ and $\tan \theta_{\max}$, respectively. The operator $\Tm_0\Km_0^\dagger $ has the matrix form of
% \[
% 	\left[
% 	\begin{array}{cccc}
% 	0 & \tan\theta_a & 0 & 0 \\
% 	0 & \tan^2\theta_a & 0 & 0 \\
% 	0 & 0 & \tan\theta_b & 0 \\
% 	0 & 0 & 0 & \tan^2\theta_b 
% 	\end{array}
% 	\right],
% \]
% whose nonzero singular values are $\tan\theta_a/\cos\theta_{a}$ and $\tan\theta_b/\cos\theta_b$. On the other hand, the operator $\Tm_0\Km_0^\dagger \Sm_0^\perp \Tm_0^{\perp}$ has the matrix form of
% \[{%\scriptsize %invalid in math mode
% 	\left[
% 	\begin{array}{cccc}
% 	-\sin^2\theta_a & \sin\theta_a\cos\theta_a & 0 & 0 \\
% 	-\sin^2\theta_a\tan \theta_a & \sin^2\theta_a & 0 & 0 \\
% 	0 & 0 & -\sin^2\theta_b & \sin\theta_b\cos\theta_b \\
% 	0 & 0 & -\sin^2\theta_b\tan \theta_b & \sin^2\theta_b 
% 	\end{array}
% 	\right],}
% \]
% whose nonzero singular values are $\tan\theta_{a}$ and $\tan\theta_b$.

\bibliographystyle{IEEEtran}
\bibliography{refs}

\begin{IEEEbiography}[{\includegraphics[width=1.1in,clip,keepaspectratio]{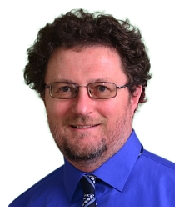}}]{Andrew Knyazev} graduated from the Faculty of Computational Mathematics and Cybernetics of Moscow State University in 1981 and received the Ph.D. degree in Numerical Mathematics at the Russian Academy of Sciences, Moscow, Russia in 1985. 

He is Distinguished Research Scientist at the Mitsubishi Electric Research Laboratories (MERL) and Professor Emeritus at the University of Colorado Denver. He is a Fellow of the Society for Industrial and Applied Mathematics (SIAM) and Senior Member of the IEEE. During his 30 years in academia, he has contributed to numerical analysis of partial differential equations and computational linear algebra, with emphasis on eigenvalue problems, supported by NSF and DOE awards, and graduated 7 Ph.D. students. 

Since 2012, his research interests at MERL are in algorithms for image and video processing, data sciences, optimal control, material sciences, and numerical simulation of complex phenomena. He has over $100$ publications, over a dozen of patent applications, and several U.S. and international patents.
\end{IEEEbiography}

\begin{IEEEbiography}[{\includegraphics[width=1.1in,clip,keepaspectratio]{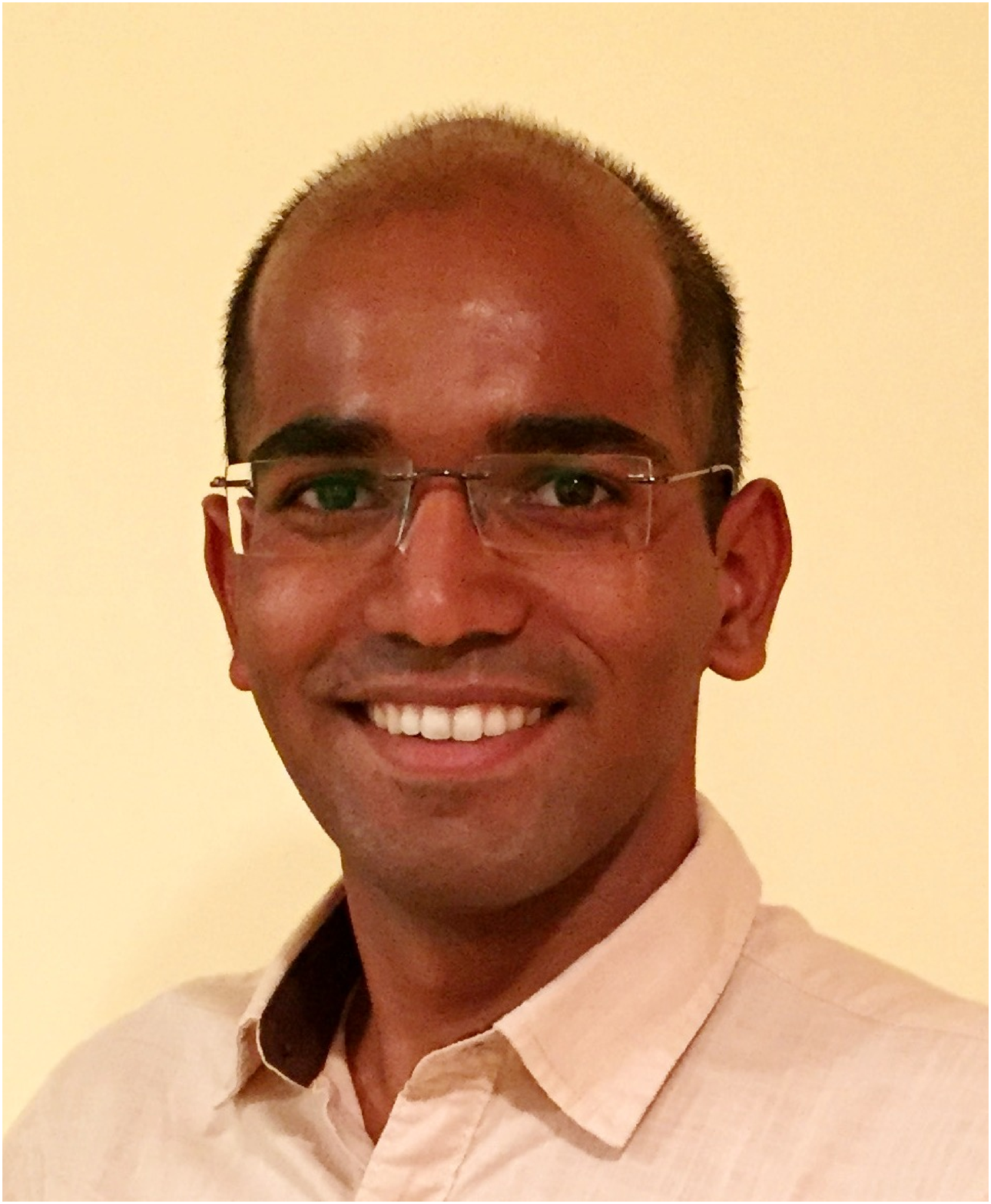}}]{Akshay Gadde} (S'13) received his Bachelor of Technology degree in electrical engineering from Indian
Institute of Technology (IIT), Kharagpur, India, in 2011. He has been working towards a Ph.D. in electrical engineering at the University of Southern California (USC), Los Angeles, since 2011, supported by the Provost's Fellowship. He~is a recipient of the Best Student Paper Award at ICASSP 2014. 
His research interests include graph signal processing and machine learning with applications to
multimedia data processing and compression.
\end{IEEEbiography}

% insert where needed to balance the two columns on the last page with biographies
%\newpage

\begin{IEEEbiography}[{\includegraphics[width=1.1in,clip,keepaspectratio]{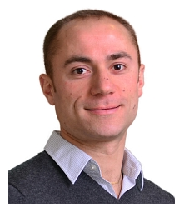}}]{Hassan Mansour} (S'99, M'09) received his Bachelor of Engineering (2003) from the American University of Beirut, and his M.A.Sc. (2005) and Ph.D. (2009) degrees from the Department of Electrical and Computer, University of British Columbia (UBC), Vancouver, Canada.

He is a currently a Principal Research Scientist in the Multimedia Group at Mitsubishi Electric Research Laboratories, Cambridge, MA. Prior to joining MERL, he pursued a postdoctoral fellowship in the Departments of Mathematics, Computer Science, and Earth and Ocean Sciences at UBC. During his graduate studies, he conducted research on scalable video coding and transmission. His research has since focused on theoretical and algorithmic aspects of compressed sensing, image and video analytics, remote sensing and array signal processing. 
\end{IEEEbiography}

\begin{IEEEbiography}[{\includegraphics[width=1.1in,clip,keepaspectratio]{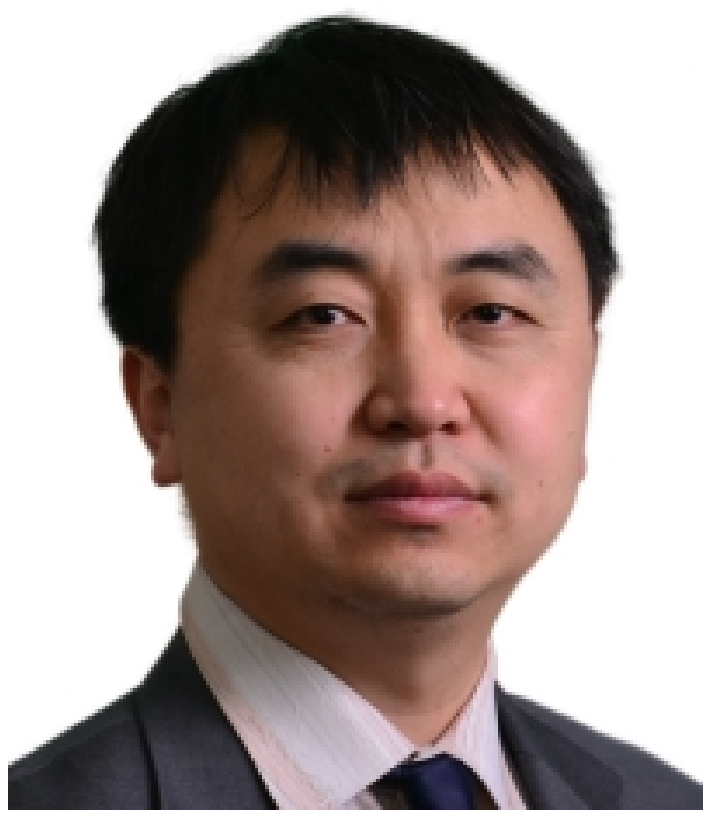}}]{Dong Tian} received the Ph.D. degree at Beijing University of Technology in 2001, and the M.Eng. and B.Eng. degrees on automation from the University of Science and Technology of China (USTC) in 1998 and 1995, respectively. 

He is Senior Principal Member Research Staff in the Multimedia Group of Mitsubishi Electric Research Laboratories (MERL) at Cambridge, MA. Prior to joining MERL, he has worked with Thomson Corporate Research at Princeton, NJ for over 4 years, where he was devoted to H.264/MPEG AVC encoder optimization and 3D video coding/processing, especially to the standards of Multiview Video Coding (MVC) and later on 3D Video Coding (3DV) within MPEG. From Jan. 2002 to Dec. 2005, he has been a postdoc at Tampere University of Technology in Finland for a Nokia funded project and made contributions on video coding standards and applications for mobile environments. His current research interests include graph signal processing, point cloud processing, machine learning, image/video coding and processing. Besides academic publications, he has over 20 US-granted patents. He is a senior member of IEEE.
\end{IEEEbiography}


%\vfill


\end{document}